\documentclass[journal,twoside,web]{ieeecolor}
\usepackage{generic}
\usepackage{cite}
\usepackage{amsmath,amssymb,amsfonts}
\usepackage{algorithmic}
\usepackage{relsize}
\usepackage{graphicx}
\usepackage{multirow}
\newtheorem{definition}{Definition}
\newtheorem{theorem}{Theorem}
\newtheorem{lemma}{Lemma}
\newtheorem{assumption}{Assumption}
\newtheorem{corollary}{Corollary}
\newtheorem{remark}{Remark}
\usepackage{steinmetz}
 \usepackage{cancel}
 \usepackage{subcaption}
\usepackage{tikz,pgfplots} 
\usetikzlibrary{calc,patterns,arrows,shapes.arrows,intersections}
\usepackage{textcomp}
\def\BibTeX{{\rm B\kern-.05em{\sc i\kern-.025em b}\kern-.08em
    T\kern-.1667em\lower.7ex\hbox{E}\kern-.125emX}}
\begin{document}
\title{Closed-Loop Frequency Analysis of Reset Control Systems}
\author{Ali Ahmadi Dastjerdi, \IEEEmembership{Member, IEEE}, Alessandro~Astolfi, \IEEEmembership{Fellow, IEEE}, Niranjan~Saikumar, \IEEEmembership{Member, IEEE}, Nima~Karbasizadeh, \IEEEmembership{Member, IEEE}, Duarte~Val\'erio, \IEEEmembership{Senior Member, IEEE} and S.~Hassan~HosseinNia,~\IEEEmembership{Senior Member, IEEE}
\thanks{This paper is submitted for review on 24-11-2020. This work has been partially supported by NWO through OTP TTW project $\#$16335, by the EACEA, by the European Union's Horizon 2020 Research and Innovation Programme under grant agreement No 739551 (KIOS CoE), and by the Italian Ministry for Research in the framework of the 2017 Program for Research Projects of National Interest (PRIN), Grant no. 2017YKXYXJ.}
\thanks{A. Ahmadi Dastjerdi, N. Saikumar, N. Karbasizadeh and S.H.~HosseinNia are with Department of Precision and Microsystems Engineering, Delft University of Technology, Delft, The Netherlands, (e-mail: A.AhmadiDastjerdi@tudelft.nl, N.Saikumar@tudelft.nl, N.KarbasizadehEsfahani@tudelft.nl, S.H.HosseinNiaKani@tudelft.nl). }
\thanks{A. Astolfi is with the Department of Electrical and Electronic Engineering, Imperial College London, London, SW7 2AZ, UK and with the Dipartimento di Ingegneria Civile e Ingegneria Informatica, Universita di Roma ``Tor Vergata", Rome, 00133, Italy (e-mail: a.astolfi@imperial.ac.uk).}
\thanks{D. Val\'erio is with faculty of IDMEC, Instituto Superior T\'ecnico, Universidade de Lisboa, Lisboa, Portugal (duarte.valerio@tecnico.ulisboa.pt).}}

\maketitle

\begin{abstract}
This paper introduces a closed-loop frequency analysis tool for reset control systems. To begin with sufficient conditions for the existence of the steady-state response for a closed-loop system with a reset element and driven by periodic references are provided. It is then shown that, under specific conditions, such a steady-state response for periodic inputs is periodic with the same period as the input. Furthermore, a framework to obtain the steady-state response and to define a notion of closed-loop frequency response, including high order harmonics, is presented. Finally, pseudo-sensitivities for reset control systems are defined. These simplify the analysis of this class of systems and allow a direct software implementation of the analysis tool. To show the effectiveness of the proposed analysis method the position control problem for a precision positioning stage is studied. In particular, comparison with the results achieved using methods based on the Describing Function shows that the proposed method achieves superior closed-loop performance.
\end{abstract}

%%%%%%%%%%%%%%%%%%%%%%%%%%%%%%%%%%%INTRODUCTION
\section{Introduction}\label{sec:1}
\IEEEPARstart{P}{roportional} Integral Derivative (PID) controllers are used in more than $90\%$ of industrial control applications \cite{dastjerdi2018tuning,o2009handbook,chen2006ubiquitous}. However, cutting-edge industrial applications have control requirements that cannot be fulfilled by PID controllers. To overcome this problem linear controllers may be substituted by non-linear ones. Reset controllers are one such non-linear controllers which have attracted attention due to their simple structure and their ability to improve closed-loop performance \cite{clegg1958nonlinear,beker2004fundamental,beker2001plant,5712180,villaverde2011reset,banos2011reset,van2017frequency,wu2007reset,pavlov2013steady,panni2014position,hazeleger2016second,beerens2019reset,saikumar2019resetting,saikumar2019constant,7332744,7448867,4781997}.

A traditional reset controller consists of a linear element the state of which is reset to zero when the input equals zero. The simplest reset element is the Clegg Integrator (CI), which is a linear integrator with a reset mechanism \cite{clegg1958nonlinear}. To provide design freedom and applicability, reset controllers such as First Order Reset Elements (FORE) \cite{zaccarian2005first,horowitz1975non} and Second Order Reset Elements (SORE) have also been introduced \cite {hazeleger2016second}. These reset elements are utilized to construct new compensators to achieve significant performance enhancement \cite{hunnekens2014synthesis,van2018hybrid,palanikumar2018no,chen2019development,valerio2019reset,saikumar2019constant}. In order to further improve the performance of reset control systems several techniques, such as the considerations of non-zero reset values \cite{banos2011reset,horowitz1975non}, reset bands \cite{barreiro2014reset,banos2014tuning}, fixed reset instants, and $PI+CI$ configurations \cite{vidal2008qft,nair2018grid,hosseinnia2014general} have been introduced.

Frequency-domain analysis is preferred in industry since this allows ascertaining closed-loop performance measures in an intuitive way. However, the lack of such methods for non-linear controllers is one of the reasons why non-linear controllers are not widely popular in industry. The Describing Function (DF) method is one of the few methods for approximately studying non-linear controllers in the frequency-domain and this has been widely used also in the literature of reset controllers \cite{guo2009frequency,banos2011reset,valerio2019reset,saikumar2019constant}. The DF method relies on a quasi-linear approximation of the steady-state output of a non-linear system considering only the first harmonic of the Fourier series expansion of the input and output signals (assumed periodic). The general formulation of the DF method for reset controllers is presented in \cite{guo2009frequency}, which however does not provide any information on the closed-loop steady-state response.

In this paper, first, sufficient conditions for the existence of the steady-state response for a closed-loop system with a reset element and driven by a periodic input are given. Then, a notion of closed-loop frequency response for reset control systems, including high order harmonics, is introduced. Pseudo-sensitivities to combine harmonics and facilitate analyzing reset control systems in the closed-loop configuration are then defined. All of these ideas can be utilized to develop a toolbox which is briefly discussed. Furthermore, the method is used to analyze the performance of a precision positioning stage. Note finally that, contrary to the DF method, which provides only approximations for the periodic steady-state response of reset control systems, the proposed tools allow computing exact steady-state responses to periodic excitations.      

The paper is organized as follows. Preliminaries on the frequency analysis for reset controllers are presented in Section~\ref{sec:2}. In Section~\ref{sec:3} sufficient conditions to define a notion of frequency response are presented. Then, a method to obtain closed-loop frequency responses for reset control systems, including high order harmonics, is developed, and pseudo-sensitivities are defined. In Section~\ref{sec:4} the steady-state response of reset controllers to periodic inputs is studied. In Section~\ref{sec:5} the performance of our proposed methods is assessed on an illustrative example. Finally, some concluding remarks and suggestions for future studies are given in Section~\ref{sec:6}.
%%%%%%%%%%%%%%%%%%%%%%%%%%%%%%%%%%%%%%%%%%%%%Preliminaries
%%%%%%Preliminaries
\section{Preliminaries}\label{sec:2}
In this section frequency-domain descriptions for reset controllers are briefly recalled. The state-space representation of a reset element is given by equations of the form
\begin{equation}\label{E-201}
\left\{
\begin{aligned}
\dot{x}_r(t) &=A_rx_r(t)+B_rr(t), & r(t)\neq0,  \\
x_r(t^+) &=A_\rho x(t), & r(t)=0, \\
u_r(t) &=C_rx(t)+D_rr(t),
\end{aligned}
\right.
\end{equation}
in which $x_r(t)\in\mathbb{R}^{n_r}$ is the reset states, $r(t)\in\mathbb{R}$ is an external signal, $u_r(t)$ is the control input, $A_r$, $B_r$, $C_r$ and $D_r$ are the dynamic matrices of the reset element, and $A_\rho$ determines the value of the reset states after the reset action. The transfer function $C_r(sI-A_r)^{-1}B_r+D_r$ is called the base transfer function of the reset controller. To study the reset controller~(\ref{E-201}) in the frequency-domain one could use various approaches. For example, in order to find the DF, a sinusoidal reference $r(t)=a_0\sin(\omega t)$, $\omega>0$ is applied and the output is approximated by means of the first harmonic of the Fourier series expansion of the steady-state response (provided if exists). In order to have a well-defined steady-state response we assume that $A_r$ has all eigenvalues with negative real part and $A_\rho e^{\frac{A_r\pi}{\omega}}$ has all eigenvalues with magnitude smaller than one \cite{guo2009frequency}. In this case, the state-space representation of the reset element~(\ref{E-201}) can be re-written as  
\begin{equation}\label{E-202}
\left\{
\begin{aligned}
\dot{x}_r(t) &=A_rx_r(t)+a_0B_r\sin(\omega t), & t\neq t_k, \\
x_r(t^+) &=A_\rho x(t), & t=t_k, \\
u_r(t) &=C_rx(t)+a_0D_r\sin(\omega t),
\end{aligned}
\right.
\end{equation}
with $\omega>0$, in which $t_k=\frac{k\pi}{\omega}$, with $k\in\mathbb{N}$, are the reset instants. According to \cite{guo2009frequency}, the DF associated to system~(\ref{E-202}) is given by
\begin{equation}\label{E-203}
\begin{aligned}
\mathcal{N}_{DF}(\omega) &=\dfrac{a_1(\omega)e^{j\varphi_1(\omega)}}{a_0} \\
&=C_r(j\omega I-A_r)^{-1}(I+j\theta(\omega))B_r+D_r,
\end{aligned}
\end{equation}
where 
\setlength{\arraycolsep}{0.0em}
\begin{eqnarray}\label{E-204}
\theta(\omega)&{=}&\frac{-2\omega^2}{\pi}(I+e^{\frac{\pi A_r}{\omega}})\Bigg((I+A_\rho e^{\frac{\pi A_r}{\omega}})^{-1}A_\rho(I+e^{\frac{\pi A_r}{\omega}})\nonumber\\
&&{-}\:I\Bigg)(\omega^2I+A_r^2)^{-1}.
\end{eqnarray}
\setlength{\arraycolsep}{5pt}Recently, a new tool, called Higher-Order Sinusoidal Input Describing Functions (HOSIDF), for studying non-linearities in the frequency-domain has been introduced in \cite{nuij2006higher}. In this method a non-linear system is considered as a virtual harmonic generator and the HOSIDF is defined as \cite{nuij2006higher}:
\begin{equation}\label{E-204-205}
H_n(j\omega)=\dfrac{a_n(\omega)e^{j\varphi_n(a_0,\omega)}}{a_0},
\end{equation} 
in which $a_n(\omega)$ and $\varphi_n(a_0,\omega)$ are the $n^{\text{th}}$ components of the Fourier series expansion of the steady-state output of the system to a sinusoidal input. This framework has been extended to the reset controller (\ref{E-201}) in \cite{ka} and $H_n(j\omega)$ for this controller is obtained as 
\begin{equation}\label{E-205}
\begin{cases}
C_r(j\omega I-A_r)^{-1}(I+j\theta(\omega))B_r+D_r, & n=1,\\
C_r(jn\omega I-A_r)^{-1}j\theta(\omega)B_r, & n>1\ \text{odd,}\\
0, & n\ \text{even.}
\end{cases}
\end{equation}
Note that the above frequency analysis is made simple by the fact that the reset instants are known, that is the reset controller is studied in the open-loop. Frequency properties of reset controllers as part of a closed-loop system in the presence of a periodic reference or disturbance input are much more difficult to study, and are the subject of the next section. 
%%%%%%%%%%%%%%%%%%%%%%%%%%%%%%%%%%%%%%frequency-domain Satbility Analysis
\section{Closed-loop frequency response of reset control systems}\label{sec:3}  
Consider the single-input single-output (SISO) control architecture in the top diagram of Fig.~\ref{F-31}. This includes as particular cases all schemes discussed in Section~\ref{sec:1}. The closed-loop system consists of a linear plant with transfer function $G(s)$, two linear controllers with proper transfer function $C_{\mathfrak{L}_1}(s)$ and $C_{\mathfrak{L}_2}(s)$, and a reset controller with base transfer function $C_\mathcal{R}(s)$. Let $\mathcal{L}$ be the LTI part of the system and assume that $G(s)$ is strictly proper. The state-space realization of $\mathcal{L}$ is described by the equations 
%%%%%%%%%%%%%%%%%%%%%%%%%%%%%%%%%%Figure
\tikzset{
	pattern size/.store in=\mcSize, 
	pattern size = 5pt,
	pattern thickness/.store in=\mcThickness, 
	pattern thickness = 0.3pt,
	pattern radius/.store in=\mcRadius, 
	pattern radius = 1pt}
\makeatletter
\pgfutil@ifundefined{pgf@pattern@name@_3fy9swdp5}{
	\pgfdeclarepatternformonly[\mcThickness,\mcSize]{_3fy9swdp5}
	{\pgfqpoint{0pt}{0pt}}
	{\pgfpoint{\mcSize+\mcThickness}{\mcSize+\mcThickness}}
	{\pgfpoint{\mcSize}{\mcSize}}
	{
		\pgfsetcolor{\tikz@pattern@color}
		\pgfsetlinewidth{\mcThickness}
		\pgfpathmoveto{\pgfqpoint{0pt}{0pt}}
		\pgfpathlineto{\pgfpoint{\mcSize+\mcThickness}{\mcSize+\mcThickness}}
		\pgfusepath{stroke}
}}
\makeatother
\begin{figure*}[!t]
	\centering
	\resizebox{0.7\textwidth}{!}{
		\tikzset{every picture/.style={line width=0.75pt}} %set default line width to 0.75pt        
		\begin{tikzpicture}[x=0.75pt,y=0.75pt,yscale=-1,xscale=1]
		%uncomment if require: \path (0,607); %set diagram left start at 0, and has height of 607
	%Down Arrow [id:dp24051202315657672] 
	\draw  [color={rgb, 255:red, 72; green, 147; blue, 233 }  ,draw opacity=1 ][line width=1.5]  (535,336) -- (538.63,336) -- (538.63,306) -- (545.88,306) -- (545.88,336) -- (549.5,336) -- (542.25,356) -- cycle ;
	%Shape: Rectangle [id:dp5396348120893653] 
	\draw  [color={rgb, 255:red, 155; green, 155; blue, 155 }  ,draw opacity=1 ] (213,387) -- (286.5,387) -- (286.5,572) -- (213,572) -- cycle ;
	%Straight Lines [id:da8629441761258987] 
	\draw [color={rgb, 255:red, 0; green, 0; blue, 0 }  ,draw opacity=1 ][line width=1.5]    (7,477) -- (209,477.98) ;
	\draw [shift={(213,478)}, rotate = 180.28] [fill={rgb, 255:red, 0; green, 0; blue, 0 }  ,fill opacity=1 ][line width=0.08]  [draw opacity=0] (11.61,-5.58) -- (0,0) -- (11.61,5.58) -- cycle    ;
	%Straight Lines [id:da5578509382959564] 
	\draw [color={rgb, 255:red, 155; green, 155; blue, 155 }  ,draw opacity=1 ][line width=1.5]    (286.5,416) -- (476.5,416.98) ;
	\draw [shift={(480.5,417)}, rotate = 180.3] [fill={rgb, 255:red, 155; green, 155; blue, 155 }  ,fill opacity=1 ][line width=0.08]  [draw opacity=0] (11.61,-5.58) -- (0,0) -- (11.61,5.58) -- cycle    ;
	%Straight Lines [id:da19893626905818396] 
	\draw [color={rgb, 255:red, 155; green, 155; blue, 155 }  ,draw opacity=1 ][line width=1.5]    (286.5,457) -- (475.5,456.02) ;
	\draw [shift={(479.5,456)}, rotate = 539.7] [fill={rgb, 255:red, 155; green, 155; blue, 155 }  ,fill opacity=1 ][line width=0.08]  [draw opacity=0] (11.61,-5.58) -- (0,0) -- (11.61,5.58) -- cycle    ;
	%Shape: Circle [id:dp9480243095091221] 
	\draw  [color={rgb, 255:red, 155; green, 155; blue, 155 }  ,draw opacity=1 ][fill={rgb, 255:red, 155; green, 155; blue, 155 }  ,fill opacity=1 ] (374,477.5) .. controls (374,476.12) and (375.12,475) .. (376.5,475) .. controls (377.88,475) and (379,476.12) .. (379,477.5) .. controls (379,478.88) and (377.88,480) .. (376.5,480) .. controls (375.12,480) and (374,478.88) .. (374,477.5) -- cycle ;
	%Shape: Circle [id:dp07849738945654305] 
	\draw  [color={rgb, 255:red, 155; green, 155; blue, 155 }  ,draw opacity=1 ][fill={rgb, 255:red, 155; green, 155; blue, 155 }  ,fill opacity=1 ] (375,519.5) .. controls (375,518.12) and (376.12,517) .. (377.5,517) .. controls (378.88,517) and (380,518.12) .. (380,519.5) .. controls (380,520.88) and (378.88,522) .. (377.5,522) .. controls (376.12,522) and (375,520.88) .. (375,519.5) -- cycle ;
	%Shape: Circle [id:dp16456518160951794] 
	\draw  [color={rgb, 255:red, 155; green, 155; blue, 155 }  ,draw opacity=1 ][fill={rgb, 255:red, 155; green, 155; blue, 155 }  ,fill opacity=1 ] (375,496.5) .. controls (375,495.12) and (376.12,494) .. (377.5,494) .. controls (378.88,494) and (380,495.12) .. (380,496.5) .. controls (380,497.88) and (378.88,499) .. (377.5,499) .. controls (376.12,499) and (375,497.88) .. (375,496.5) -- cycle ;
	%Straight Lines [id:da3247910053997456] 
	\draw [color={rgb, 255:red, 155; green, 155; blue, 155 }  ,draw opacity=1 ][line width=1.5]    (286.5,552) -- (474.5,552) ;
	\draw [shift={(478.5,552)}, rotate = 180] [fill={rgb, 255:red, 155; green, 155; blue, 155 }  ,fill opacity=1 ][line width=0.08]  [draw opacity=0] (11.61,-5.58) -- (0,0) -- (11.61,5.58) -- cycle    ;
	%Shape: Square [id:dp21438006306710577] 
	\draw  [color={rgb, 255:red, 155; green, 155; blue, 155 }  ,draw opacity=1 ] (480,537) -- (510,537) -- (510,567) -- (480,567) -- cycle ;
	%Shape: Square [id:dp9119748532302885] 
	\draw  [color={rgb, 255:red, 155; green, 155; blue, 155 }  ,draw opacity=1 ] (479,441) -- (509,441) -- (509,471) -- (479,471) -- cycle ;
	%Shape: Square [id:dp6039841496538612] 
	\draw  [color={rgb, 255:red, 155; green, 155; blue, 155 }  ,draw opacity=1 ] (480,402) -- (510,402) -- (510,432) -- (480,432) -- cycle ;
	%Straight Lines [id:da6116989773907988] 
	\draw [color={rgb, 255:red, 155; green, 155; blue, 155 }  ,draw opacity=1 ][line width=1.5]    (510.5,415) -- (770.5,415) -- (845,456.07) ;
	\draw [shift={(848.5,458)}, rotate = 208.87] [fill={rgb, 255:red, 155; green, 155; blue, 155 }  ,fill opacity=1 ][line width=0.08]  [draw opacity=0] (11.61,-5.58) -- (0,0) -- (11.61,5.58) -- cycle    ;
	%Straight Lines [id:da2795396156308452] 
	\draw [color={rgb, 255:red, 155; green, 155; blue, 155 }  ,draw opacity=1 ][line width=1.5]    (509,455) -- (776.5,456) -- (827.23,474.15) ;
	\draw [shift={(831,475.5)}, rotate = 199.69] [fill={rgb, 255:red, 155; green, 155; blue, 155 }  ,fill opacity=1 ][line width=0.08]  [draw opacity=0] (11.61,-5.58) -- (0,0) -- (11.61,5.58) -- cycle    ;
	%Straight Lines [id:da46968801467063415] 
	\draw [color={rgb, 255:red, 155; green, 155; blue, 155 }  ,draw opacity=1 ][line width=1.5]    (511.5,553) -- (774.5,552) -- (845.37,495.49) ;
	\draw [shift={(848.5,493)}, rotate = 501.43] [fill={rgb, 255:red, 155; green, 155; blue, 155 }  ,fill opacity=1 ][line width=0.08]  [draw opacity=0] (11.61,-5.58) -- (0,0) -- (11.61,5.58) -- cycle    ;
	%Shape: Circle [id:dp32128541802612776] 
	\draw  [color={rgb, 255:red, 155; green, 155; blue, 155 }  ,draw opacity=1 ][line width=1.5]  (831,475.5) .. controls (831,465.84) and (838.84,458) .. (848.5,458) .. controls (858.16,458) and (866,465.84) .. (866,475.5) .. controls (866,485.16) and (858.16,493) .. (848.5,493) .. controls (838.84,493) and (831,485.16) .. (831,475.5) -- cycle ;
	%Straight Lines [id:da09033592769574539] 
	\draw [color={rgb, 255:red, 0; green, 0; blue, 0 }  ,draw opacity=1 ][line width=1.5]    (866,476) -- (1225,473.03) ;
	\draw [shift={(1229,473)}, rotate = 539.53] [fill={rgb, 255:red, 0; green, 0; blue, 0 }  ,fill opacity=1 ][line width=0.08]  [draw opacity=0] (11.61,-5.58) -- (0,0) -- (11.61,5.58) -- cycle    ;
	%Shape: Circle [id:dp25129676286146774] 
	\draw  [color={rgb, 255:red, 155; green, 155; blue, 155 }  ,draw opacity=1 ][fill={rgb, 255:red, 155; green, 155; blue, 155 }  ,fill opacity=1 ] (492.24,475.01) .. controls (493.62,474.87) and (494.84,475.87) .. (494.99,477.24) .. controls (495.13,478.62) and (494.13,479.84) .. (492.76,479.99) .. controls (491.38,480.13) and (490.16,479.13) .. (490.01,477.76) .. controls (489.87,476.38) and (490.87,475.16) .. (492.24,475.01) -- cycle ;
	%Shape: Circle [id:dp48426346751973415] 
	\draw  [color={rgb, 255:red, 155; green, 155; blue, 155 }  ,draw opacity=1 ][fill={rgb, 255:red, 155; green, 155; blue, 155 }  ,fill opacity=1 ] (490,520.5) .. controls (490,519.12) and (491.12,518) .. (492.5,518) .. controls (493.88,518) and (495,519.12) .. (495,520.5) .. controls (495,521.88) and (493.88,523) .. (492.5,523) .. controls (491.12,523) and (490,521.88) .. (490,520.5) -- cycle ;
	%Shape: Circle [id:dp5311332625934171] 
	\draw  [color={rgb, 255:red, 155; green, 155; blue, 155 }  ,draw opacity=1 ][fill={rgb, 255:red, 155; green, 155; blue, 155 }  ,fill opacity=1 ] (491,497.5) .. controls (491,496.12) and (492.12,495) .. (493.5,495) .. controls (494.88,495) and (496,496.12) .. (496,497.5) .. controls (496,498.88) and (494.88,500) .. (493.5,500) .. controls (492.12,500) and (491,498.88) .. (491,497.5) -- cycle ;
	%Shape: Circle [id:dp43161351349896426] 
	\draw  [color={rgb, 255:red, 155; green, 155; blue, 155 }  ,draw opacity=1 ][fill={rgb, 255:red, 155; green, 155; blue, 155 }  ,fill opacity=1 ] (652,474.5) .. controls (652,473.12) and (653.12,472) .. (654.5,472) .. controls (655.88,472) and (657,473.12) .. (657,474.5) .. controls (657,475.88) and (655.88,477) .. (654.5,477) .. controls (653.12,477) and (652,475.88) .. (652,474.5) -- cycle ;
	%Shape: Circle [id:dp8065922446453497] 
	\draw  [color={rgb, 255:red, 155; green, 155; blue, 155 }  ,draw opacity=1 ][fill={rgb, 255:red, 155; green, 155; blue, 155 }  ,fill opacity=1 ] (652,516.5) .. controls (652,515.12) and (653.12,514) .. (654.5,514) .. controls (655.88,514) and (657,515.12) .. (657,516.5) .. controls (657,517.88) and (655.88,519) .. (654.5,519) .. controls (653.12,519) and (652,517.88) .. (652,516.5) -- cycle ;
	%Shape: Circle [id:dp31069373318156934] 
	\draw  [color={rgb, 255:red, 155; green, 155; blue, 155 }  ,draw opacity=1 ][fill={rgb, 255:red, 155; green, 155; blue, 155 }  ,fill opacity=1 ] (652,493.5) .. controls (652,492.12) and (653.12,491) .. (654.5,491) .. controls (655.88,491) and (657,492.12) .. (657,493.5) .. controls (657,494.88) and (655.88,496) .. (654.5,496) .. controls (653.12,496) and (652,494.88) .. (652,493.5) -- cycle ;
	%Shape: Rectangle [id:dp6605462554925079] 
	\draw  [pattern=_3fy9swdp5,pattern size=5pt,pattern thickness=0.55pt,pattern radius=0pt, pattern color={rgb, 255:red, 155; green, 155; blue, 155}][line width=1.5]  (420.5,106) -- (491.5,106) -- (491.5,161) -- (420.5,161) -- cycle ;
	%Straight Lines [id:da0031703432484878613] 
	\draw [color={rgb, 255:red, 0; green, 0; blue, 0 }  ,draw opacity=1 ][line width=1.5]    (2.5,131) -- (209,130.02) ;
	\draw [shift={(213,130)}, rotate = 539.73] [fill={rgb, 255:red, 0; green, 0; blue, 0 }  ,fill opacity=1 ][line width=0.08]  [draw opacity=0] (11.61,-5.58) -- (0,0) -- (11.61,5.58) -- cycle    ;
	%Straight Lines [id:da7453473518186187] 
	\draw [color={rgb, 255:red, 0; green, 0; blue, 0 }  ,draw opacity=1 ][line width=1.5]    (817,131) -- (1218.5,128.03) ;
	\draw [shift={(1222.5,128)}, rotate = 539.5799999999999] [fill={rgb, 255:red, 0; green, 0; blue, 0 }  ,fill opacity=1 ][line width=0.08]  [draw opacity=0] (11.61,-5.58) -- (0,0) -- (11.61,5.58) -- cycle    ;
	%Shape: Circle [id:dp5938671444643586] 
	\draw  [color={rgb, 255:red, 0; green, 0; blue, 0 }  ,draw opacity=1 ][line width=1.5]  (210,131) .. controls (210,122.85) and (216.6,116.25) .. (224.75,116.25) .. controls (232.9,116.25) and (239.5,122.85) .. (239.5,131) .. controls (239.5,139.15) and (232.9,145.75) .. (224.75,145.75) .. controls (216.6,145.75) and (210,139.15) .. (210,131) -- cycle ;
	%Straight Lines [id:da01083905776548133] 
	\draw [color={rgb, 255:red, 0; green, 0; blue, 0 }  ,draw opacity=1 ][line width=1.5]    (240.5,132) -- (276,132) ;
	\draw [shift={(280,132)}, rotate = 180] [fill={rgb, 255:red, 0; green, 0; blue, 0 }  ,fill opacity=1 ][line width=0.08]  [draw opacity=0] (11.61,-5.58) -- (0,0) -- (11.61,5.58) -- cycle    ;
	%Straight Lines [id:da10270460416684901] 
	\draw [line width=1.5]    (835.5,132) -- (835.5,242) -- (225.5,240) -- (224.78,149.75) ;
	\draw [shift={(224.75,145.75)}, rotate = 449.54] [fill={rgb, 255:red, 0; green, 0; blue, 0 }  ][line width=0.08]  [draw opacity=0] (11.61,-5.58) -- (0,0) -- (11.61,5.58) -- cycle    ;
	%Straight Lines [id:da34741076551226613] 
	\draw [color={rgb, 255:red, 0; green, 0; blue, 0 }  ,draw opacity=1 ][line width=1.5]    (492.5,132) -- (557.5,132) ;
	\draw [shift={(561.5,132)}, rotate = 180] [fill={rgb, 255:red, 0; green, 0; blue, 0 }  ,fill opacity=1 ][line width=0.08]  [draw opacity=0] (11.61,-5.58) -- (0,0) -- (11.61,5.58) -- cycle    ;
	%Straight Lines [id:da38332059848832956] 
	\draw [color={rgb, 255:red, 0; green, 0; blue, 0 }  ,draw opacity=1 ][line width=1.5]    (632,132) -- (683,132) ;
	\draw [shift={(687,132)}, rotate = 180] [fill={rgb, 255:red, 0; green, 0; blue, 0 }  ,fill opacity=1 ][line width=0.08]  [draw opacity=0] (11.61,-5.58) -- (0,0) -- (11.61,5.58) -- cycle    ;
	%Shape: Rectangle [id:dp1066763362219616] 
	\draw  [line width=1.5]  (279.5,106) -- (351.5,106) -- (351.5,159) -- (279.5,159) -- cycle ;
	%Straight Lines [id:da8263895189538635] 
	\draw [color={rgb, 255:red, 0; green, 0; blue, 0 }  ,draw opacity=1 ][line width=1.5]    (353,131) -- (416.5,131) ;
	\draw [shift={(420.5,131)}, rotate = 180] [fill={rgb, 255:red, 0; green, 0; blue, 0 }  ,fill opacity=1 ][line width=0.08]  [draw opacity=0] (11.61,-5.58) -- (0,0) -- (11.61,5.58) -- cycle    ;
	%Shape: Circle [id:dp11966234546254917] 
	\draw  [color={rgb, 255:red, 0; green, 0; blue, 0 }  ,draw opacity=1 ][line width=1.5]  (687,132) .. controls (687,123.85) and (693.6,117.25) .. (701.75,117.25) .. controls (709.9,117.25) and (716.5,123.85) .. (716.5,132) .. controls (716.5,140.15) and (709.9,146.75) .. (701.75,146.75) .. controls (693.6,146.75) and (687,140.15) .. (687,132) -- cycle ;
	%Straight Lines [id:da7386714856366257] 
	\draw [color={rgb, 255:red, 0; green, 0; blue, 0 }  ,draw opacity=1 ][line width=1.5]    (716.5,132) -- (740,132) ;
	\draw [shift={(744,132)}, rotate = 180] [fill={rgb, 255:red, 0; green, 0; blue, 0 }  ,fill opacity=1 ][line width=0.08]  [draw opacity=0] (11.61,-5.58) -- (0,0) -- (11.61,5.58) -- cycle    ;
	%Straight Lines [id:da44762221952202297] 
	\draw [color={rgb, 255:red, 0; green, 0; blue, 0 }  ,draw opacity=1 ][line width=1.5]    (702.5,62) -- (701.8,113.25) ;
	\draw [shift={(701.75,117.25)}, rotate = 270.78] [fill={rgb, 255:red, 0; green, 0; blue, 0 }  ,fill opacity=1 ][line width=0.08]  [draw opacity=0] (11.61,-5.58) -- (0,0) -- (11.61,5.58) -- cycle    ;
	%Shape: Rectangle [id:dp42328883593426314] 
	\draw  [line width=1.5]  (559.5,107) -- (631.5,107) -- (631.5,160) -- (559.5,160) -- cycle ;
	%Shape: Rectangle [id:dp3851956651721279] 
	\draw  [line width=1.5]  (744.5,106) -- (816.5,106) -- (816.5,159) -- (744.5,159) -- cycle ;
	%Shape: Path Data [id:dp24659861493521584] 
	\draw  [color={rgb, 255:red, 128; green, 128; blue, 128 }  ,draw opacity=1 ][dash pattern={on 5.63pt off 4.5pt}][line width=1.5]  (848.5,99) -- (848.5,271) -- (198.5,271) -- (198.5,99) -- (402.81,99) -- (402.81,167.44) -- (499.96,167.44) -- (499.96,99) -- (848.5,99) -- cycle ;
	%Rounded Rect [id:dp3470152354004399] 
	\draw  [color={rgb, 255:red, 74; green, 144; blue, 226 }  ,draw opacity=1 ][dash pattern={on 6.75pt off 4.5pt}][line width=2.25]  (188.5,103.4) .. controls (188.5,76.12) and (210.62,54) .. (237.9,54) -- (817.1,54) .. controls (844.38,54) and (866.5,76.12) .. (866.5,103.4) -- (866.5,251.6) .. controls (866.5,278.88) and (844.38,301) .. (817.1,301) -- (237.9,301) .. controls (210.62,301) and (188.5,278.88) .. (188.5,251.6) -- cycle ;
	%Rounded Rect [id:dp5416402798097402] 
	\draw  [color={rgb, 255:red, 74; green, 144; blue, 226 }  ,draw opacity=1 ][dash pattern={on 6.75pt off 4.5pt}][line width=2.25]  (196.5,405.8) .. controls (196.5,379.95) and (217.45,359) .. (243.3,359) -- (827.7,359) .. controls (853.55,359) and (874.5,379.95) .. (874.5,405.8) -- (874.5,546.2) .. controls (874.5,572.05) and (853.55,593) .. (827.7,593) -- (243.3,593) .. controls (217.45,593) and (196.5,572.05) .. (196.5,546.2) -- cycle ;
	
	% Text Node
	\draw (251,478.5) node  [scale=1.2,color={rgb, 255:red, 74; green, 74; blue, 74 }  ,opacity=1 ] [align=left] { \ {\fontfamily{ptm}\selectfont \textbf{Virtual}}\\{\fontfamily{ptm}\selectfont \textbf{Harmonic}}\\{\fontfamily{ptm}\selectfont \textbf{Generator}}};
	% Text Node
	\draw (94,459) node  [scale=1.2,font=\large]  {$\boldsymbol{r(t) =a_{0}\sin( \omega t+\varphi _{0})}$};
	% Text Node
	\draw (359,399) node  [scale=1.4,font=\large,color={rgb, 255:red, 74; green, 74; blue, 74 }  ,opacity=1 ]  {$a_{0}\sin( \omega t+\varphi _{0})$};
	% Text Node
	\draw (370,441) node  [scale=1.4,font=\large,color={rgb, 255:red, 74; green, 74; blue, 74 }  ,opacity=1 ]  {$a_{0}\sin( 2( \omega t+\varphi _{0}))$};
	% Text Node
	\draw (371,536) node  [scale=1.2,font=\large,color={rgb, 255:red, 74; green, 74; blue, 74 }  ,opacity=1 ]  {$a_{0}\sin( n( \omega t+\varphi _{0}))$};
	% Text Node
	\draw (495,552) node  [scale=1.2,font=\large,color={rgb, 255:red, 74; green, 74; blue, 74 }  ,opacity=1 ]  {$\mathbf{H_{n}}$};
	% Text Node
	\draw (494,456) node  [scale=1.2,font=\large,color={rgb, 255:red, 74; green, 74; blue, 74 }  ,opacity=1 ]  {$\mathbf{H_{2}}$};
	% Text Node
	\draw (495,417) node  [scale=1.2,font=\large,color={rgb, 255:red, 74; green, 74; blue, 74 }  ,opacity=1 ]  {$\mathbf{H_{1}}$};
	% Text Node
	\draw (635,399) node  [scale=1.25,font=\large,color={rgb, 255:red, 74; green, 74; blue, 74 }  ,opacity=1 ]  {$a_{1}( \omega )\sin( \omega t+\varphi _{0} +\varphi _{1}( \omega ))$};
	% Text Node
	\draw (644,437.5) node  [scale=1.25,font=\large,color={rgb, 255:red, 74; green, 74; blue, 74 }  ,opacity=1 ]  {$a_{2}( \omega )\sin( 2( \omega t+\varphi _{0}) +\varphi _{2}( \omega ))$};
	% Text Node
	\draw (646,536) node  [scale=1.25,font=\large,color={rgb, 255:red, 74; green, 74; blue, 74 }  ,opacity=1 ]  {$a_{n}( \omega )\sin( n( \omega t+\varphi _{0}) +\varphi _{n}( \omega ))$};
	% Text Node
	\draw (849.5,474) node  [scale=1.2,font=\large,color={rgb, 255:red, 155; green, 155; blue, 155 }  ,opacity=1 ]  {$\boldsymbol{\sum }$};
	% Text Node
	\draw (1073,445) node  [scale=1.2,font=\large]  {$\boldsymbol{y( t) =\sum\limits ^{\infty }_{n=1} a_{n}( \omega )\sin( n( \omega t+\varphi _{0}) +\varphi _{n}( \omega ))}$};
	% Text Node
	\draw (456,74) node   [scale=1.2,align=left] {{\fontfamily{ptm}\selectfont \textbf{ \ \ Reset}}\\{\fontfamily{ptm}\selectfont \textbf{Controller}}};
	% Text Node
	\draw (219,128) node  [scale=1.2,font=\large]  {$\boldsymbol{+}$};
	% Text Node
	\draw (225,139) node  [scale=1.2,font=\large]  {$\boldsymbol{-}$};
	% Text Node
	\draw (460,130.5) node  [scale=1.2,font=\huge]  {$C\mathfrak{_{\mathfrak{R}}}$};
	% Text Node
	\draw (780.5,131.5) node  [scale=1.2,font=\huge]  {$G$};
	% Text Node
	\draw (255,111) node  [scale=1.2,font=\large]  {$\boldsymbol{e(t)}$};
	% Text Node
	\draw (657,115) node  [scale=1.3,font=\large]  {$\boldsymbol{u(t)}$};
	% Text Node
	\draw (782,90) node   [scale=1.3,align=left] {{\fontfamily{ptm}\selectfont \textbf{Plant}}};
	% Text Node
	\draw (90,111) node  [scale=1.2,font=\large,color={rgb, 255:red, 0; green, 0; blue, 0 }  ,opacity=1 ]  {$\boldsymbol{r( t) =a_{0}\sin( \omega t+\varphi _{0})}$};
	% Text Node
	\draw (1070,99) node  [scale=1.2,font=\large]  {$\boldsymbol{y( t) =\sum\limits ^{\infty }_{n=1} a_{n}( \omega )\sin( n( \omega t+\varphi _{0}) +\varphi _{n}( \omega ))}$};
	% Text Node
	\draw (522,28) node   [scale=1.2,align=left] {{\fontfamily{ptm}\selectfont \textbf{{\large Non-Linear}}}\\{\fontfamily{ptm}\selectfont \textbf{{\large  \ \ \ System}}}};
	% Text Node
	\draw (319.5,132) node  [scale=1.2,font=\huge]  {$C_{\mathfrak{L}_1}$};
	% Text Node
	\draw (315,77) node   [scale=1.2,align=left] {{\fontfamily{ptm}\selectfont \textbf{ \ \ Linear}}\\{\fontfamily{ptm}\selectfont \textbf{Controller}}};
	% Text Node
	\draw (694,132) node  [scale=1.2,font=\large]  {$\boldsymbol{+}$};
	% Text Node
	\draw (701,122) node  [scale=1.2,font=\large]  {$\boldsymbol{+}$};
	% Text Node
	\draw (685,78) node  [scale=1.2,font=\large]  {$\boldsymbol{d(t)}$};
	% Text Node
	\draw (379,113) node  [scale=1.2,font=\large]  {$\boldsymbol{e_{R}(t)}$};
	% Text Node
	\draw (531,114) node  [scale=1.2,font=\large]  {$\boldsymbol{u_{R}( t)}$};
	% Text Node
	\draw (599.5,132) node  [scale=1.2,font=\huge]  {$C_{\mathfrak{L}_2}$};
	% Text Node
	\draw (596,76) node   [scale=1.2,align=left] {{\fontfamily{ptm}\selectfont \textbf{ \ Linear}}\\{\fontfamily{ptm}\selectfont \textbf{Controller}}};
	% Text Node
	\draw (211.5,255.5) node  [scale=1.2,font=\LARGE]  {$\mathcal{L}$};
		\end{tikzpicture}}
	\caption{Closed-loop architecture with reset controller (top). HOSIDF representation of the closed-loop configuration (bottom).}
	\label{F-31}
\end{figure*}
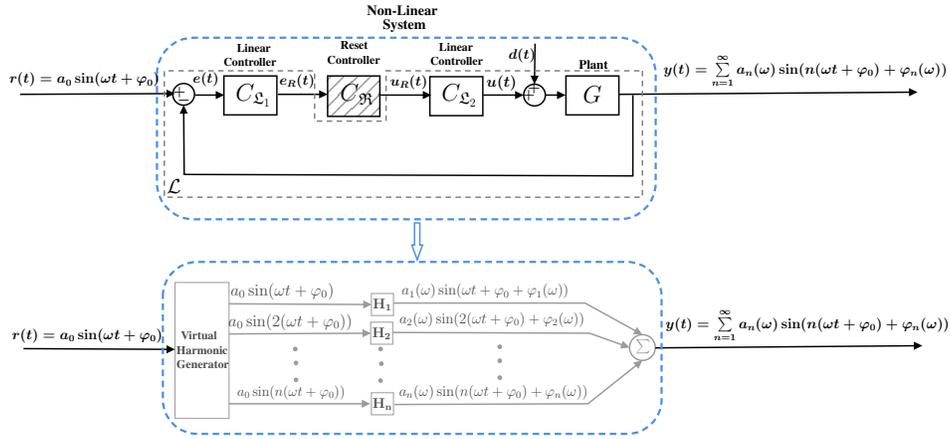
\begin{equation}\label{E-301}
\mathcal{L}:\left\{
\begin{aligned}
\dot{\zeta}(t) &=A\zeta(t)+Bw(t)+B_uu_R(t),\\
u(t) &=C_u\zeta(t)+D_{u}r(t),\\
e_R(t) &=C_{e_R}\zeta(t)+D_{e_R}r(t),\\
y(t) &=C\zeta(t),
\end{aligned}
\right.
\end{equation}
where $\zeta(t)\in\mathbb{R}^{n_p}$ describes the states of the plant and of the linear controllers ($n_p$ is the number of states of the linear part), $A$, $B$, $C$, $B_u$, $C_{e_R}$, $C_u$, $D_u$ and $D_{e_R}$ are the corresponding dynamic matrices, $y(t)\in\mathbb{R}$ is the output of the plant and $w(t)=[r(t)\ d(t)]^T\in\mathbb{R}^2$ is an external input. The state-space representation of the reset controller is given by the equations
\begin{equation}\label{E-302}
\left\{
\begin{aligned}
\dot{x}_r(t) &=A_rx_r(t)+B_re_R(t), & e_R(t)\neq0,  \\
x_r(t^+) &=A_\rho x_r(t), & e_R(t)=0, \\
u_R(t) &=C_rx_r(t)+D_re_R(t).
\end{aligned}
\right.
\end{equation}
The closed-loop state-space representation of the overall system can, therefore, be written as
\begin{equation}\label{E-303}
\left\{
\begin{aligned}
\dot{x}(t)&=\bar{A}x(t)+\bar{B}w(t), & e_R(t)\neq 0,\\
x(t^+)&=\bar{A}_\rho x(t), & e_R(t)=0,  \\
u(t)&=\bar{C}_ux(t)+\bar{D}_{u}r(t),\\
e_R(t)&=\bar{C}_{e_R}x(t)+D_{e_R}r(t),\\
y(t)&=\bar{C}x(t),
\end{aligned}
\right.
\end{equation} 
where $x(t)=[x_r(t)^T\quad \zeta(t)^T]^T\in\mathbb{R}^{n_p+n_r}$, and \newline\newline 
$\bar{A}=\begin{bmatrix} A_r & B_rC_{eR} \\ B_uC_r & A+B_uD_rC_{eR}\end{bmatrix}$, $\bar{C}=\begin{bmatrix} 0_{1\times n_r} & C \end{bmatrix}$, $\bar{B}=\begin{bmatrix}0_{n_r\times 2}\\B\end{bmatrix}+\begin{bmatrix} B_rD_{eR} & 0_{n_r\times 1} \\ B_uD_rD_{eR} & 0_{n_p\times 1} \end{bmatrix}$, $\bar{A}_\rho=\begin{bmatrix}A_\rho & 0_{n_r\times n_p} \\ 0_{n_p\times n_r} & I_{n_p\times n_p} \end{bmatrix}$, $\bar{C}_u=\begin{bmatrix} C_rD_{\mathfrak{L}_2} & C_{e_R}D_rD_{\mathfrak{L}_2}+C_u \end{bmatrix}$, $\bar{C}_{e_R}=\begin{bmatrix} 0_{1\times n_r} & C_{e_R} \end{bmatrix}$, and $\bar{D}_{u}=D_{u}D_{e_R}D_r$ with $D_{\mathfrak{L}_2}$ the feedthrough matrix of $C_{\mathfrak{L}_2}(s)$.  
%%%%%%%%%%%%%%%%%%%%%%%%%%%%%%%%%%%%%%%%Stability and convergence
\subsection{Stability and Convergence}\label{sec:31}
In this section sufficient conditions for the existence of a steady-state solution for the closed-loop reset control system~(\ref{E-303}) driven by periodic inputs is provided. This is based on the $H_\beta$ condition \cite{beker2004fundamental,guo2015analysis,hollot2001establishing,AliCDC}, which we recall in what follows. Let 
\begin{equation}\label{Ex-31}
\begin{array}{*{35}{c}}
C_0=\begin{bmatrix}\rho & \beta C_{eR}\end{bmatrix},\quad B_0=\begin{bmatrix} I_{n_r\times n_r} \\ 0_{n_p\times n_r} \end{bmatrix},\\
\rho=\rho^T>0,\quad \rho\in\mathbb{R}^{n_r\times n_r},\quad \beta\in\mathbb{R}^{n_r\times 1}.
\end{array}
\end{equation}
The $H_\beta$ condition states that the reset control system~(\ref{E-303}) with $w=0$ is quadratically stable if and only if there exist $\rho=\rho^T>0$ and $\beta$ such that the transfer function 
\begin{equation}\label{Ex-311}
H(s)=C_0(sI-\bar{A})^{-1}B_0
\end{equation} 
is Strictly Positive Real (SPR), $(\bar{A},B_0)$ and $(\bar{A},C_0)$ are controllable and observable, respectively, and
\begin{equation}\label{Ex-312}
A_\rho^T\rho A_\rho-\rho<0.
\end{equation}
\begin{definition}\label{DD0}
A time $\bar{T}>0$ is called a reset instant for the reset control system~(\ref{E-303}) if $e_R(\bar{T})=0$. For any given initial condition and input $w$ the resulting set of all reset instants defines the reset sequence $\{t_k\}$, with $t_k\leq t_{k+1}$, for all $k\in\mathbb{N}$. The reset instants $t_k$ of the reset control system (\ref{E-303}) have the well-posedness property if for any initial condition $x_0$ and any input $w$, all reset instants are distinct, and there exists a $\lambda>0$ such that for all $k\in\mathbb{N}$, $\lambda\leq t_{k+1}-t_k$ \cite{banos2016impulsive,banos2011reset}.
\end{definition}
\begin{remark}\label{R0}
{\rm If the $H_\beta$ condition holds, then the reset control system~(\ref{E-303}) has the uniform bounded-input bounded-state (UBIBS) property and the reset instants have the well-posedness property \cite{dastjerdistabil}. Therefore, the reset control system~(\ref{E-303}) has a unique well-defined solution for $t\geq t_0$ for any initial condition $x_0$ and input $w(t)$ which is a Bohl function~\cite{banos2016impulsive,banos2011reset}.}
\end{remark}
To develop a frequency analysis for the reset control system~(\ref{E-303}), the following assumption is required.
\begin{assumption}\label{AS1}
The initial condition of the reset controller is zero. In addition, there are infinitely many reset instants and $\displaystyle\lim_{k\to\infty} t_k=\infty$. 
\end{assumption}
The second term in Assumption~\ref{AS1} is introduced to rule out a trivial situation. In fact, if $\displaystyle\lim_{k\to\infty} t_k=T_K$, then for all $t\geq T_K$ the reset control system~(\ref{E-303}) is a stable linear system. Two important technical lemmas, which are used in the proof of the following theorem, are now formulated and proved. 
\begin{lemma}\label{L1}
Let $\{t_k\}$ and $\{\tilde{t}_k\}$ be the reset sequences of the reset control system~(\ref{E-303}) for two different initial conditions $\zeta_{0}$ and $\tilde{\zeta}_{0}$ of the linear part and for the same input. Suppose Assumption~\ref{AS1} and the $H_\beta$ condition hold and $w$ is a Bohl function. Then $\displaystyle\lim_{k\to\infty}(t_{k}-\tilde{t}_k)=0$.
\end{lemma}
\begin{proof}
To begin with note that, for any initial condition $x_0=\begin{bmatrix}0^T & \zeta_0^T\end{bmatrix}^T$, the signal $e_R(t)$ in~(\ref{E-303}) can be obtained through the equation (see Lemma~\ref{AP1} in the Appendix) 
\begin{equation}\label{E-305-306-01}
\left\{
\begin{aligned}
\dot{x}_I(t)&=\bar{A}x_I(t)+\bar{B}w(t)+\begin{bmatrix}B_r\\0_{n_p\times 2}\end{bmatrix}w_I(t), & e_R(t)&\neq 0,\\
x_I(t^+)&=\bar{A}_\rho x_I(t), & e_R(t)&=0,  \\
\\
e_R(t)&=\bar{C}_{e_R}x_I(t)+D_{e_R}r(t)+[1\ 0]w_I(t),\\
\end{aligned}
\right.
\end{equation} 
with $x_I(0)=0$ and
\begin{equation}\label{E-305-306-00}
\left\{
\begin{aligned}
\dot{Z}(t)&=AZ(t),\\
\\
w_I(t)&=\begin{bmatrix}C_{eR}\\0
\end{bmatrix}Z(t),
\end{aligned}
\right.
\quad Z(0)=\zeta_0.
\end{equation}  	
Since the linear part of the system contains the internal model~(\ref{E-305-306-00}) of $w_I$, and $w(t)$ is a Bohl function, based on~\cite{beker2004fundamental,hollot2001establishing} $e_R(t)$ is asymptotically independent of $w_I(t)$. This implies that $\displaystyle\lim_{k\to\infty}(t_{k}-\tilde{t}_{k})=0$.
\end{proof}
\begin{lemma}\label{L2}
Consider the reset control system~(\ref{E-303}). Suppose Assumption~\ref{AS1} holds, $w$ is a Bohl function, and the $H_\beta$ condition is satisfied. Then the reset control system~(\ref{E-303}) is uniformly exponentially convergent. 
\end{lemma}
\begin{proof}
To begin with note that the property of uniformly exponentially convergence is as given in \cite{pavlov2007frequency}.  Since the $H_\beta$ condition is satisfied, according to Remark~\ref{R0}, the reset control system~(\ref{E-303}) has a unique well-defined solution for any initial condition $x_0$ and any $w$ which is a Bohl function. Let $x$ and $\tilde{x}$ be two solutions of the reset control system~(\ref{E-303}) corresponding to the some input $w$ and to two different initial conditions. Since the $H_\beta$ condition is satisfied $x(t)$ and $\tilde{x}(t)$ are bounded for all $t$. Let $\Delta x:=x(t)-\tilde{x}(t)$, and let $\{t_k\}$ and $\{\tilde{t}_k\}$ be the reset sequences of $x(t)$ and $\tilde{x}(t)$. Define $\mathcal{M}=\{t\in\mathbb{R}^+|\ t\neq t_k\land t\neq\tilde{t}_k \}$. By Lemma~\ref{L1}
\begin{equation}\label{K-01}
\forall\ \delta>0,\ \exists\ \Pi>0\text{ such that } k>\Pi\Rightarrow |t_k-\tilde{t}_k|<\delta.
\end{equation} 
Moreover, by the well-posedness property, there exists a $\lambda>0$ such that $\lambda\leq t_{k+1}-t_k$ and $\lambda\leq \tilde{t}_{k+1}-\tilde{t}_k$. Thus, selecting $\delta$ sufficiently small yields
\begin{equation}\label{K-02}
x(t_k+\delta)=e^{\bar{A}\delta}\bar{A}_\rho x(t_k)+\displaystyle\int_{t_k}^{t_k+\delta}e^{\bar{A}(t_k+\delta-\tau)}\bar{B}w(\tau)d\tau,
\end{equation} 
for all $t_k$ sufficiently large. By (\ref{K-01}), $\tilde{t}_k=t_{k}+\delta^\prime,$ with $0\leq \delta^\prime\leq\delta$. Thus
\setlength{\arraycolsep}{0.0em}
\begin{eqnarray}\label{K-03}
\tilde{x}(t_k+\delta)&{=}&e^{\bar{A}(\delta-\delta^\prime)}\bar{A}_\rho \Bigg(e^{\bar{A}\delta^\prime}\tilde{x}(t_k)\nonumber\\
&&{+}\:\displaystyle\int_{t_k}^{t_k+\delta^\prime}e^{\bar{A}(t_k+\delta^\prime-\tau)}\bar{B}w(\tau)d\tau\Bigg)\nonumber\\
&&{+}\:\displaystyle\int_{t_k+\delta^\prime}^{t_k+\delta}e^{\bar{A}(t_k+\delta-\tau)}\bar{B}w(\tau)d\tau.
\end{eqnarray}
\setlength{\arraycolsep}{5pt}Now, by~(\ref{K-02}) and~(\ref{K-03})
\setlength{\arraycolsep}{0.0em}
\begin{eqnarray}\label{K-04}
\Delta x(t_k+\delta)&{=}&\bar{A}_\rho\Delta x(t_k)+(e^{\bar{A}\delta}\bar{A}_\rho-e^{\bar{A}(\delta-\delta^\prime)}\bar{A}_\rho e^{\bar{A}\delta^\prime})\tilde{x}(t_k)\nonumber\\
&&{-}\:e^{\bar{A}(\delta-\delta^\prime)}\bar{A}_\rho\displaystyle\int_{t_k}^{t_k+\delta^\prime}e^{\bar{A}(t_k+\delta^\prime-\tau)}\bar{B}w(\tau)d\tau\nonumber\\
&&{+}\:\displaystyle\int_{t_k}^{t_k+\delta^\prime}e^{\bar{A}(t_k+\delta-\tau)}\bar{B}w(\tau)d\tau\nonumber\\
&&{+}\:(e^{\bar{A}\delta}-I)\bar{A}_\rho\Delta x(t_k)\nonumber\\
&&{=}\;\bar{A}_\rho\Delta x(t_k)+O(\delta,\tilde{x}(t_k),x(t_k)),
\end{eqnarray}
\setlength{\arraycolsep}{5pt}and, using~(\ref{K-01}),
\begin{equation}\label{K-05}
\lim_{k\to\infty}O(\delta,\tilde{x}(t_k),x(t_k))=0.
\end{equation} 
The same discussion applies for $\tilde{t}_k$. Hence, for $t$ sufficiently large we have
\begin{equation}\label{E-307}
\left\{
\begin{aligned}
\Delta \dot{x}(t)&=\bar{A}\Delta x(t), & t\in\mathcal{M},\\
\Delta x(t^+)&=\bar{A}_\rho \Delta x(t), & t\notin\mathcal{M}.\\
\end{aligned}
\right.
\end{equation}  	
Due to the satisfaction of the $H_\beta$ condition \cite{beker2004fundamental,guo2015analysis,hollot2001establishing}, there exist a matrix $P\in\mathbb{R}^{(n_p+n_r)\times(n_p+n_r)},\ P=P^T>0$, and a scalar $\alpha>0$ such that    	
\begin{align}
 P\bar{A}+\bar{A}^TP\leq-2\alpha P, \label{E-308}\\
  \bar{A}_\rho^TP \bar{A}_\rho-P\leq0. \label{E-3081}
\end{align}
Using the candidate Lyapunov function $V(\Delta x)=\dfrac{1}{2}(\Delta x)^TP(\Delta x)$ yields
\begin{equation}\label{E-3089} 
\begin{cases} 
\dot{V}\leq-\alpha V, & t\in\mathcal{M},\\
V(\Delta x(t^+))=V(\Delta x(t))+\Xi(t,\delta), & t\notin\mathcal{M}.\\
\end{cases}
\end{equation}
Thus, using~(\ref{E-307}) and~(\ref{E-3081}) for $t$ sufficiently large yields
\begin{equation}\label{E-3090} 
\Xi(t,\delta)\leq0.
\end{equation}
Hence, since $\Delta x$ is bounded, there exist $\alpha_m>0$ and $\mathcal{K}>0$ such that
\begin{equation}\label{E-309} 
||x_2(t)-x_1(t)||^2_P\leq \mathcal{K}e^{-\alpha_mt},
\end{equation}
for all $t\geq0$ (see Lemma~\ref{AP2} in the Appendix). This implies that the reset control system~(\ref{E-303}) is uniformly exponentially convergent.
\end{proof}
\begin{theorem}\label{T1}
Consider the reset control system~(\ref{E-303}). Suppose Assumption~\ref{AS1} holds, $w(t)=w_0\sin(\omega t)$\footnote{For ease of the notation we consider $w(t)=w_0\sin(\omega t)$. However, Theorem~\ref{T1} is also applicable in the case in which $w(t)=[r_0\sin(\omega_t+\phi_1)\ d_0\sin(\omega_t+\phi_2)]^T$.}, and the $H_\beta$ condition is satisfied. Then the reset control system~(\ref{E-303}) has a periodic steady-state solution which can be expressed as $\bar{x}(t)=\mathcal{S}(\sin(\omega t),\cos(\omega t),\omega)$ for some function $\mathcal{S}:\mathbb{R}^3\rightarrow\mathbb{R}^{n_r+n_p}$.
\end{theorem}
\begin{proof}
Since the $H_\beta$ condition holds and $w(t)=w_0\sin(\omega t)$ is a Bohl function, by Remark~\ref{R0} the reset control system~(\ref{E-303}) has a unique solution for any initial condition $x_0$. In addition, the reset control system (\ref{E-303}) has the UBIBS property and, according to Lemma~\ref{L2}, it is uniformly exponentially convergent. Hence, the proof of the existence of the function $\mathcal{S}$ relies on the results in \cite{pavlov2007frequency}. We only need to show that $\mathcal{S}$ is unique. To this end, similarly to \cite{pavlov2004convergent}, assume that the reset control system~(\ref{E-303}) has two steady-state solutions $\bar{x}_2(t)=\mathcal{S}_2(\sin(\omega t),\cos(\omega t),\omega)(t)$ and $\bar{x}_1(t)=\mathcal{S}_1(\sin(\omega t),\cos(\omega t),\omega)(t)$, for $w(t)=w_0\sin(\omega t)$. Since the $H_\beta$ condition holds, by Lemma~\ref{L2} there exist $\alpha_m>0$ and $\mathcal{K}>0$ such that
\begin{equation}\label{E-309999} 
||\bar{x}_2(t)-\bar{x}_1(t)||^2_P\leq \mathcal{K}e^{-\alpha_mt},
\end{equation}
hence, the claim.
\end{proof}	
\begin{corollary}\label{co2}
Consider the reset control system~(\ref{E-303}) with $r(t)=r_0\sin(\omega t)$ and $d=0$, for all $t\geq0$. Then the even harmonics and the subharmonics of the steady-state response have zero amplitude, and the sequence of reset instants is periodic with period $\dfrac{\pi}{\omega}$.	
\end{corollary}
\begin{proof}
The response of~(\ref{E-303}) for $r=r_0\sin(\omega t)$ and $d=0$, for all $t\geq0$, is given by
\begin{equation}\label{E-314}
\begin{array}{*{35}{c}}
x(t)=r_0\left(e^{\bar{A}(t-t_k)}\Big(\xi_k+\psi(t_k)\Big)-\psi(t)\right), &  t\in(t_k,t_{k+1}],
\end{array}
\end{equation}
where
\setlength{\arraycolsep}{0.0em}
\begin{eqnarray}\label{E-315}
\psi(t)&{=}&(\omega I\cos(\omega t)+\bar{A}\sin(\omega t))\mathcal{F},\nonumber\\
\mathcal{F}&{=}&(\omega^2I+\bar{A}^2)^{-1}\bar{B}\begin{bmatrix}1\\0\end{bmatrix},\nonumber\\
t_k&{=}&\{t_k\in\mathbb{R}^{+},\ k\in\mathbb{Z}^{+}\mid e_R(t_k)=0\},\nonumber\\
\xi_k&{=}&\dfrac{1}{r_0}x(t_k^{+})=\dfrac{1}{r_0}\bar{A}_\rho x(t_k).
\end{eqnarray}
\setlength{\arraycolsep}{5pt}Thus 
\begin{equation}\label{E-316}
\begin{array}{*{35}{c}}
\bar{x}(t)=r_0\left(e^{\bar{A}(t-t_s)}\Big(\xi_s+\psi(t_s)\Big)-\psi(t)\right), &  t\in(t_s,t_{s+1}],
\end{array}
\end{equation}	
with
\setlength{\arraycolsep}{0.0em}
\begin{eqnarray}\label{E-317}
\xi_{s}&{=}&\resizebox{0.94\hsize}{!}{$\bar{A}_\rho e^{\bar{A}(t_s-t_{s-1})}\Bigg(\bar{A}_\rho e^{\bar{A}(t_{s-1}-t_{s-2})}\dots\ \bar{A}_\rho e^{\bar{A}(t_{1}-t_{0})}(\xi_0+\psi(t_0))$}\nonumber\\
&&{+}\:\bar{A}_\rho e^{\bar{A}(t_{s-1}-t_{s-2})}\dots\ \bar{A}_\rho e^{\bar{A}(t_2-t_1)}(I-\bar{A}_\rho)\psi(t_1)\nonumber\\
&&{+}\:\bar{A}_\rho e^{\bar{A}(t_{s-1}-t_{s-2})}\dots\ \bar{A}_\rho e^{\bar{A}(t_2-t_1)}(I-\bar{A}_\rho)\psi(t_2)\nonumber\\
&&{+}\:\dots+(I-\bar{A}_\rho)\psi(t_{s-1})\Bigg)-\bar{A}_\rho\psi(t_s).
\end{eqnarray}
\setlength{\arraycolsep}{5pt}According to \cite{pavlov2005convergent}, uniformly convergent systems forget their initial conditions. By Lemma~\ref{L1} and Lemma~\ref{L2}, $\xi_s$ and the reset instants are unique for any $t_0$ and $\zeta_0$. Hence, the transient response of $\xi_s$ converges to zero which implies that 
\setlength{\arraycolsep}{0.0em}
\begin{eqnarray}\label{E-317}
\xi_{s}&{=}&\bar{A}_\rho e^{\bar{A}(t_s-t_{s-1})}\Bigg((I-\bar{A}_\rho)\psi(t_{s-1})\nonumber\\
&&{+}\:\bar{A}_\rho e^{\bar{A}(t_{s-1}-t_{s-2})}(I-\bar{A}_\rho)\psi(t_{s-2})\nonumber\\
&&{+}\:\bar{A}_\rho e^{\bar{A}(t_{s-1}-t_{s-2})}\bar{A}_\rho e^{\bar{A}(t_{s-2}-t_{s-3})}(I-\bar{A}_\rho)\psi(t_{s-3})\nonumber\\
&&{+}\:\dots+\bar{A}_\rho e^{\bar{A}(t_{s-1}-t_{s-2})}\dots\ \bar{A}_\rho e^{\bar{A}(t_{s-m+1}-t_{s-m})}\nonumber\\
&&{}\,(I-\bar{A}_\rho)\psi(t_{s-m})\Bigg)-\bar{A}_\rho\psi(t_s).
\end{eqnarray}
\setlength{\arraycolsep}{5pt}Therefore, since reset occurs when 
\begin{equation}\label{E-320} 
\bar{C}_{e_R}\bar{x}(t)+D_{e_R}r_0\sin(\omega t)=0,
\end{equation}
if $\{t_s,t_{s-1},...,t_{s-m}\}$ are reset instants and satisfy~(\ref{E-320}), then $\{t_s,t_{s-1},...,t_{s-m}\}+\dfrac{\pi}{\omega}$ also satisfy~(\ref{E-320}), which implies that the sequence of reset instants is periodic with period $\dfrac{\pi}{\omega}$. Using this property in (\ref{E-316}) shows that $\bar{x}(t)=-\bar{x}(t+\dfrac{\pi}{\omega})$ and $t_{s+q}-t_s=\dfrac{\pi}{\omega}$, hence $\xi_s=-\xi_{s+q}$. This means that the even harmonics of the steady-state response of the reset control system~(\ref{E-303}) have zero amplitude. In addition, $\bar{x}(t)=\bar{x}(t+\dfrac{2\pi}{\omega})$, which implies that the steady-state response of the reset control system~(\ref{E-303}) does not contain any subharmonic.  
\end{proof}
\begin{remark}\label{RR2}
{\rm The reset sequence $\{t_k\}$ and the reset values $\zeta_k$ are independent of the input amplitude for $r(t)=r_0\sin(\omega_t)$}.
\end{remark}
We now show that the function $\mathcal{S}$ can be derived explicitly for $r(t)=r_0\sin(\omega_t)$ and $d=0$. Suppose there are $q-1$ reset instants between $t_s$ and $t_s+\dfrac{\pi}{\omega}$ (Fig.~\ref{F-32}). Assume $\sin(\omega t_s)=\kappa$, then $\cos(\omega t_s)=\pm\sqrt{1-\kappa^2}$ (without loss of generality we consider the positive value). 
%%%%%%%%%%%%%%%%%%%%%%FIGURE
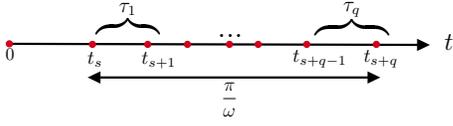
\begin{figure}
	\centering
\resizebox{0.7\hsize}{!}{
	\tikzset{every picture/.style={line width=0.75pt}} %set default line width to 0.75pt        
	\begin{tikzpicture}[x=0.75pt,y=0.75pt,yscale=-1,xscale=1]
	%uncomment if require: \path (0,127); %set diagram left start at 0, and has height of 127
	
	%Straight Lines [id:da6055645258040092] 
	\draw [line width=1.5]    (12.38,36.5) -- (391.88,37.49) ;
	\draw [shift={(395.88,37.5)}, rotate = 180.15] [fill={rgb, 255:red, 0; green, 0; blue, 0 }  ][line width=0.08]  [draw opacity=0] (11.61,-5.58) -- (0,0) -- (11.61,5.58) -- cycle    ;
	%Shape: Circle [id:dp8739736620649443] 
	\draw  [color={rgb, 255:red, 208; green, 2; blue, 27 }  ,draw opacity=1 ][fill={rgb, 255:red, 208; green, 2; blue, 27 }  ,fill opacity=1 ] (87.63,37) .. controls (87.63,35.48) and (88.86,34.25) .. (90.38,34.25) .. controls (91.89,34.25) and (93.13,35.48) .. (93.13,37) .. controls (93.13,38.52) and (91.89,39.75) .. (90.38,39.75) .. controls (88.86,39.75) and (87.63,38.52) .. (87.63,37) -- cycle ;
	%Shape: Circle [id:dp3355635229100672] 
	\draw  [color={rgb, 255:red, 208; green, 2; blue, 27 }  ,draw opacity=1 ][fill={rgb, 255:red, 208; green, 2; blue, 27 }  ,fill opacity=1 ] (12.38,36.5) .. controls (12.38,34.98) and (13.61,33.75) .. (15.13,33.75) .. controls (16.64,33.75) and (17.88,34.98) .. (17.88,36.5) .. controls (17.88,38.02) and (16.64,39.25) .. (15.13,39.25) .. controls (13.61,39.25) and (12.38,38.02) .. (12.38,36.5) -- cycle ;
	%Shape: Circle [id:dp2018554613057928] 
	\draw  [color={rgb, 255:red, 208; green, 2; blue, 27 }  ,draw opacity=1 ][fill={rgb, 255:red, 208; green, 2; blue, 27 }  ,fill opacity=1 ] (137.63,37) .. controls (137.63,35.48) and (138.86,34.25) .. (140.38,34.25) .. controls (141.89,34.25) and (143.13,35.48) .. (143.13,37) .. controls (143.13,38.52) and (141.89,39.75) .. (140.38,39.75) .. controls (138.86,39.75) and (137.63,38.52) .. (137.63,37) -- cycle ;
	%Shape: Circle [id:dp3747751642294964] 
	\draw  [color={rgb, 255:red, 208; green, 2; blue, 27 }  ,draw opacity=1 ][fill={rgb, 255:red, 208; green, 2; blue, 27 }  ,fill opacity=1 ] (281.63,37) .. controls (281.63,35.48) and (282.86,34.25) .. (284.38,34.25) .. controls (285.89,34.25) and (287.13,35.48) .. (287.13,37) .. controls (287.13,38.52) and (285.89,39.75) .. (284.38,39.75) .. controls (282.86,39.75) and (281.63,38.52) .. (281.63,37) -- cycle ;
	%Shape: Circle [id:dp6310567754339504] 
	\draw  [color={rgb, 255:red, 208; green, 2; blue, 27 }  ,draw opacity=1 ][fill={rgb, 255:red, 208; green, 2; blue, 27 }  ,fill opacity=1 ] (344.63,37) .. controls (344.63,35.48) and (345.86,34.25) .. (347.38,34.25) .. controls (348.89,34.25) and (350.13,35.48) .. (350.13,37) .. controls (350.13,38.52) and (348.89,39.75) .. (347.38,39.75) .. controls (345.86,39.75) and (344.63,38.52) .. (344.63,37) -- cycle ;
	%Straight Lines [id:da33275883361089675] 
	\draw [line width=1.5]    (90,67) -- (347,67) ;
	\draw [shift={(351,67)}, rotate = 180] [fill={rgb, 255:red, 0; green, 0; blue, 0 }  ][line width=0.08]  [draw opacity=0] (11.61,-5.58) -- (0,0) -- (11.61,5.58) -- cycle    ;
	\draw [shift={(86,67)}, rotate = 0] [fill={rgb, 255:red, 0; green, 0; blue, 0 }  ][line width=0.08]  [draw opacity=0] (11.61,-5.58) -- (0,0) -- (11.61,5.58) -- cycle    ;
	%Shape: Circle [id:dp955078555042085] 
	\draw  [color={rgb, 255:red, 208; green, 2; blue, 27 }  ,draw opacity=1 ][fill={rgb, 255:red, 208; green, 2; blue, 27 }  ,fill opacity=1 ] (173.63,37) .. controls (173.63,35.48) and (174.86,34.25) .. (176.38,34.25) .. controls (177.89,34.25) and (179.13,35.48) .. (179.13,37) .. controls (179.13,38.52) and (177.89,39.75) .. (176.38,39.75) .. controls (174.86,39.75) and (173.63,38.52) .. (173.63,37) -- cycle ;
	%Shape: Circle [id:dp08106725916221469] 
	\draw  [color={rgb, 255:red, 208; green, 2; blue, 27 }  ,draw opacity=1 ][fill={rgb, 255:red, 208; green, 2; blue, 27 }  ,fill opacity=1 ] (211.63,37) .. controls (211.63,35.48) and (212.86,34.25) .. (214.38,34.25) .. controls (215.89,34.25) and (217.13,35.48) .. (217.13,37) .. controls (217.13,38.52) and (215.89,39.75) .. (214.38,39.75) .. controls (212.86,39.75) and (211.63,38.52) .. (211.63,37) -- cycle ;
	%Shape: Circle [id:dp7648092602769468] 
	\draw  [color={rgb, 255:red, 208; green, 2; blue, 27 }  ,draw opacity=1 ][fill={rgb, 255:red, 208; green, 2; blue, 27 }  ,fill opacity=1 ] (237.63,37) .. controls (237.63,35.48) and (238.86,34.25) .. (240.38,34.25) .. controls (241.89,34.25) and (243.13,35.48) .. (243.13,37) .. controls (243.13,38.52) and (241.89,39.75) .. (240.38,39.75) .. controls (238.86,39.75) and (237.63,38.52) .. (237.63,37) -- cycle ;
	
	% Text Node
	\draw (413,35) node  [scale=1.2,font=\Large]  {$t$};
	% Text Node
	\draw (10,39) node [anchor=north west][inner sep=0.75pt][scale=1.2]    {$0$};
	% Text Node
	\draw (84,42) node [anchor=north west][inner sep=0.75pt][scale=1.3]    {$t_{s}$};
	% Text Node
	\draw (92,-5) node [anchor=north west][inner sep=0.75pt]  [scale=1.2,font=\LARGE]  {$\overbrace{\ \ \ \ \ \ }^{\tau _{1}}$};
	% Text Node
	\draw (202,26) node [anchor=north west][inner sep=0.75pt]  [scale=1.2,font=\LARGE]  {$...$};
	% Text Node
	\draw (133,42) node [anchor=north west][inner sep=0.75pt][scale=1.3]    {$t_{s+1}$};
	% Text Node
	\draw (270,40) node [anchor=north west][inner sep=0.75pt][scale=1.3]    {$t_{s+q-1}$};
	% Text Node
	\draw (334,40) node [anchor=north west][inner sep=0.75pt][scale=1.3]    {$t_{s+q}$};
	% Text Node
	\draw (290,-6) node [anchor=north west][inner sep=0.75pt]  [scale=1.2,font=\LARGE]  {$\overbrace{\ \ \ \ \ \ \ }^{\tau _{q}}$};
	% Text Node
	\draw (206,70) node [anchor=north west][inner sep=0.75pt]  [scale=0.9,font=\Large]  {$\dfrac{\mathbf{\pi }}{\mathbf{\omega }}$};
	\end{tikzpicture}}
	\caption{Steady-state reset instants of the reset control system~(\ref{E-303})}
	\label{F-32}
\end{figure}  
Using trigonometry relations, one has that
\setlength{\arraycolsep}{0.0em}
\begin{eqnarray}\label{E-32333}
\psi(t_s)&=&f_0(\kappa),\nonumber\\
\psi(t_s+\tau_1)&=&f_1(\kappa,\tau_1),\nonumber\\ 
&\ \vdots\nonumber\\
\psi(t_s+\tau_1+...+\tau_q)&=&f_q(\kappa,\tau_1,\tau_2,...,\tau_q).
\end{eqnarray}
\setlength{\arraycolsep}{5pt}Moreover,
\setlength{\arraycolsep}{0.0em}
\begin{eqnarray}\label{E-323}
\xi_{s+i}&{=}&\bar{A}_\rho\Bigg(e^{\bar{A}\tau_i}(g_{i-1}(\kappa,\xi_s,\tau_1,..,\tau_{i-1})+f_{i-1}(\kappa,\tau_1,..,\tau_{i-1}))\nonumber\\
&&{-}\:f_i(\kappa,\tau_1,\tau_2,...,\tau_i)\Bigg)=g_i(\kappa,\xi_s,\tau_1,\tau_2,...,\tau_i),
\end{eqnarray}
\setlength{\arraycolsep}{5pt}with $i=1,2,...,q$ and $g_0(\kappa,\xi_s)=\xi_s$. Now, since $e_{R}(t)$ is zero at reset instants, one has that
\setlength{\arraycolsep}{0.0em}
\begin{eqnarray}\label{E-324}
&{}&\bar{C}_{e_R}\Bigg(e^{\bar{A}\tau_i}(g_{i-1}(\kappa,\xi_s,\tau_1,..,\tau_{i-1})+f_{i-1}(\kappa,\tau_1,..,\tau_{i-1}))\nonumber\\
&{-}&f_i(\kappa,\tau_1,\tau_2,...,\tau_i)\Bigg)+D_{eR}\sin(\omega(t_s+\tau_1+...+\tau_i) )\nonumber\\
&{=}&E_i(\kappa,\xi_s,\tau_1,...,\tau_i)=0,
\end{eqnarray}
\setlength{\arraycolsep}{5pt}with $i=1,2,...,q$. In addition,
\begin{equation}\label{E-325}
\begin{array}{*{35}{c}}
\tau_1+\tau_2+...+\tau_q=\dfrac{\pi}{\omega},\\
\xi_s=-\xi_{s+q}\Rightarrow g_q(\kappa,\xi_s,\tau_1,\tau_2,...,\tau_q)+\xi_s=0.
\end{array}
\end{equation}
Moreover, by the well-posedness property of reset instants (see Definition~\ref{DD0}), reset instants are distinct. Hence, there are $q+2$ independent equations and $q+2$ parameters $(\kappa,\xi_s,q,\tau_1,\tau_2,...,\tau_q),\ q\in\mathbb{N}$. In addition, the well-posedness property implies that the reset intervals are lower bounded \cite{banos2011reset}. Hence, 
\begin{equation}\label{E-326}
\exists\ \lambda\leq \tau_i\Rightarrow q\leq\dfrac{\pi}{\lambda\omega}-1.
\end{equation}
Furthermore, for $q=1$, the equations have always a unique solution. Thus, there exists a bounded non-empty set $Q=\{Q_i\in\mathbb{N}|Q_i\leq q_{\max}\}$ such that for $q\in Q$, the equations have a solution. Hence, $\bar{x}(t)$, the steady-state response of the reset control system~(\ref{E-303}) to $r(t)=r_0\sin(\omega t)$, is the solution of (\ref{E-324})-(\ref{E-325}) for $q=q_{\max}$. Since $\bar{x}(t)$ is periodic with period $\dfrac{2\pi}{\omega}$, one has
\begin{equation}\label{E-405} 
\bar{x}(t)=\mathlarger{\sum}\limits ^{\infty}_{n=1}a_n\cos((2n+1)\omega t)+b_n\sin((2n+1)\omega t).
\end{equation}
According to Theorem~\ref{T1}, $\bar{x}$ is unique and equal to the function $\mathcal{S}$. Thus,
\setlength{\arraycolsep}{0.0em}
\begin{eqnarray}\label{E-4055}
\bar{x}(t)&{=}&\mathlarger{\sum}\limits ^{\infty}_{n=1}a_n\cos((2n+1)\omega t)+b_n\sin((2n+1)\omega t)\nonumber\\
&{=}&\mathcal{S}(\sin(\omega t),\cos(\omega t),\omega).
\end{eqnarray}
\setlength{\arraycolsep}{5pt}Finally, one could also use De Moivre's formula to find a formal polynomial expansion for $\mathcal{S}$ in terms of $\sin(\omega t)$ and $\cos(\omega t)$.
%%%%%%%%%%%%%%%%%%%%%%%%%%%%%%%%%%%%%%%%%%%%%%HOSIDF of closed-loop
\subsection{HOSIDF of The Closed-Loop Reset Control Systems}\label{sec:32}
In Section~\ref{sec:31} sufficient conditions for the existence of the steady-state solution for the reset control system~(\ref{E-303}) driven by periodic inputs have been presented. Moreover, the steady-state solution has been explicitly calculated. In this section the HOSIDF technique \cite{nuij2006higher} is applied to the steady-state response of the system to derive a notion of frequency response for the reset control system~(\ref{E-303}), which allows analyzing tracking and disturbance rejection performance (see the bottom diagram of Fig.~\ref{F-31}).
%%%%%%%%%%%%Reference tracking analysis
\subsubsection{Tracking Performance}\label{sec:321}
Consider the reset control system~(\ref{E-303}) with $r(t)=r_0\sin(\omega t)$ and $d(t)=0$, for all $t\geq0$. We now derive relations between the input $r(t)$ and the steady-state response of the output $y$, of the error $e$, and of the control input $u$. To this end, consider the steady-state reset instants $t_s,t_{s+1},...,t_{s+q}$ and their associated reset values $\xi_s,\xi_{s+1},...,\xi_{s+q}$ which are calculated through (\ref{E-324}) and (\ref{E-325}).
\begin{theorem}\label{thm1}
	Consider the reset control system~(\ref{E-303}) with $r(t)=r_0\sin(\omega t)$ and $d(t)=0$, for all $t\geq0$. Let $T_n(j\omega)$ be the ratio of the $n^\text{th}$ harmonic component of the output signal $y$ to the first harmonic component of $r$. Then 
	\begin{equation}\label{E-327}
%\resizebox{\hsize}{!}{$	
T_n(j\omega)=\begin{cases}
		\mathfrak{T}(1,\omega)-\bar{C}(j\omega I+\bar{A})\mathcal{F}, &n=1,\\
		\mathfrak{T}(n,\omega), & n>1\text{ odd,}\\
		0, & n\text{ even,}
		\end{cases}
	\end{equation}
	in which
	\begin{equation}\label{E-3278}	
	\mathfrak{T}(n,\omega)=\dfrac{2j\omega \bar{C}}{\pi}(\bar{A}-jn\omega I)^{-1}(\mathlarger{\sum}\limits ^{q}_{i=1}\mathcal{R}(i,n,\omega)),
	\end{equation}  
\setlength{\arraycolsep}{0.0em}
\begin{eqnarray}\label{E-328}
\mathcal{R}(i,n,\omega)&{=}&\left(\dfrac{e^{\bar{A}(t_{s+i}-t_{s+i-1})}}{e^{jn\omega t_{s+i}}}-\dfrac{I}{e^{jn\omega t_{s+i-1}}}\right)\Bigg(\xi_{s+i-1}+\nonumber\\
&&{+}\:\psi(t_{s+i-1})\Bigg).
\end{eqnarray}
\setlength{\arraycolsep}{5pt}\end{theorem}
\begin{proof}
By \cite{guo2009frequency,nuij2006higher}
\begin{equation}\label{E-328-329}
	T_n(j\omega)=\dfrac{\mathlarger{\int}_{t_{s}}^{t_{s}+\frac{2\pi}{\omega}}y(t)e^{-jn\omega t}dt}{\mathlarger{\int}_{t_{s}}^{t_{s}+\frac{2\pi}{\omega}}r_0\sin(\omega t)e^{-j\omega t}dt}.
	\end{equation}
	Using (\ref{E-316}), (\ref{E-328-329}) is rewritten as
	\setlength{\arraycolsep}{0.0em}
\begin{eqnarray}\label{E-329}
T_n(j\omega)&{=}&\dfrac{j\omega \bar{C}}{\pi}\Bigg(\mathlarger{\sum}\limits ^{2q}_{i=1}\big(\mathlarger{\int}_{t_{s+i-1}}^{t_{s+{i}}}\mathcal{X}_{i}(t)e^{-jn\omega t}dt\big)\nonumber\\
&&{-}\:\mathlarger{\int}_{t_{s}}^{t_{s}+\frac{2\pi}{\omega}}\psi(t)e^{-jn\omega t}dt\Bigg),
\end{eqnarray}
\setlength{\arraycolsep}{5pt}where 
	\begin{equation}\label{E-330}
	\mathcal{X}_i(t)=e^{\bar{A}(t-t_{s+{i-1}})}\Big(\xi_{{s+{i-1}}}+\psi(t_{s+{i-1}})\Big).
	\end{equation}
	For $n$ even the first part of (\ref{E-329}) is zero by Corollary \ref{co2}, while for $n$ odd one has
\setlength{\arraycolsep}{0.0em}
\begin{eqnarray}\label{E-331}
&{}&\mathlarger{\int}_{t_{s+{i-1}}}^{t_{s+{i}}}\mathcal{X}_{i}(t)e^{-jn\omega t}dt\nonumber\\
&{=}&\mathlarger{\int}_{t_{s+{i-1}}}^{t_{s+{i}}}e^{\bar{A}(t-t_{s+{i-1}})}\Big(\xi_{{s+{i-1}}}+\psi(t_{s+{i-1}})\Big)e^{-jn\omega t}dt\nonumber\\
&{=}&\mathlarger{\int}_{t_{s+{i-1}}+\frac{\pi}{\omega}}^{t_{s+{i}}+\frac{\pi}{\omega}}e^{\bar{A}(t-t_{s+{i-1}})}\Big(-\xi_{{s+{i-1}}}-\psi(t_{s+{i-1}})\Big)\dfrac{e^{-jn\omega t}}{e^{-jn\pi}}dt\nonumber\\
&{=}&\mathlarger{\int}_{t_{s+{i-1}}+\frac{\pi}{\omega}}^{t_{s+{i}}+\frac{\pi}{\omega}}\mathcal{X}_{i}(t)e^{-jn\omega t}dt=(A-jn\omega I)^{-1} \mathcal{R}_{(i,n)},
\end{eqnarray}
\setlength{\arraycolsep}{5pt}while the second term in (\ref{E-329}) is given by  
	\begin{equation}\label{E-332}
	\mathlarger{\int}_{t_{ss_0}}^{t_{ss_m}}\psi(t)e^{-jn\omega dt}=\begin{cases}
	\pi(I-\dfrac{j\bar{A}}{\omega})\mathcal{F}, & n=1,\\
	0, & n>2.
	\end{cases}
	\end{equation}
Thus, substituting (\ref{E-331}) and (\ref{E-332}) to (\ref{E-329}) yields the claim.
\end{proof} 
\begin{definition}\label{D1}
The family of complex valued functions $T_n(j\omega)$, $n=1,2,...$ is the complementary sensitivity of the reset control system (\ref{E-303}). 
\end{definition}
\begin{corollary}\label{coro1}
Consider the reset control system (\ref{E-303}) with $r(t)=r_0\sin(\omega t)$ and $d(t)=0$, for all $t\geq0$. Let $S_n(j\omega)$ be the ratio of the $n^\text{th}$ harmonic component of the error signal $e$ to the first harmonic component of $r$. Then
	\begin{equation}\label{E-333}
	%\resizebox{\hsize}{!}{$
	S_n(j\omega)+T_n(j\omega)=\begin{cases}
		1, & n=1,\\
		0, & n>1.
		\end{cases}
		%$}
\end{equation}
\end{corollary}
\begin{corollary}\label{coro11}
Consider the reset control system (\ref{E-303}) with $r(t)=r_0\sin(\omega t)$ and $d(t)=0$, for all $t\geq0$. Let $CS_n(j\omega)$ be the ratio of the $n^\text{th}$ harmonic component of the control input signal $u$ to the first harmonic component of $r$. If the plant is stable, then 
\begin{equation}\label{E-334}
CS_n(j\omega)=\dfrac{T_n(j\omega)}{G(nj\omega)}.
\end{equation}
\end{corollary}
\begin{definition}\label{D1}
The families of complex valued functions $S_n(j\omega)$ and $CS_n(j\omega)$, $n=1,2,...$, are the sensitivity and the control sensitivity of the reset control system (\ref{E-303}), respectively. 
\end{definition}
%%%%%%%%%%%%%%%%%%%Disturbance Rejection analysis
\subsubsection{Disturbance Rejection}\label{sec:322}
In this section relations between $d(t)=\sin(\omega t)$ and the error $e(t)$ and the control input $u(t)$ are found in the case in which $r(t)=0$ for the reset control system~(\ref{E-303}) using the same procedure provided in Section \ref{sec:321}. The matrix $\psi(t)$ has to be replaced by
\setlength{\arraycolsep}{0.0em}
\begin{eqnarray}\label{E-336}
\psi_\mathcal{D}(t)&{=}&(\omega I\cos(\omega t)+\bar{A}\sin(\omega t))\mathcal{F}_\mathcal{D},\nonumber\\
 \mathcal{F}_\mathcal{D}&{=}&(\omega^2I+\bar{A}^2)^{-1}\bar{B}\begin{bmatrix}0\\1\end{bmatrix}.
\end{eqnarray}
\setlength{\arraycolsep}{5pt}Let $t^\prime_s,t^\prime_{s+1},...,t^\prime_{s+q^\prime}$ and $\xi^\prime_s,\xi^\prime_{s+1},...,\xi^\prime_{s+q^\prime}$ be the steady-state reset instants and their associated reset values for the reset control system (\ref{E-303}) with $d(t)=d_0\sin(\omega t)$ and $r(t)=0$, respectively. In addition, since $r(t)=0$, (\ref{E-324}) is changed to  
\setlength{\arraycolsep}{0.0em}
\begin{eqnarray}\label{E-335}
&{}&\bar{C}_{e_R}\Bigg(e^{\bar{A}\tau^\prime_i}(g_{i-1}(\kappa^\prime,\xi^\prime_s,\tau^\prime_1,..,\tau^\prime_{i-1})+f_{i-1}(\kappa^\prime,\tau^\prime_1,..,\tau^\prime_{i-1}))\nonumber\\
&{-}&f_i(\kappa^\prime,\tau^\prime_1,\tau^\prime_2,...,\tau^\prime_i)\Bigg)=E_i(\kappa^\prime,\xi^\prime_s,\tau^\prime_1,...,\tau^\prime_i)=0,
\end{eqnarray}
\setlength{\arraycolsep}{5pt}with $i=1,2,...,q^\prime$. Now, substituting $\psi(t)$ with $\psi_\mathcal{D}(t)$ in relations (\ref{E-32333}) and (\ref{E-323}), and considering (\ref{E-335}) instead of (\ref{E-324}), the steady-state response of the reset control system (\ref{E-303}) for $d(t)=d_0\sin(\omega t)$ and $r(t)=0$ can be found using the same procedure provided in Section \ref{sec:31}.
\begin{corollary}\label{coro3}
Consider the reset control system (\ref{E-303}) with $d(t)=d_0\sin(\omega t)$ and $r(t)=0$, for all $t\geq0$. Let $PS_n(j\omega)$ be the ratio of the $n^\text{th}$ harmonic component of the error signal $e$ to the first harmonic component of $d$. Then
\begin{equation}\label{E-337}
	PS_n(j\omega)=\begin{cases}
		\mathfrak{P}(1,\omega)+\bar{C}(j\omega I+\bar{A})\mathcal{F}_\mathcal{D}, & n=1,\\
		\mathfrak{P}(n,\omega), & n>1\text{ odd,}\\
		0, & n\text{ even,}
		\end{cases}
	\end{equation}
in which
\begin{equation}\label{E-3278}	
	\mathfrak{P}(n,\omega)=\dfrac{2j\omega \bar{C}}{\pi}(jn\omega I-\bar{A})^{-1}(\mathlarger{\sum}\limits ^{q^\prime}_{i=1}\mathcal{R}_\mathcal{D}(i,n,\omega)),
	\end{equation}  
\setlength{\arraycolsep}{0.0em}
\begin{eqnarray}\label{E-328}
\mathcal{R}_\mathcal{D}(i,n,\omega)&{=}&\left(\dfrac{e^{\bar{A}(t^\prime_{s+i}-t^\prime_{s+i-1})}}{e^{jn\omega t^\prime_{s+i}}}-\dfrac{I}{e^{jn\omega t^\prime_{s+i-1}}}\right)\Bigg(\xi^\prime_{s+i-1}+\nonumber\\
&&{+}\:\psi_\mathcal{D}(t^\prime_{s+i-1})\Bigg).
\end{eqnarray}
\setlength{\arraycolsep}{5pt}\end{corollary}
\begin{corollary}\label{coro111}
Consider the reset control system (\ref{E-303}) with $d(t)=d_0\sin(\omega t)$ and $r(t)=0$, for all $t\geq0$. Let $CS_{d_n}(j\omega)$ be the ratio of the $n^\text{th}$ harmonic component of the control input signal $u$ to the first harmonic component of $d$. If the plant is stable, then 
\begin{equation}\label{E-338}
	CS_{d_n}(j\omega)=\begin{cases}
	\dfrac{-PS_1(j\omega)}{G(j\omega)}-1, & n=1,\\
	\dfrac{-PS_n(j\omega)}{G(nj\omega)}, & n>1.\\  
	\end{cases}
	\end{equation}
\end{corollary}
\begin{definition}
The families of complex valued functions $PS_n(j\omega)$ and $CS_{d_n}(j\omega)$, $n=1,2,...$, are the process-sensitivity and the control sensitivity due to the presence of the disturbance of the reset control system (\ref{E-303}), respectively. 
\end{definition}
%%%%%%%%%%%%%%%%%%%%%%%%%%%%%%%%%%%%%%%%Re-defined sensitivity frequency responses for reset elements
\subsection{Pseudo-Sensitivities For Reset Control Systems}\label{sec:3.4}
The analysis of the error signal $e$ and of the control input $u$ is one of the main factors while designing a controller. In linear systems this analysis is performed using the closed-loop transfer functions \cite{schmidt2014design}. As discussed in Section \ref{sec:1}, although reset control systems may be analyzed using the DF of the reset controller in the closed-loop sensitivity equations, this yields an approximation which is not precise due to the existence of high order harmonics. On the other hand, it is not trivial to analyze reset controllers considering all harmonics. In order to perform the analysis of reset control systems straightforwardly we combine all harmonics into one frequency function for each closed-loop frequency response. In the literature, there are several studies about definition of Bode plot for non-linear systems \cite{beker2002analysis,pavlov2006frequency}. However, all of these focus only on the gain of the system. In the following, pseudo-sensitivities, which have both gain and phase components, are defined.

It has been proven that the error and the control input signals of the reset control system (\ref{E-303}) are periodic with period $\frac{2\pi}{\omega}$ (Fig.~\ref{F-33}). 
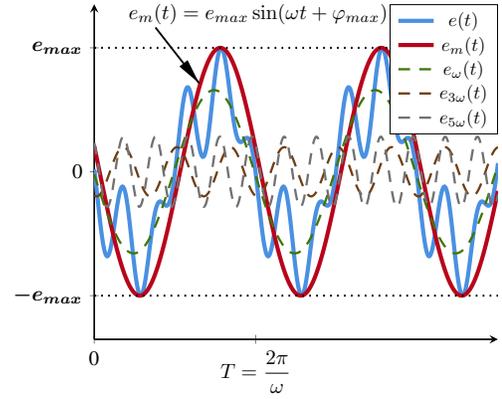
\begin{figure}
	\begin{center}
		\resizebox{.75\hsize}{!}{
			\begin{tikzpicture}
			\begin{axis}[
			axis lines = left,
			axis line style = thick,
			ymin=-0.18,
			ymax=0.18,
			ytick={-0.1331,0,0.1331},
			xtick={0.00693,0.0808},
			yticklabels={$\boldsymbol{-e_{max}}$,0,$\boldsymbol{e_{max}}$},
			xticklabels={0,$T=\dfrac{2\pi}{\omega}$},
			legend style={at={(1,1)},nodes={scale=0.9}}
			]
			%e(t)
			\addplot [
			domain=0.006930:0.191605, 
			samples=500, 
			line width=2,
			color={rgb, 255:red, 74; green, 144; blue, 226 },
			]
			{0.0377*sin((426.5*x-0.78)*180/3.14)+0.02645*sin((255.9*x-3.606)*180/3.14)-0.0877*sin((85.3*x-0.56)*180/3.14)};
			\addlegendentry{$e(t)$}
			
			%e_m(t)
			\addplot [
			domain=0.006930:0.191605,  
			samples=500, 
			line width=2,
			color={rgb, 255:red, 184; green, 4; blue, 27 },
			]
			{-0.1331*sin((85.3*x-0.7996)*180/3.14)};
			\addlegendentry{$e_m(t)$}
			%ew(t)
			\addplot [
			domain=0.006930:0.191605, 
			samples=500, 
			line width=1.1,
			dash pattern={on 5.63pt off 4.5pt},
			color={rgb, 255:red, 65; green, 117; blue, 5 },
			]
			{-0.0877*sin((85.3*x-0.56)*180/3.14)};
			\addlegendentry{$e_\omega(t)$}
			%e3w(t)
			\addplot [
			domain=0.006930:0.191605, 
			samples=500, 
			line width=1.1,
			dash pattern={on 5.63pt off 4.5pt},
			color={rgb, 255:red, 111; green, 62; blue, 20 },
			]
			{0.02645*sin((255.9*x-3.606)*180/3.14)};
			\addlegendentry{$e_{3\omega}(t)$}
			%e5w(t)
			\addplot [
			domain=0.006930:0.191605, 
			samples=500, 
			line width=1.1,
			dash pattern={on 5.63pt off 4.5pt},
			color={rgb, 255:red, 114; green, 113; blue, 113 },
			]
			{0.0377*sin((426.5*x-0.78)*180/3.14)};
			\addlegendentry{$e_{5\omega}(t)$}
			%emax
			\addplot [
			domain=0.006930:0.15467, 
			samples=500, 
			line width=1,
			dash pattern={on 0.84pt off 2.51pt},
			color={rgb, 255:red, 0; green, 0; blue, 0 },
			]
			{0.1331};
			%-emax
			\addplot [
			domain=0.006930:0.191605, 
			samples=500, 
			line width=1,
			dash pattern={on 0.84pt off 2.51pt},
			color={rgb, 255:red, 0; green, 0; blue, 0 },
			]
			{-0.1331};
			% Text Node
			\draw (75,350) node [scale=1]  {$e_m(t)=e_{max}\sin(\omega t+\varphi_{max})$};
			%Straight Lines [id:da19893626905818396] 
			\draw [color={rgb, 255:red, 0; green, 0; blue, 0 }  ,draw opacity=1 ][line width=0.9]    (28,335) -- (48,275) ;
			\draw [shift={(48,275)}, rotate = 128] [fill={rgb, 255:red, 0; green, 0; blue, 0 }  ,fill opacity=1 ][line width=0.9]  [draw opacity=0] (11.61,-5.58) -- (0,0) -- (11.61,5.58) -- cycle    ;
			\end{axis} 
			\end{tikzpicture}}
		\caption{The error signal $e(t)$ with its 1st, 3rd, and 5th harmonics. $e_m(t)$ is fitted to $e(t)$ and it is an indicator of the maximum error of the system.}  % width is 8.4 cm.
		\label{F-33}                                 % Size the figures 
	\end{center}                                 % accordingly.
\end{figure}
We define the pseudo-sensitivity as the ratio of the maximum error of the reset control system (\ref{E-303}), for $r(t)=r_0\sin(\omega t)$ and $d(t)=0$, for all $t\geq0$, to the amplitude of the reference at each frequency.
\begin{definition}\label{d1}
The Pseudo-sensitivity $S_\infty$ is, for all $\omega\in\mathbb{R}^{+}$,
	$$S_\infty(j\omega)=e_{\max}(\omega)e^{j\varphi_{\max}(\omega)},$$
	where $\varphi_{\max}=\frac{\pi}{2}-\omega t_{\max}$,	
\setlength{\arraycolsep}{0.0em}
\begin{eqnarray}
e_{\max}(\omega)&{=}&\dfrac{\underset{t_{s}\leq t\leq t_{s+2q}}{\max}(r(t)-y(t))}{r_0}\nonumber\\
&{=}&\sin(\omega t_{\max})-\dfrac{1}{r_0}\bar{C}\bar{x}(t_{\max})\nonumber,\\
t_{\max}&{\in}&\{t_{ext} \mid  \dot{e}(t_{ext})=0,\ t_{s}\leq t_{ext} \leq t_{s+2q}\}\nonumber\\
&{\cup}&\{t_{s+i} \mid i\in\mathbb{Z},\ 0\leq i\leq 2q\}.\nonumber
\end{eqnarray}
\setlength{\arraycolsep}{5pt}\end{definition}
Using (\ref{E-303}) and (\ref{E-316}) $t_{ext}$ can be obtained from
\setlength{\arraycolsep}{0.0em}
\begin{eqnarray}\label{E-340}
\dot{e}(t_{ext})&{=}&0\Rightarrow \omega\cos(\omega t_{ext})-\bar{C}\bar{B}\begin{bmatrix}1\\0\end{bmatrix}\sin(\omega t_{ext})\nonumber\\
&{=}&\bar{C}\bar{A}\Big(e^{\bar{A}(t_{ext}-t_{s+i})}\big(\xi_{s+{i}}+\psi(t_{s+{i}})\big)-\psi(t_{ext})\Big),\nonumber\\
t_{ext}&{\in}&(t_{s+{i}},t_{s+i+1}],\ i=\{i\in\mathbb{Z}^+\mid i<2q\}.
\end{eqnarray}
\setlength{\arraycolsep}{5pt}Similarly, the pseudo-process sensitivity is defined as the ratio of the maximum error signal of the reset control system (\ref{E-303}) for $d(t)=d_0\sin(\omega t)$ and $r(t)=0$, for all $t\geq0$, to the amplitude of the disturbance at each frequency.
\begin{definition}\label{d2}
The Pseudo-process sensitivity $PS_\infty$ is, for all $\omega\in\mathbb{R}^{+}$,
	$$PS_\infty(j\omega)=e_{\max_d}(\omega)e^{j\varphi_{\max_d}(\omega)},$$
	where $\varphi_{\max_d}=\frac{\pi}{2}-\omega t_{\max_d},$
	\setlength{\arraycolsep}{0.0em}
\begin{eqnarray}
e_{\max_d}(\omega)&{=}&\dfrac{\underset{t^\prime_{s}\leq t\leq t^\prime_{s+2q}}{\max}-y(t)}{d_0}=-\dfrac{1}{d_0}\bar{C}\bar{x}(t_{\max_d}),\nonumber\\
t_{\max_d}&{\in}&\{t_{ext_d}\}\cup\{t^\prime_{s+i}, i\in\mathbb{Z}, 0\leq i\leq2q^\prime\}.\nonumber
\end{eqnarray}
\setlength{\arraycolsep}{5pt}\end{definition}
In a similar way $t_{ext_d}$ is obtained from 
\setlength{\arraycolsep}{0.0em}
\begin{eqnarray}\label{E-341}
\dot{e}(t_{ext_d})&{=}&\dfrac{1}{d_0}\bar{C}\dot{\bar{x}}(t_{ext_d})=0\Rightarrow \bar{C}\bar{B}\begin{bmatrix}0\\1\end{bmatrix}\sin(\omega t_{ext_d})\nonumber\\
&{=}&\resizebox{0.845\hsize}{!}{$\bar{C}\bar{A}\Big(\psi_\mathcal{D}(t_{ext_d})-e^{\bar{A}(t_{ext_d}-t^\prime_{s+i})}\big(\xi^\prime_{s+{i}}+\psi_\mathcal{D}(t^\prime_{s+{i}})\big)\Big),$}\nonumber\\
t_{ext_d}&{\in}&(t^\prime_{s+{i}},t^\prime_{s+i+1}],\ i=\{i\in\mathbb{Z}^+\mid i<2q^\prime\}.
\end{eqnarray}
\setlength{\arraycolsep}{5pt}In order to analyze the noise rejection capability of the system the pseudo-complementary sensitivity is defined as the ratio of the maximum output of the reset control system (\ref{E-303}) for $r(t)=r_0\sin(\omega t)$ and $d=0$, for all $t\geq0$, to the amplitude of the reference at each frequency.
\begin{definition}\label{d2-3-1}
The Pseudo-complementary sensitivity $T_\infty$ is, for all $\omega\in\mathbb{R}^{+}$, 
	$$T_\infty(j\omega)=y_{\max}(\omega)e^{j\varphi_{\max_y}(\omega)},$$
	where $\varphi_{\max_y}=\frac{\pi}{2}-\omega t_{\max_y},$
	\setlength{\arraycolsep}{0.0em}
        \begin{eqnarray}	
	y_{\max}(\omega)&{=}&\dfrac{\underset{t_{s}\leq t\leq t_{s+2q}}{\max}y(t)}{r_0}=\dfrac{1}{r_0}\bar{C}\bar{x}(t_{\max_y}),\nonumber\\
	t_{\max_y}&{\in}&\{t_{ext_y}\}\cup\{t_{s+i}, i\in\mathbb{Z}, 0\leq i\leq2q\}.\nonumber
     \end{eqnarray}
    \setlength{\arraycolsep}{5pt}\end{definition}
Similarly, 
\setlength{\arraycolsep}{0.0em}
        \begin{eqnarray}\label{E-342}
\dot{y}(t_{ext_y})&{=}&\dfrac{1}{r_0}\bar{C}\dot{\bar{x}}(t_{ext_y})=0\Rightarrow \bar{C}\bar{B}\begin{bmatrix}1\\0\end{bmatrix}\sin(\omega t_{ext_y})\nonumber\\
&{=}&\bar{C}\bar{A}\Big(\psi(t_{ext_y})-e^{\bar{A}(t_{ext_y}-t_{s+i})}\big(\xi_{s+{i}}+\psi(t_{s+{i}})\big)\Big),\nonumber\\
t_{ext_y}&{\in}&(t_{s+{i}},t_{s+i+1}],\ i=\{i\in\mathbb{Z}^+\mid i<2q\}.
\end{eqnarray}
    \setlength{\arraycolsep}{5pt}The pseudo-control sensitivity is defined as the ratio of the maximum control input signal of the reset control system (\ref{E-303}) for $r(t)=r_0\sin(\omega t)$ and $d=0$, for all $t\geq0$, to the amplitude of the reference at each frequency.
\begin{definition}\label{d3}
The Pseudo-control sensitivity $CS_\infty$ is, for all $\omega\in\mathbb{R}^{+}$, 	
	$$CS_\infty(j\omega)=u_{\max}(\omega)e^{j\varphi_{\max_u}(\omega)},$$
where $\varphi_{\max_u}=\frac{\pi}{2}-\omega t_{\max_u},$
\setlength{\arraycolsep}{0.0em}
        \begin{eqnarray}
u_{\max}(\omega)&{=}&\dfrac{\underset{t_{s}\leq t\leq t_{s+2q}}{\max}u(t)}{r_0}\nonumber\\
&{=}&\dfrac{1}{r_0}\bar{C}_u\bar{x}(t_{\max_u})+\bar{D}_u\sin(\omega t_{\max_u}),\nonumber\\
t_{\max_u}&{\in}&\{t_{ext_u}\}\cup\{t^\prime_{s+i}, i\in\mathbb{Z}, 0\leq i\leq2q^\prime\}.\nonumber		
\end{eqnarray}
\setlength{\arraycolsep}{5pt}\end{definition}
In addition, $t_{ext_u}$ can be found utilizing the relation
\setlength{\arraycolsep}{0.0em}
        \begin{eqnarray}\label{E-343}
\dot{u}(t_{ext_u})&{=}&0\Rightarrow \bar{D}_u\omega\cos(\omega t_{ext_u})+\bar{C}_u\bar{B}\begin{bmatrix}1\\0\end{bmatrix}\sin(\omega t_{ext_u})\nonumber\\
&{=}&\bar{C}_u\bar{A}\Big(\psi(t_{ext_u})-e^{\bar{A}(t_{ext_u}-t_{s+i})}\big(\xi_{s+{i}}+\psi(t_{s+{i}})\big)\Big),\nonumber\\
t_{ext_u}&{\in}&(t_{s+{i}},t_{s+i+1}],\ i=\{i\in\mathbb{Z}^+\mid i<2q\}.
\end{eqnarray}
\setlength{\arraycolsep}{5pt}In linear control theory, the transfer function of the closed-loop system from the disturbance input $d$ to the control input $u$ is equal to minus the transfer function from the reference signal $r$ to the output signal $y$. However, this relation does not hold for the pseudo-sensitivities due to the non-linear nature of the controller. Hence, the pseudo-control sensitivity of the disturbance is defined as the ratio of the maximum amplitude of the control input, for $d(t)=d_0\sin(\omega t)$ and $r=0$ for all $t\geq0$, to the amplitude of the disturbance, at each frequency.  
\begin{definition}\label{d4}
The Pseudo-control sensitivity of the disturbance $CS_{d_\infty}$ is, for all $\omega\in\mathbb{R}^{+}$,	
	$$CS_{d_\infty}(j\omega)=u_{\max_d}(\omega)e^{j\varphi_{\max_{u_d}}(\omega)},$$
where $\varphi_{\max_{u_d}}=\frac{\pi}{2}-\omega t_{\max_{u_d}},$
\setlength{\arraycolsep}{0.0em}
        \begin{eqnarray}		
	u_{\max_d}(\omega)&{=}&\dfrac{\underset{t^\prime_{s}\leq t\leq t^\prime_{s+2q}}{\max}u(t)}{d_0}=\dfrac{1}{d_0}\bar{C}_u\bar{x}(t_{\max_{u_d}}),\nonumber\\
	t_{\max_{u_d}}&{\in}&\{t_{ext_{u_d}}\}\cup\{t^\prime_{s+i}, i\in\mathbb{Z}, 0\leq i\leq2q^\prime\}.
\end{eqnarray}
\setlength{\arraycolsep}{5pt}\end{definition}
Finally, $t_{ext_{u_d}}$ is calculated through the relation 
\setlength{\arraycolsep}{0.0em}
        \begin{eqnarray}\label{E-344}
\dot{u}(t_{ext_{u_d}})&{=}&0\Rightarrow\bar{C}_u\bar{B}\begin{bmatrix}0\\1\end{bmatrix}\sin(\omega t_{ext_{u_d}})=\bar{C}_u\bar{A}\Big(\psi_\mathcal{D}(t_{ext_{u_d}})\nonumber\\
&&{-}\:e^{\bar{A}(t_{ext_u}-t^\prime_{s+i})}\big(\xi^\prime_{s+{i}}+\psi_\mathcal{D}(t^\prime_{s+{i}})\big)\Big),\nonumber\\\
t_{ext_{u_d}}&{\in}&(t^\prime_{s+{i}},t^\prime_{s+i+1}],\ i=\{i\in\mathbb{Z}^+\mid i<2q^\prime\}.
\end{eqnarray}
\setlength{\arraycolsep}{5pt}We conclude this series of definitions with the following result.
\begin{corollary}\label{coro5}
Consider the reset control system (\ref{E-303}). The pseudo-sensitivities and the closed-loop HOSIDFs are independent of the amplitude of the harmonic excitation input.
\end{corollary}
%%%%%%%%%%%%%%%%%%%%%%%%%%%%%%%%%%%%%%High frequency analysis
\subsection{High Frequency Analysis}\label{sec:3.5}
The evaluation of the sensitivities and the pseudo-sensitivities may be computationally expensive, particularly at high frequencies. In order to simplify these relations the reset instants at high frequencies can be approximated. For sufficiently large frequencies, since the open-loop transfer function is strictly proper and (\ref{E-327}) and~(\ref{E-314}) hold, one has   
\begin{equation}\label{E-3455}
\lim_{\omega \to \infty}\dfrac{\underset{t_s\leq t\leq t_{s+q}}{\max}|e_R(t)-e_{R_1}(t)|}{\underset{t_s\leq t\leq t_{s+q}}{\max}|e_{R_1}(t)|}=0,
\end{equation} 
where $e_{R_1}(t)=R_1\sin(\omega t+\varphi_{e_{R_1}})$ is the first harmonic of $e_R(t)$ (see Appendix~\ref{APC}). Thus,
\begin{equation}\label{E-345}
\begin{array}{*{35}{c}}
\forall \epsilon\in(0,1)\ \exists\ \omega_h\in\mathbb{R}^{+} \mid \forall\ \omega\geq\omega_h:\\
\dfrac{\underset{t_s\leq t\leq t_{s+q}}{\max}|e_R(t)-e_{R_1}(t)|}{\underset{t_s\leq t\leq t_{s+q}}{\max}|e_{R_1}(t)|}\leq \epsilon.
\end{array}
\end{equation} 
Therefore, if $\epsilon$ is chosen sufficiently small, the steady-state reset instants for $\omega\geq\omega_h$ can be approximated as
\begin{equation}\label{E-346}
t_k\approx\dfrac{k\pi-\varphi_{e_{R_1}}}{\omega},
\end{equation}
in which 
\begin{equation}\label{E-347}
\varphi_{e_{R_1}}\approx\phase{\dfrac{C_{\mathfrak{L}_1}(j\omega)}{1+C_{\mathfrak{L}_1}C_{\mathcal{R}_{DF}}C_{\mathfrak{L}_2}G(j\omega)}},
\end{equation}
\newline where $C_{\mathcal{R}_{DF}}$ is the DF of $C_{\mathcal{R}}$ obtained using (\ref{E-203}). 
\begin{remark}
{\rm The accuracy of the approximation depends on the magnitude of $\epsilon$. The smaller the value of $\epsilon$, the more accurate the approximation is.}
\end{remark}
Let $\omega\geq\omega_h$ and $r(t)=\sin(\omega t-\varphi_{e_{R_1}})$. Then (\ref{E-303}) can be re-written as
\begin{equation}\label{E-348}
\left\{
\begin{aligned}
\dot{\bar{x}}(t)&=\bar{A}\bar{x}(t)+\bar{B}\sin(\omega t-\varphi_{e_{R_1}}), & t\neq\frac{k\pi}{\omega}, \\
\bar{x}(t^+)&=\bar{A}_\rho\bar{x}(t), & t=\frac{k\pi}{\omega},\\
y(t)&=\bar{C}\bar{x}(t).
\end{aligned}
\right.
\end{equation}
Thus, $\psi(t)$ in (\ref{E-314}) is given by
\begin{equation}\label{E-349}
\psi_\varphi(t)=(\omega I\cos(\omega t-\varphi_{e_{R_1}})+\bar{A}\sin(\omega t-\varphi_{e_{R_1}}))\mathcal{F}.
\end{equation}
The steady-state reset instants are $\{\frac{(2k)\pi}{\omega},\frac{(2k+1)\pi}{\omega}\}$, $k\in\mathbb{N}$. Therefore,
\setlength{\arraycolsep}{0.0em}
        \begin{eqnarray}\label{E-350}		
	\xi_{s}&{=}&-\xi_{s+1}=\dfrac{-\bar{A}_\rho(I+e^{\frac{\bar{A}\pi}{\omega}})\psi_\varphi(0)}{I+\bar{A}_\rho e^{\frac{\bar{A}\pi}{\omega}}}\Rightarrow\mathcal{R}(1,n,\omega)\nonumber\\
	&{=}&-(e^{\frac{\bar{A}\pi}{\omega}}+I)(\xi_{s}+\psi_\varphi(0)).
\end{eqnarray}
\setlength{\arraycolsep}{5pt}Hence, for $\omega\geq\omega_h$, $T_n(j\omega)$ for the reset control system (\ref{E-303}) are approximated by  
	\begin{equation}\label{E-351}
	\begin{cases}
	\bar{C}(\bar{A}-j\omega I)^{-1}\theta_{\varphi}(\omega)-\bar{C}(j\omega I+\bar{A})\mathcal{F}, &  n=1,\\
	\bar{C}(\bar{A}-jn\omega I)^{-1}\theta_{\varphi}(\omega), & n>1\text{ odd,}\\
	0, & n\text{ even,}
	\end{cases}
	\end{equation}  
in which
\setlength{\arraycolsep}{0.0em}
        \begin{eqnarray}\label{E-354}		
	\theta_{\varphi}(\omega)&{=}&\dfrac{-2j\omega e^{j\varphi_{e_{R_1}}}}{\pi}(I+e^{\frac{\bar{A}\pi}{\omega}})\Bigg(I\nonumber\\
	&&{-}\:(I+\bar{A}_\rho e^{\frac{\bar{A}\pi}{\omega}})^{-1}(\bar{A}_\rho(I+e^{\frac{\bar{A}\pi}{\omega}}))\Bigg)\psi_\varphi(0).
\end{eqnarray}
\setlength{\arraycolsep}{5pt}A similar analysis holds for the steady-state response of the reset control system (\ref{E-303}) to a disturbance input. Similarly, let $d=\sin(\omega t-\varphi_{e_{d_1}})$. Then
\begin{equation}\label{E-355}
\varphi_{e_{d_1}}\approx\phase{\dfrac{-C_{\mathfrak{L}_1}(j\omega)G(j\omega)}{1+C_{\mathfrak{L}_1}C_{\mathcal{R}_{DF}}C_{\mathfrak{L}_2}G(j\omega)}},
\end{equation}
and $\psi_\mathcal{D}$ in (\ref{E-336}) is given by 
\begin{equation}\label{E-356}
\psi_{\mathcal{D}_\varphi}(t)=(\omega I\cos(\omega t-\varphi_{e_{d_1}})+A\sin(\omega t-\varphi_{e_{d_1}}))\mathcal{F}_\mathcal{D}.
\end{equation} 
Similarly,
\setlength{\arraycolsep}{0.0em}
        \begin{eqnarray}\label{E-357}		
	\xi^\prime_{s}&{=}&-\xi^\prime_{s+1}=\dfrac{-\bar{A}_\rho(I+e^{\frac{\bar{A}\pi}{\omega}})\psi_{\mathcal{D}_\varphi}(0)}{I+\bar{A}_\rho e^{\frac{\bar{A}\pi}{\omega}}}\Rightarrow\mathcal{R}_\mathcal{D}(1,n,\omega)\nonumber\\
	&{=}&-(e^{\frac{\bar{A}\pi}{\omega}}+I)(\xi^\prime_{s}+\psi_{\mathcal{D}_\varphi}(0)).
\end{eqnarray}
\setlength{\arraycolsep}{5pt}Therefore, for $\omega$ sufficiently large $PS_n(j\omega)$ for the reset control system (\ref{E-303}) are approximated as
\begin{equation}\label{E-358}
  \begin{cases}
	\bar{C}(j\omega I-\bar{A})^{-1}\theta_{\mathcal{D}_\varphi}(\omega)+\bar{C}(j\omega I+\bar{A})\mathcal{F}_\mathcal{D}, &  n=1,\\
	\bar{C}(jn\omega I-\bar{A})^{-1}\theta_{\mathcal{D}_\varphi}(\omega), & n>1\text{ odd,}\\
	0, & n \text{ even,}
	\end{cases}
\end{equation}
in which
\setlength{\arraycolsep}{0.0em}
        \begin{eqnarray}\label{E-360}		
	\theta_{\mathcal{D}_\varphi}(\omega)&{=}&\dfrac{-2j\omega e^{j\varphi_{e_{d_1}}}}{\pi}(I+e^{\frac{\bar{A}\pi}{\omega}})\Bigg(I\nonumber\\
	&&{-}\:(I+\bar{A}_\rho e^{\frac{\bar{A}\pi}{\omega}})^{-1}(\bar{A}_\rho(I+e^{\frac{\bar{A}\pi}{\omega}}))\Bigg)\psi_{\mathcal{D}_\varphi}(0).
\end{eqnarray}
\setlength{\arraycolsep}{5pt}
\begin{remark}\label{tool}
{\rm The presented results have been integrated into an open source toolbox, which has been developed using Matlab, see \cite{Toolbox}. This toolbox facilitates the analysis and design for reset control systems.} 
\end{remark}
%%%%%%%%%%%%%%%%%%%%%%%%%%%%%%%%%%%%%%%%%%%%%Periodic inputs
\section{Periodic inputs}\label{sec:4}
In Section~\ref{sec:3} a notion of frequency response and pseudo-sensitivities for reset control systems have been defined. These serve as graphical tools for performance analysis of reset controllers. The pseudo-sensitivities determine how a system amplifies harmonic inputs at various frequencies, information which is essential for control designers. However, this information is obtained for a single harmonic excitation and since the superposition principle does not hold, it provides only an approximation in the case of multi-harmonics excitation. In this section the steady-state performance in the presence of multi-harmonics excitation and periodic inputs is investigated. This is reasonable since most references and disturbances are periodic \cite{pavlov2013steady,huang2002application}. 
\newline For ease of notation let $\text{lcm}\left(\dfrac{a_1}{b_1},\dfrac{a_2}{b_2},...,\dfrac{a_i}{b_i}\right)$ denote the least common multiple of $\dfrac{a_1}{b_1},\dfrac{a_2}{b_2},...,\text{ and }\dfrac{a_i}{b_i}$ and $\text{gcd}\left(\dfrac{a_1}{b_1},\dfrac{a_2}{b_2},...,\dfrac{a_i}{b_i}\right)$ denote the greatest common divisor of $\dfrac{a_1}{b_1},\dfrac{a_2}{b_2},...,\text{ and }\dfrac{a_i}{b_i}$ in which $a_i\in\mathbb{N}$ and $b_i\in\mathbb{N}$.  
\begin{theorem}\label{thm3}
Consider the reset control system (\ref{E-303}). Suppose the $H_\beta$ condition and Assumption \ref{AS1} hold. Then for any periodic excitation of the form
\begin{equation}\label{E-401}
w(t)=w_0\sin(\dfrac{2\pi}{T_0} t)+w_1\sin(\dfrac{2\pi}{T_1} t)+...+w_N\sin(\dfrac{2\pi}{T_N} t),
\end{equation}
with $w_i=[r_i,\ d_i]^T$, the reset control system (\ref{E-303}) has a periodic steady-state solution of the form
\setlength{\arraycolsep}{0.0em}
        \begin{eqnarray}		
	\bar{x}(t)&{=}&a_0+\mathlarger{\sum}\limits ^{\infty}_{n=1}a_n\cos(n\omega_Mt)+b_n\sin(n\omega_Mt),\nonumber\\
	\omega_M&{=}&2\pi\times \text{gcd}(\dfrac{1}{T_0},\dfrac{1}{T_1},...,\dfrac{1}{T_N}).\nonumber
\end{eqnarray}
\setlength{\arraycolsep}{5pt}\end{theorem}
\begin{proof}
Let $t_{s_{M+i}}$ be the steady-state reset instants of the reset control system (\ref{E-303}) for $w$ is given in (\ref{E-401}). By (\ref{E-303}) the steady-state solution for $w$ as in (\ref{E-401}) is given by
\begin{equation}\label{E-402}
\resizebox{\hsize}{!}{$
\bar{x}(t)=e^{\bar{A}(t-t_{s_M})}\Big(\xi_{s_M}+\psi_M(t_{s_M})\Big)-\psi_M(t),\ t\in(t_{s_M},t_{s_M+1}],
$}
\end{equation} 	
where
\begin{align*}
\psi_M(t) &=\psi_0(t)+\psi_1(t)+...+\psi_N(t),\\
\psi_i(t) &=(\omega_i I\cos(\omega_i t)+\bar{A}\sin(\omega_i t))\mathcal{F}_i,\\ 
\mathcal{F}_i &=(\omega_i^2I+\bar{A}^2)^{-1}\bar{B}w_i.
\end{align*}
By Lemma \ref{L2} the reset control system (\ref{E-303}) forgets the initial condition; thus, using as similar procedure as the one in Section \ref{sec:31}, yields
\setlength{\arraycolsep}{0.0em}
\begin{eqnarray}\label{E-403}
\bar{x}(t)&{=}&e^{\bar{A}(t-t_{s_M})}\Bigg((I-\bar{A}_\rho)\psi_M(t_{s_M})+e^{\bar{A}(t_{s_M}-t_{s_M-1})}\bigg(\nonumber\\
&&{}\,(I-\bar{A}_\rho)\psi_M(t_{s_M-1})+\dots+\bar{A}_\rho e^{\bar{A}(t_{s_M-1}-t_{s_M-2})}\dots\nonumber\\
&&{}\,\bar{A}_\rho e^{\bar{A}(t_{s_M-m+1}-t_{s_M-m})}(I-\bar{A}_\rho)\psi_M(t_{s_M-m})\Bigg)\Bigg)\nonumber\\
&&{-}\:\psi_M(t),\ t\in(t_{s_M},t_{s_M+1}].\nonumber
\end{eqnarray}
\setlength{\arraycolsep}{5pt}Since the reseting condition is 
\begin{equation}\label{E-404} 
\bar{C}_{e_R}\bar{x}(t)+D_{e_R}[1\ 0]w_M(t)=0,
\end{equation} 
if $\{t_{s_M},t_{s_M-1},...,t_{s_M-m}\}$ are reset instants and satisfy (\ref{E-404}), then $t\in\{t_{s_M},t_{s_M-1},...,t_{s_M-m}\}+\dfrac{2\pi}{\omega_M}$ are such that (\ref{E-404}) holds, which implies that the sequence of reset instants is periodic with period $\dfrac{2\pi}{\omega_M}$; hence, $\bar{x}(t)=\bar{x}(t+\dfrac{2\pi}{\omega_M})$, and using the Fourier series representation yields
\begin{equation}\label{E-405} 
\bar{x}(t)=a_0+\mathlarger{\sum}\limits ^{\infty}_{n=1}a_n\cos(n\omega_Mt)+b_n\sin(n\omega_Mt).
\end{equation}
\end{proof}
\begin{corollary}\label{coro6}
Consider the reset control system~(\ref{E-303}). Suppose the $H_\beta$ condition and Assumption~\ref{AS1} hold. Then for any periodic input $w_P(t)=w_P(t+T_P)$ the reset control system~(\ref{E-303}) has a steady-state periodic solution with the same period time $T_P$. 
\end{corollary}
\begin{proof}
A periodic function can be written as
\begin{equation}\label{E-406} 
w_P(t)=a^\prime_0+\mathlarger{\sum}\limits ^{\infty}_{n=1}a^\prime_n\cos(n\dfrac{2\pi}{T_P}t)+b^\prime_n\sin(n\dfrac{2\pi}{T_P}t).
\end{equation}
Using Theorem \ref{thm3} the steady-state solution of~(\ref{E-303}) for the input~(\ref{E-406}) is 
\setlength{\arraycolsep}{0.0em}
\begin{eqnarray}
\bar{x}(t)&{=}&a_0+\mathlarger{\sum}\limits ^{\infty}_{n=1}a_n\cos(n\omega_Mt)+b_n\sin(n\omega_Mt),\nonumber\\
\omega_M&{=}&2\pi\times gcd(\dfrac{1}{T_P},\dfrac{2}{T_P},...,\dfrac{n}{T_P})=\dfrac{2\pi}{T_P},\text{ hence the claim.}\nonumber
\end{eqnarray}
\setlength{\arraycolsep}{5pt}\end{proof}
%It is possible to find the steady-state solution of the reset control system~(\ref{E-303}) for a periodic input using the similar procedure in Section \ref{sec:31}. However, it will be computationally expensive and complicated. In the literature there are several numerical methods for finding the periodic steady-state response of non-linear systems \cite{pavlov2013steady,ascher1994numerical,parker2012practical}. 
%%%%%%%%%%%%%%%%%%%%%%%%%%%%%%%%%%%%%%%%%%%%%%%%%%%%%%%%%%%%%%%%%%%%Illustrative example
\section{An Illustrative example}\label{sec:5}
In this section an illustrative example showing the effectiveness of the developed results is presented. A 3DOF precision positioning system, see Fig.~\ref{F-04-05}, \cite{dastjerdi2018tuning} and \cite{saikumar2019complex}, is selected for this purpose. In this system we only consider mass 3 and actuator 1A (see Fig.~\ref{F-04-05}) which can be modelled via the transfer function, see~\cite{Houu},  
\begin{equation}\label{E-501}
G(s)=\dfrac{8695}{s^2+4.36s+7627}.
\end{equation}
Note that the $H_\beta$ condition holds for all of the controllers designed in the following subsections.
\begin{figure}
	\begin{center}
		\includegraphics[width=0.6\hsize]{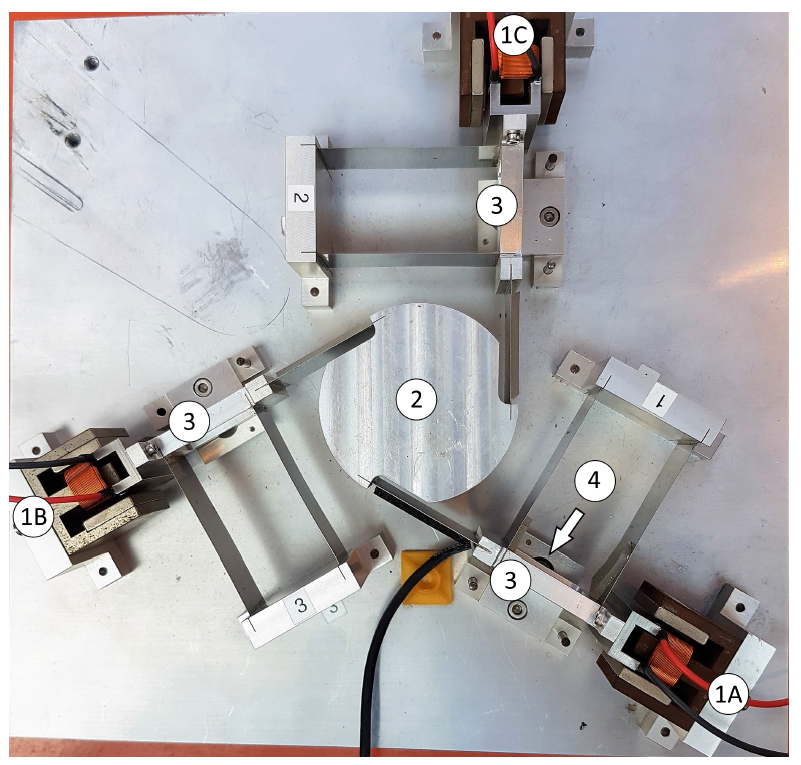}     
		\caption{A 3 DOF planar precision positioning Spyder stage. The voice coil actuators 1A, 1B and 1C control three masses (labelled as 3) which are constrained by leaf flexures. The three masses are connected to a central mass (labelled as 2) through leaf flexures. A Linear encoder (labelled as 4) is placed under mass 3 to provide the position feedback} 
		\label{F-04-05}                                 
	\end{center}                              
\end{figure}
%%%%%%%%%%%%%%%%%Digging more into Clegg integrator
\subsection{The Optimal Structure for CI}\label{sec:5.1}
The closed-loop frequency responses of the system with two reset controllers are compared against the closed-loop frequency responses achievable with a ``tamed" PID controller~\cite{dastjerdi2018tuning} with base linear transfer function 
\begin{equation}\label{E-502}
C_{PID}(s)=k_p\left(1+\dfrac{\omega_i}{s}\right) \left(\dfrac{\dfrac{s}{\omega_d}+1}{\dfrac{s}{\omega_t}+1}\right).
\end{equation}
The first reset controller is obtained by replacing the integrator in (\ref{E-502}) with a CI yielding
\begin{equation}\label{E-503}
C_{SP(CI)D}(s)=k_p\left(1+\cancelto{}{\frac{\omega_i}{s}}\right)\left(\dfrac{\dfrac{s}{\omega_d}+1}{\dfrac{s}{\omega_t}+1}\right),
\end{equation} 
and the second reset controller is the parallel form of (\ref{E-503}), that is
\begin{equation}\label{E-5033}
C_{PP(CI)D}(s)=k_p\left(1+\cancelto{}{\frac{\omega_i}{s}}+\dfrac{\dfrac{s}{\omega_d}}{\dfrac{s}{\omega_t}+1}\right).
\end{equation} 
Note that, unlike the case of linear controllers, the parallel and series configuration of reset controllers can result in totally different responses. In this example we show that in contrast with the DF method, our method is capable of exposing this difference. Setting $100$ Hz as the crossover frequency $\omega_c$, the control parameters have been tuned based on the method proposed in~\cite{schmidt2014design,krijnen2017application,dastjerdi2018tuning} as $K_p=\dfrac{1}{3|G(j\omega_c)|}=14.35$, $\omega_t=3\omega_c=600\pi$, $\omega_i=\dfrac{\omega_c}{10}=20\pi$, and $\omega_d=\dfrac{\omega_c}{3}=66.6\pi$. All frequency responses are obtained utilizing the toolbox in~\cite{Toolbox}. The open-loop frequency response of the system with the controller $C_{PID}$, and the DFs and the amplitudes of third harmonics of the system with the controllers $C_{SP(CI)D}$, $C_{PP(CI)D}$ are shown in Fig.~\ref{F-05}. Based on the DF analysis it is expected that the tracking performances and the disturbance rejection capabilities of the system with controllers $C_{PP(CI)D}$ and $C_{SP(CI)D}$ are the same, and these performance capabilities are superior to those of the system with the controller $C_{PID}$. In addition, the control inputs and the noise attenuation capabilities of the system with these controllers are expected to be almost the same. However, the magnitude of high order harmonics of the reset controllers is different. The time-domain results (Fig.~\ref{F-08}) disprove the predictions which rely on the DF method. In this figure the tracking errors and the amplitude of the control inputs of the system with these controllers are displayed for $r(t)=100\sin(2\pi t)$. It is seen that the control input of the system with the controller $C_{SP(CI)D}$ is much larger than the amplitude of the control inputs of the system with the controllers $C_{PID}$ and $C_{PP(CI)D}$, whereas the tracking performance of the system with the controller $C_{SP(CI)D}$ is worse than the tracking performances of the system with the controllers $C_{PID}$ and $C_{PP(CI)D}$. Note that similar to results presented in~\cite{saikumar2019constant,palanikumar2018no,mecha,Nima}, the amplitudes of even harmonics of the response are zero. 
\begin{figure}
	\begin{center}
		\includegraphics[width=0.5\hsize]{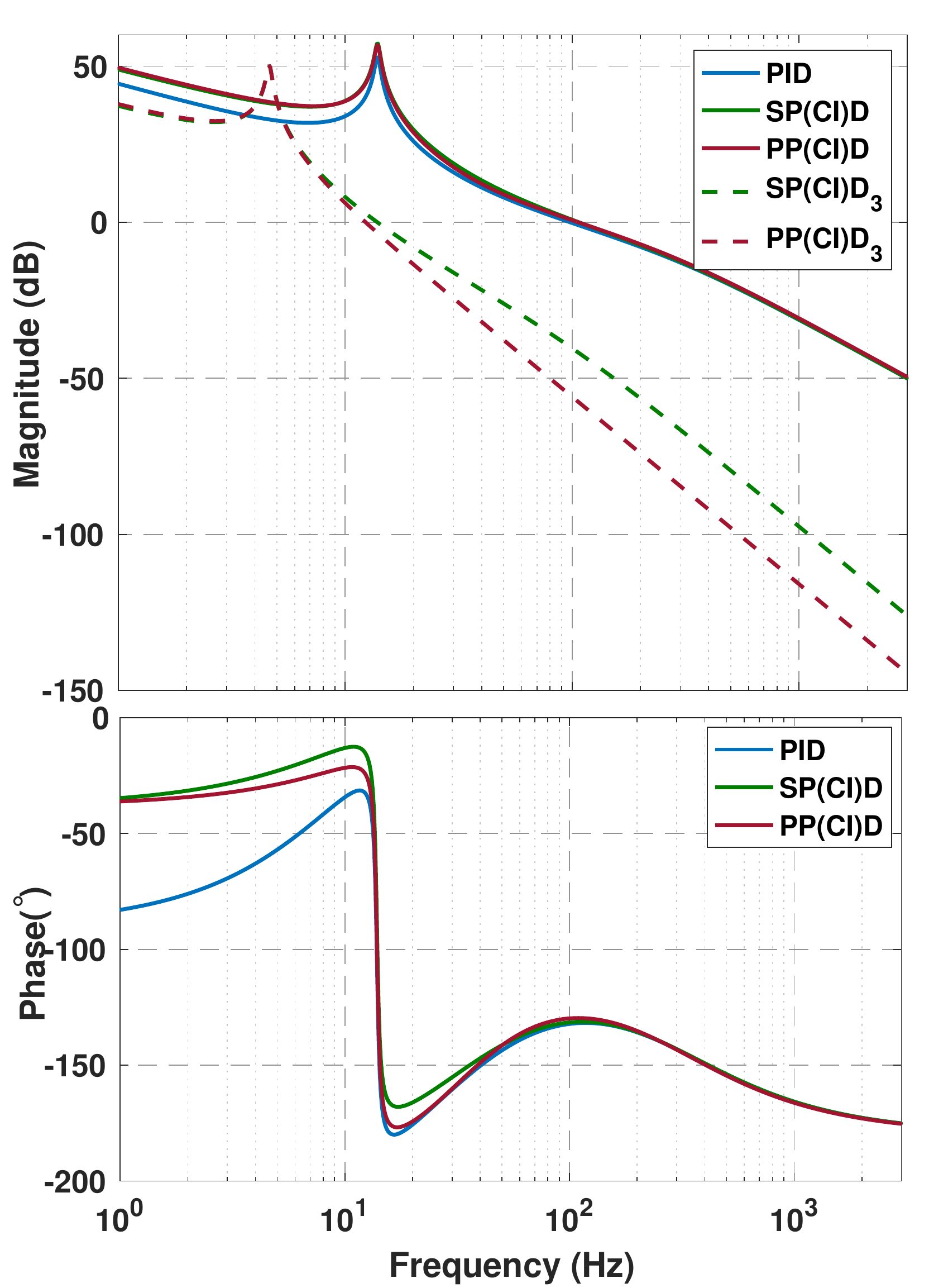}    % The printed column  
		\caption{The DFs and the amplitudes of the third harmonics of the open-loop system with the controllers $C_{SPI(CI)D}$ and $C_{PPI(CI)D}$, and open-loop frequency response of the system with the controller $C_{PID}$} 
		\label{F-05}                                 
	\end{center}                              
\end{figure}
\begin{figure}[!t]
	\centering
	\begin{subfigure}[b]{0.45\columnwidth}
		\includegraphics[width=\hsize]{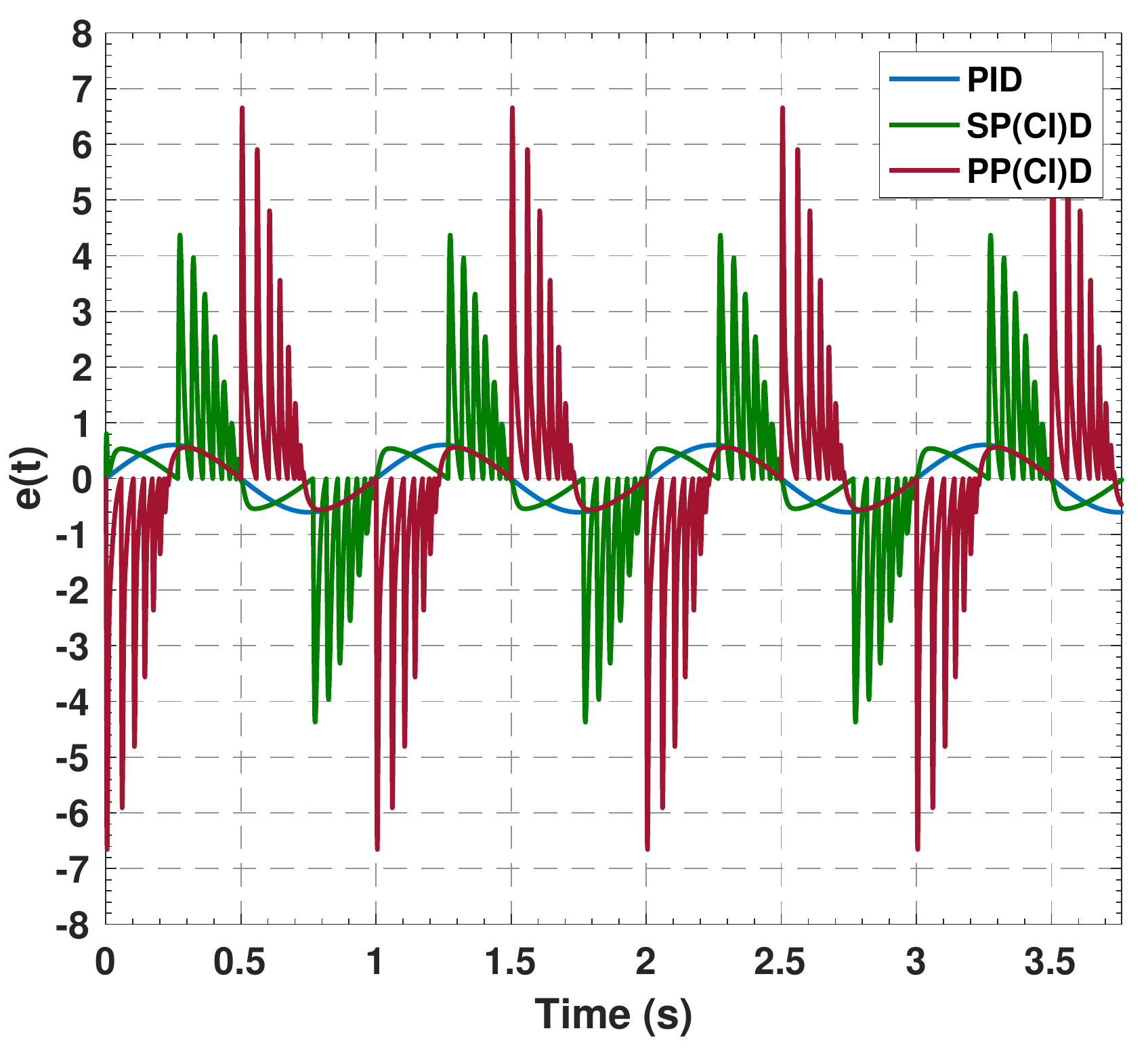}
		\caption{Tracking error}
		\label{F-81}
	\end{subfigure}
	\hfil
	\begin{subfigure}[b]{0.45\columnwidth}
		\includegraphics[width=\hsize]{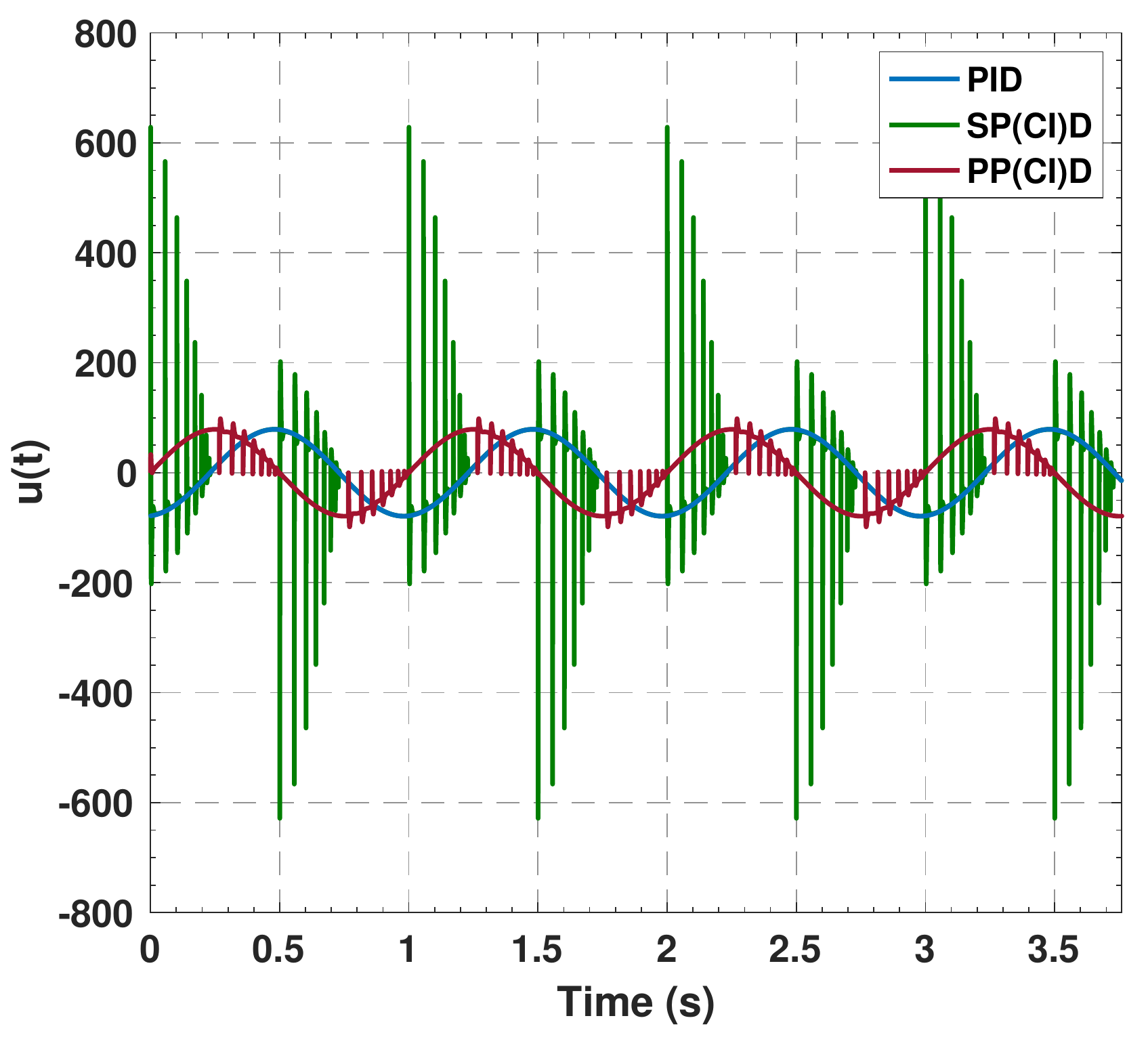}
		\caption{Control input}
		\label{F-82}
	\end{subfigure}
	\caption{Time histories of the tracking errors and of the control inputs of the system with the controllers $C_{PP(CI)D}$, $C_{SP(CI)D}$  and $C_{PID}$ for $r(t)=100\sin(2\pi t)$}
	\label{F-08}
\end{figure}
\begin{figure*}
	\centering
	\begin{subfigure}[b]{0.49\columnwidth}
		\centering    
		\includegraphics[width=\hsize]{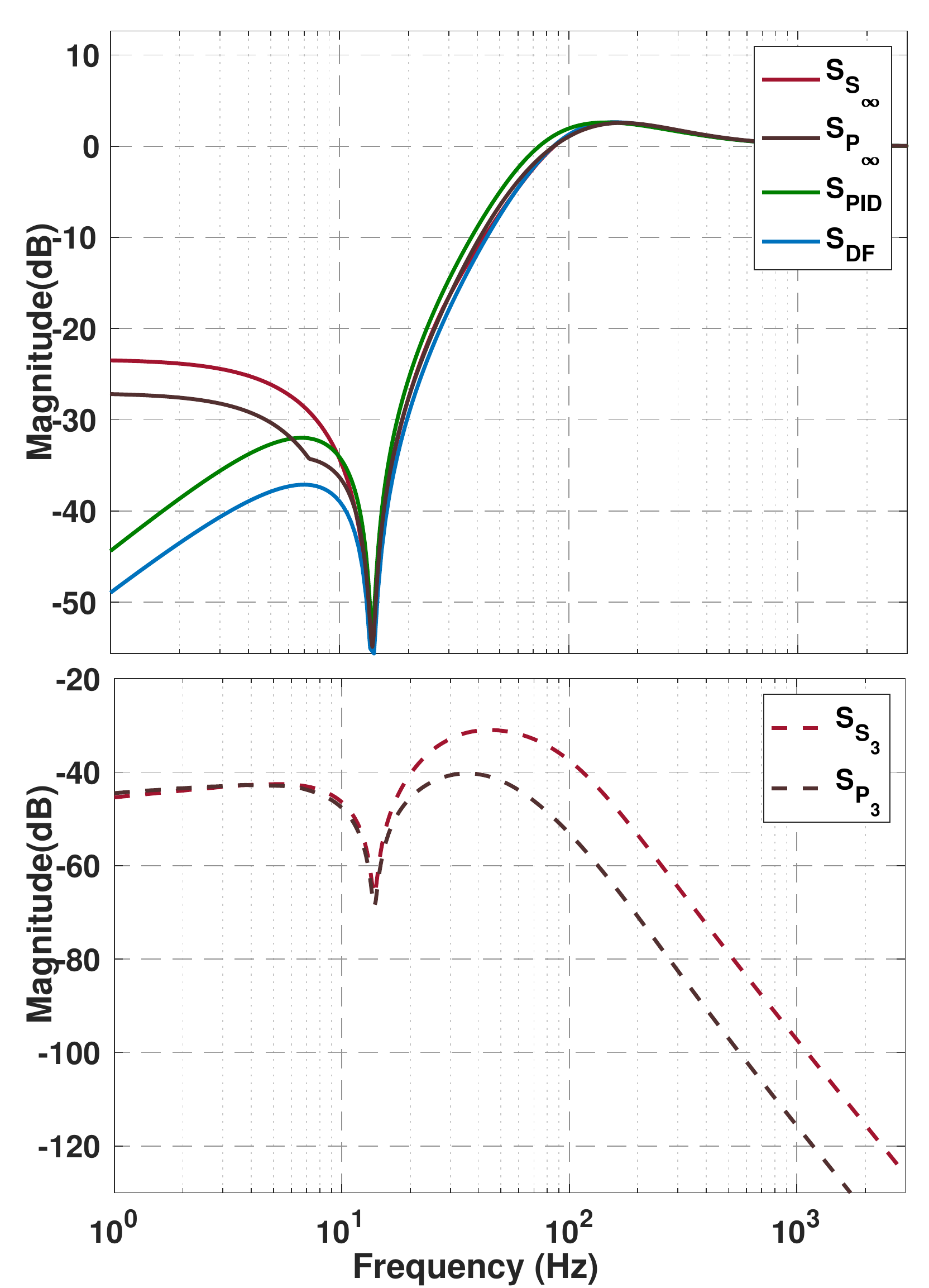}
		\caption{Sensitivity}        
		\label{F-71}
	\end{subfigure}
	\hfil
	\begin{subfigure}[b]{0.49\columnwidth}
		\centering
		\includegraphics[width=\hsize]{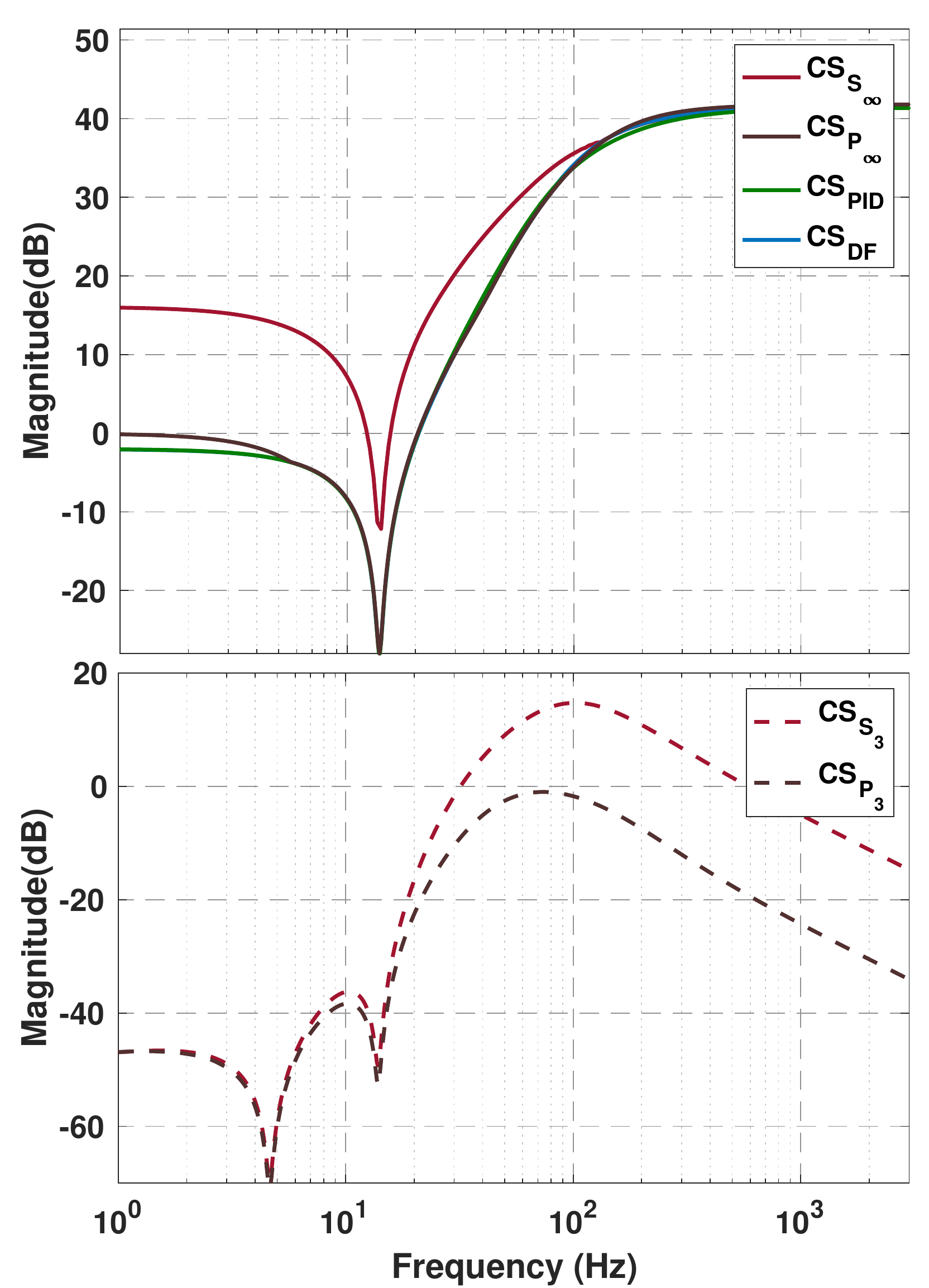}
		\caption{Control sensitivity}
		\label{F-73}
	\end{subfigure}
	\hfil
	\begin{subfigure}[b]{0.49\columnwidth}
		\centering 
		\includegraphics[width=\hsize]{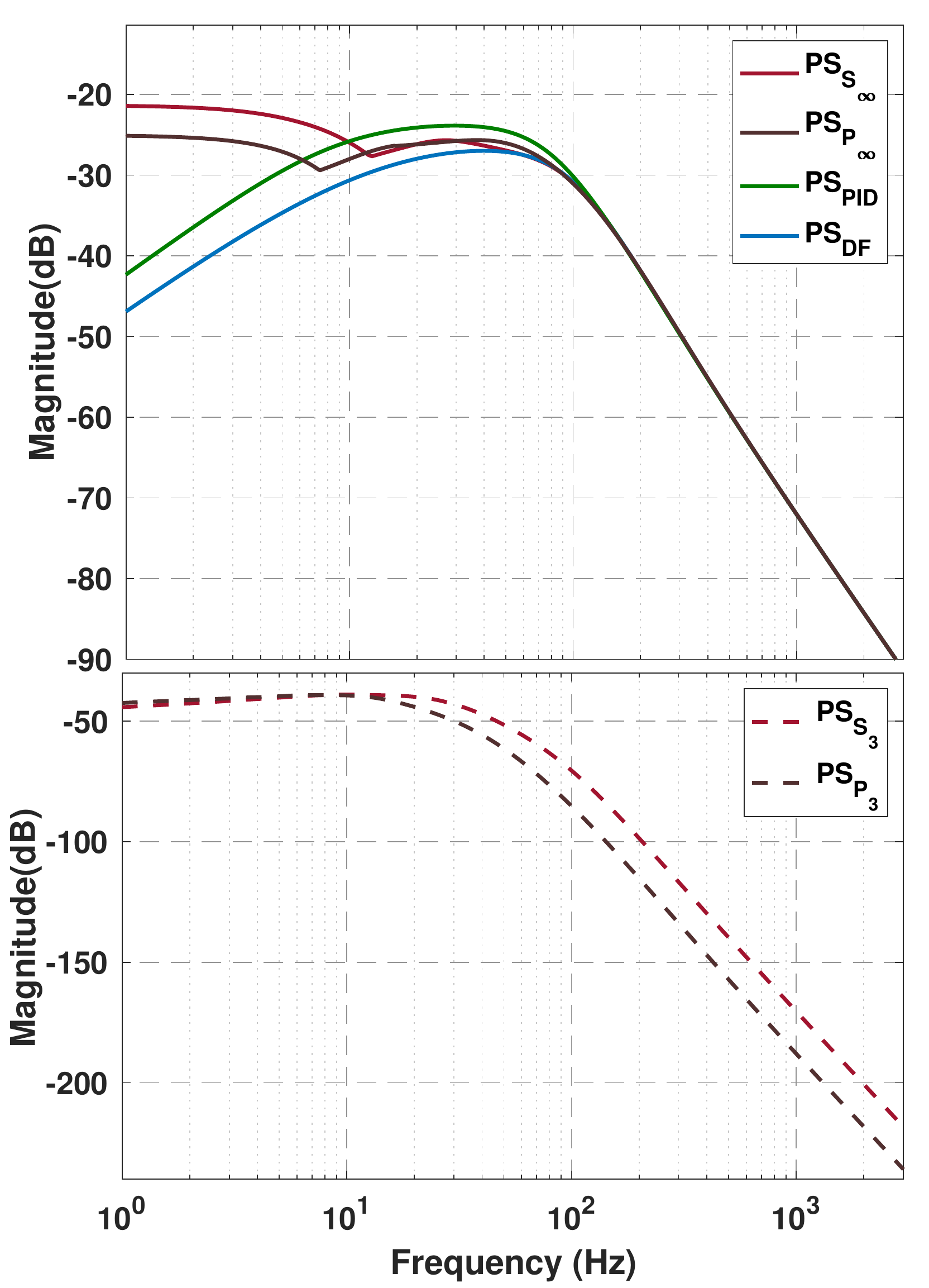}
		\caption{Process sensitivity}
		\label{F-72}
	\end{subfigure}
	\hfil
	\begin{subfigure}[b]{0.49\columnwidth}
		\centering
		\includegraphics[width=\hsize]{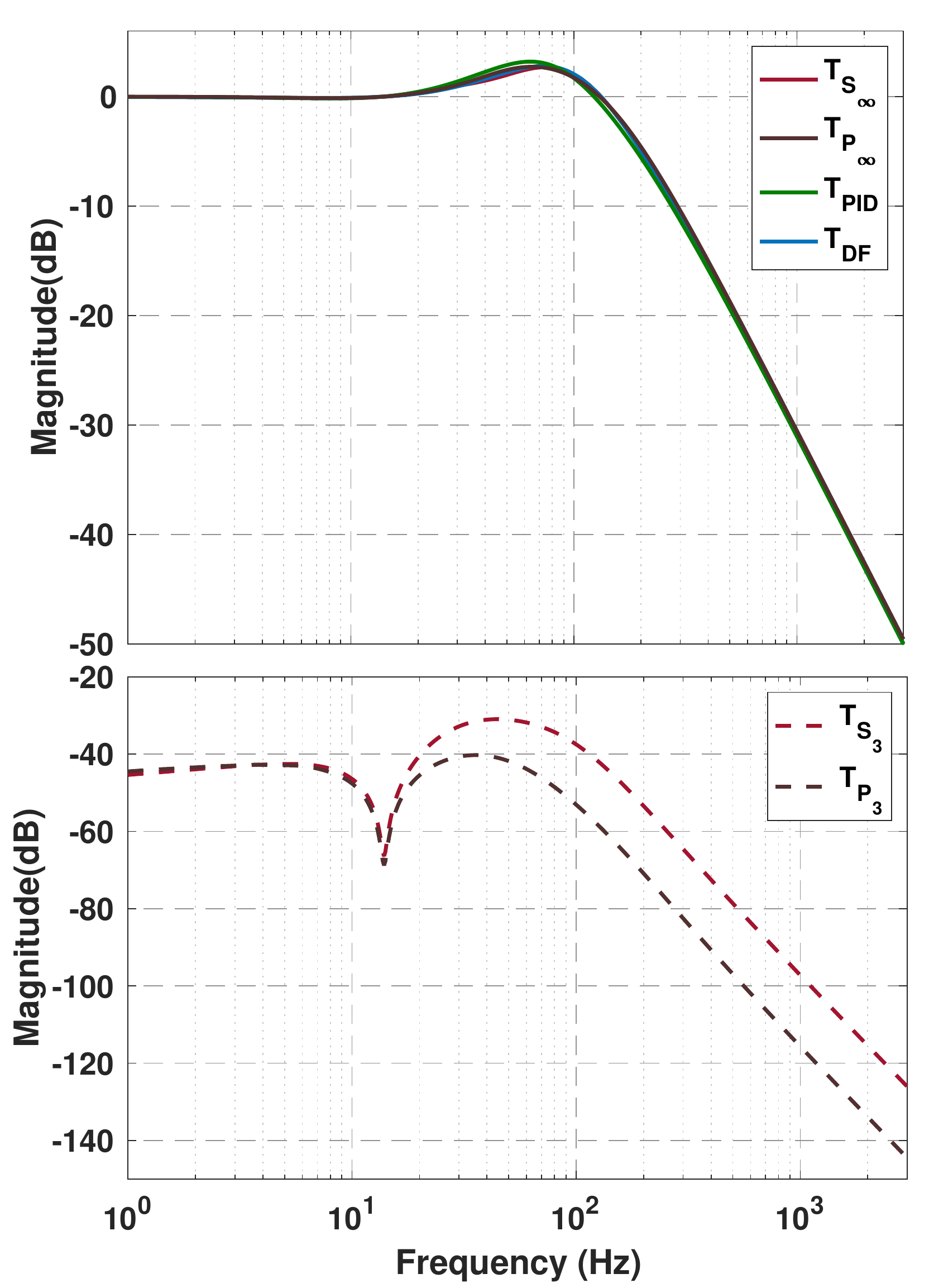}
		\caption{Complementary sensitivity}
		\label{F-74}
	\end{subfigure}
	~ %add desired spacing between images, e. g. ~, \quad, \qquad, \hfill etc. 
	\caption{The DFs $(.\_$ DF), amplitudes of the third harmonics of the sensitivities (.$\_3$), and amplitudes of pseudo-sensitivities (.$\_\ \infty$) of the closed-loop system with the controllers $C_{SP(CI)D}$, $C_{PP(CI)D}$, and closed-loop sensitivities of the system with the controller $C_{PID}$}
	\label{F-07}
\end{figure*}

Unlike the DF method, the pseudo-sensitivities (Fig.~\ref{F-07}) allows justifying why the performance of the system with the controller $C_{PID}$ is superior to the performances of the system with the controllers $C_{PP(CI)D}$ and $C_{SP(CI)D}$ in terms of precision and control effort. As illustrated in Fig.~\ref{F-71}, at low frequency the tracking performance of the system with the controller $C_{PID}$ is better than that of the system with the controllers $C_{PP(CI)D}$ and $C_{SP(CI)D}$. Moreover, the tracking performance of the system with the controller $C_{PP(CI)D}$ is better than the tracking performance of the system with the controller $C_{SP(CI)D}$ at frequencies around the cross-over frequency. As it can be seen in Fig.~\ref{F-73}, the amplitude of the function $CS_\infty$ of the system with the controller $C_{SP(CI)D}$ is much higher than that of the system with the controller $C_{PP(CI)D}$ and of the control sensitivity of the system with the controller $C_{PID}$. Thus, to avoid saturation problems designers should use the function $CS_\infty$ instead of using the result obtained from the DF method when reset controllers are used.  

In addition, as shown in Fig.~\ref{F-72}, the low frequency disturbance rejection capability of the system with the controller $C_{PID}$ is better than that of the system with the controllers $C_{PP(CI)D}$ and $C_{SP(CI)D}$. Furthermore, as illustrated in Fig.~\ref{F-74}, the noise attenuating capabilities of the system with these three controllers are the same. The differences between the performances of the system with the controllers $C_{PP(CI)D}$ and $C_{SP(CI)D}$ are due to the differences in the amplitude and phase of the high order harmonics produced by these controllers. 

To sum up, although it has been shown that using CIs, instead of linear integrators, improves the transient response of the system, the proposed results show that this deteriorates the tracking performance of the system, and the system needs a ``stronger" actuator. Moreover, the actual implementation of the CI has significant effects on the performance of the system which cannot be exposed by using the results obtained with the DF method. This analysis reveals that the CI should be used in the parallel architecture (\ref{E-5033}), yielding a system with better precision and lower control input once compared with the system with the CI in the series architecture (\ref{E-503}). 
%In the series structure (\ref{E-503}) the CI produces discontinuity in the signal and it passes through a lead filter which amplifies this type of non-linearity. As a result, the control input of the system is increased while its precision is worsened. 
%%%%%%%%%%%%%%%%%%%%%%%%%%%%%%%%%%Performance of CgLp 
\subsection{Performance of "Constant in gain Lead in phase (CgLp)" Compensators}\label{sec:5.2}
Reset elements are utilized to introduce new compensators to enhance performance of control systems \cite{hunnekens2014synthesis,van2018hybrid,palanikumar2018no,chen2019development,valerio2019reset,saikumar2019constant}. In this section a new reset compensator called Constant in gain Lead in phase (CgLp) is analyzed. It consists of a reset filter FORE and a Proportional Derivative (PD) filter in series~\cite{saikumar2019constant,palanikumar2018no}. The DF of a CgLp compensator is given in Fig.~\ref{F-10}. Note that the combination of a PD and a FORE produces a compensator with a constant gain, while providing a phase lead.
\begin{figure}[!t]
	\centering
	\includegraphics[width=0.6\hsize]{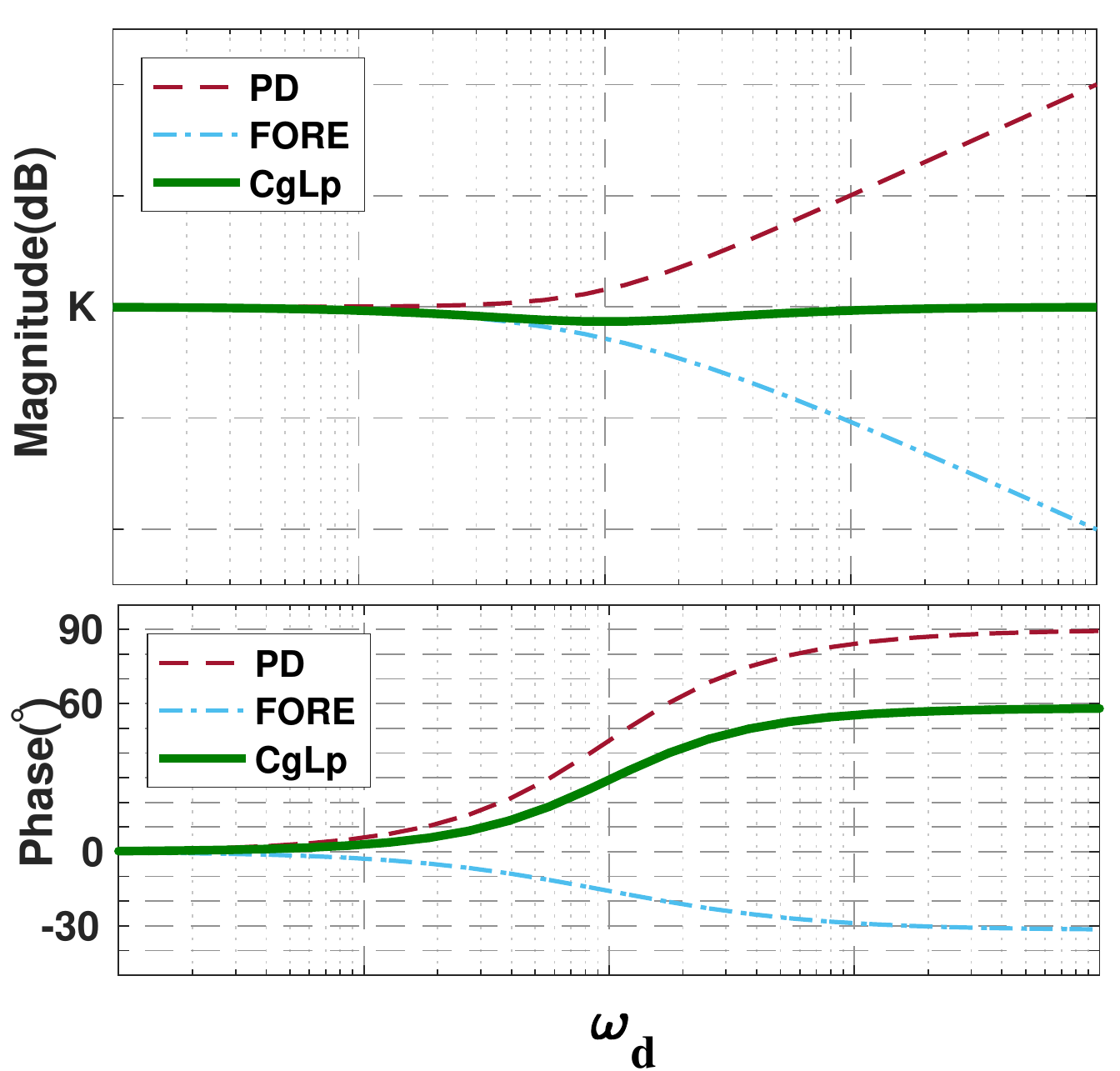}
	\caption{The DF of a CgLp compensator}
	\label{F-10}
\end{figure}
To study the effects of the ``position" of the control elements on the performance of systems with reset controllers, two controllers (see Fig.~\ref{F-N1}) with CgLp compensators are considered. Both controllers are described by
\begin{equation}\label{E-83}
C _ {g_i}(s) = k_p\underbrace{\overbrace{\left(\cancelto{\gamma}{\dfrac{1}{\frac{s}{\omega_r}+1}}\right)}^{\mathrm{FORE}}\overbrace{\left(\dfrac{\frac{s}{\omega_d}+1}{\frac{s}{\omega_t}+1}\right)}^{\mathrm{Lead}}}_{\mathrm{CgLp}}\underbrace{\overbrace{\left(1+\frac{\omega_i}{s}\right)}^{\mathrm{PI}}\overbrace{\left(\dfrac{\frac{s}{\omega_l}+1}{\frac{s}{\omega_f}+1}\right)}^{\mathrm{Lead}}}_{\mathrm{Tamed\ PID}}.
\end{equation}
The parameters of these two controllers are the same and tuned optimally based on the method described in \cite{mecha}, yielding $k_p=25.5,\ \omega_r=111\pi,\ \omega_d=105.2\pi,\ \omega_t=1640\pi,\ \omega_i=20\pi,\ \omega_l=105.2\pi,\ \omega_f=260\pi,$ and $\gamma=0.3$. The only difference between these two controllers is in the ``position" of the filters. For $C_{g_1}$ is FORE-lead-proportional-integrator, while for $C_{g_2}$ one has lead-FORE-proportional-integrator. The DFs and the amplitudes of the third harmonic of the open-loop system with both controllers are given in Fig.~\ref{F-11}. The DFs of the open-loop system with both controllers are the same, but the amplitudes of their third harmonic are different which yields different performances. 
\begin{figure}
  \centering
  \begin{subfigure}{\columnwidth}
    \centering
    \resizebox{\textwidth}{!}{
    \tikzset{every picture/.style={line width=0.75pt}} %set default line width to 0.75pt        
\begin{tikzpicture}[x=0.75pt,y=0.75pt,yscale=-1,xscale=1]
%uncomment if require: \path (0,218); %set diagram left start at 0, and has height of 218
%Shape: Rectangle [id:dp3558696855554735] 
\draw  [line width=1.5]  (451.5,65) -- (552.5,65) -- (552.5,133) -- (451.5,133) -- cycle ;
%Shape: Ellipse [id:dp605355155007869] 
\draw  [line width=1.5]  (42.63,96.45) .. controls (42.63,88.84) and (49.3,82.68) .. (57.53,82.68) .. controls (65.76,82.68) and (72.43,88.84) .. (72.43,96.45) .. controls (72.43,104.06) and (65.76,110.22) .. (57.53,110.22) .. controls (49.3,110.22) and (42.63,104.06) .. (42.63,96.45) -- cycle ;
%Straight Lines [id:da8035539331450979] 
\draw [line width=1.5]    (642,102) -- (642,210) -- (57,210) -- (57.51,114.22) ;
\draw [shift={(57.53,110.22)}, rotate = 450.3] [fill={rgb, 255:red, 0; green, 0; blue, 0 }  ][line width=0.08]  [draw opacity=0] (11.61,-5.58) -- (0,0) -- (11.61,5.58) -- cycle    ;
%Straight Lines [id:da10064598351922416] 
\draw [line width=1.5]    (-0.5,98) -- (38.63,98.41) ;
\draw [shift={(42.63,98.45)}, rotate = 180.6] [fill={rgb, 255:red, 0; green, 0; blue, 0 }  ][line width=0.08]  [draw opacity=0] (11.61,-5.58) -- (0,0) -- (11.61,5.58) -- cycle    ;
%Straight Lines [id:da053830610313541416] 
\draw [line width=1.5]    (629.77,101.37) -- (667.5,101.94) ;
\draw [shift={(671.5,102)}, rotate = 180.87] [fill={rgb, 255:red, 0; green, 0; blue, 0 }  ][line width=0.08]  [draw opacity=0] (11.61,-5.58) -- (0,0) -- (11.61,5.58) -- cycle    ;
%Straight Lines [id:da3228273022349576] 
\draw [line width=1.5]    (553.49,101.5) -- (578.5,101.93) ;
\draw [shift={(582.5,102)}, rotate = 180.99] [fill={rgb, 255:red, 0; green, 0; blue, 0 }  ][line width=0.08]  [draw opacity=0] (11.61,-5.58) -- (0,0) -- (11.61,5.58) -- cycle    ;
%Shape: Rectangle [id:dp42206635347440513] 
\draw  [line width=1.5]  (581.17,73.09) -- (629.5,73.09) -- (629.5,131) -- (581.17,131) -- cycle ;
%Rounded Rect [id:dp5779728376664172] 
\draw  [fill={rgb, 255:red, 241; green, 241; blue, 241 }  ,fill opacity=1 ][dash pattern={on 1.69pt off 2.76pt}][line width=1.5]  (93,68.2) .. controls (93,51.52) and (106.52,38) .. (123.2,38) -- (213.8,38) .. controls (230.48,38) and (244,51.52) .. (244,68.2) -- (244,164.8) .. controls (244,181.48) and (230.48,195) .. (213.8,195) -- (123.2,195) .. controls (106.52,195) and (93,181.48) .. (93,164.8) -- cycle ;
%Straight Lines [id:da01758881003880053] 
\draw [line width=1.5]    (233.5,104) -- (233.5,167) -- (218.5,167) ;
%Straight Lines [id:da6308898192301104] 
\draw [line width=1.5]    (73.43,99.45) -- (112.5,99.04) ;
\draw [shift={(116.5,99)}, rotate = 539.4] [fill={rgb, 255:red, 0; green, 0; blue, 0 }  ][line width=0.08]  [draw opacity=0] (11.61,-5.58) -- (0,0) -- (11.61,5.58) -- cycle    ;
%Shape: Ellipse [id:dp9248377414166179] 
\draw  [line width=1.5]  (115.5,100.37) .. controls (115.5,92.76) and (122.17,86.6) .. (130.4,86.6) .. controls (138.63,86.6) and (145.3,92.76) .. (145.3,100.37) .. controls (145.3,107.97) and (138.63,114.14) .. (130.4,114.14) .. controls (122.17,114.14) and (115.5,107.97) .. (115.5,100.37) -- cycle ;
%Shape: Rectangle [id:dp4654420410714898] 
\draw  [line width=1.5]  (172.17,72.09) -- (210.5,72.09) -- (210.5,130) -- (172.17,130) -- cycle ;
%Straight Lines [id:da7545515630823963] 
\draw [line width=1.5]    (145.5,101) -- (168.5,101) ;
\draw [shift={(172.5,101)}, rotate = 180] [fill={rgb, 255:red, 0; green, 0; blue, 0 }  ][line width=0.08]  [draw opacity=0] (11.61,-5.58) -- (0,0) -- (11.61,5.58) -- cycle    ;
%Straight Lines [id:da5856664406079694] 
\draw  [dash pattern={on 4.5pt off 4.5pt}]  (98.97,99.23) -- (99,120) -- (171.5,121) ;
%Straight Lines [id:da08801959644931756] 
\draw [line width=1.5]    (212.5,103) -- (238.49,103.28) -- (253.5,103.06) ;
\draw [shift={(257.5,103)}, rotate = 539.1600000000001] [fill={rgb, 255:red, 0; green, 0; blue, 0 }  ][line width=0.08]  [draw opacity=0] (11.61,-5.58) -- (0,0) -- (11.61,5.58) -- cycle    ;
%Shape: Triangle [id:dp6059952685652308] 
\draw  [line width=1.5]  (168.56,168.69) -- (219.33,144.9) -- (218.5,191) -- cycle ;
%Straight Lines [id:da5380243100476392] 
\draw    (166.5,137) -- (218.78,62.46) ;
\draw [shift={(220.5,60)}, rotate = 485.04] [fill={rgb, 255:red, 0; green, 0; blue, 0 }  ][line width=0.08]  [draw opacity=0] (8.93,-4.29) -- (0,0) -- (8.93,4.29) -- cycle    ;
%Straight Lines [id:da9784260148823003] 
\draw [line width=1.5]    (168.56,168.69) -- (130,168) -- (130.37,118.14) ;
\draw [shift={(130.4,114.14)}, rotate = 450.42] [fill={rgb, 255:red, 0; green, 0; blue, 0 }  ][line width=0.08]  [draw opacity=0] (11.61,-5.58) -- (0,0) -- (11.61,5.58) -- cycle    ;
%Shape: Rectangle [id:dp38759240884074575] 
\draw  [line width=1.5]  (256.5,40) -- (421,40) -- (421,166) -- (256.5,166) -- cycle ;
%Straight Lines [id:da6292916216971346] 
\draw [line width=1.5]    (422.49,103.5) -- (447.5,103.93) ;
\draw [shift={(451.5,104)}, rotate = 180.99] [fill={rgb, 255:red, 0; green, 0; blue, 0 }  ][line width=0.08]  [draw opacity=0] (11.61,-5.58) -- (0,0) -- (11.61,5.58) -- cycle    ;
% Text Node
\draw (502,99) node  [font=\large,xscale=1.2,yscale=1.2]  {$k_{p}\left( 1+\dfrac{\omega _{i}}{s}\right)$};
% Text Node
\draw (57.53,94.61) node  [font=\large,xscale=1.2,yscale=1.2]  {$-$};
% Text Node
\draw (19.84,82.76) node  [font=\large,xscale=1.5,yscale=1.5]  {$\boldsymbol{r}$};
% Text Node
\draw (651.62,82.76) node  [font=\large,xscale=1.5,yscale=1.5]  {$\boldsymbol{y}$};
% Text Node
\draw (81.35,83.76) node  [font=\large,xscale=1.5,yscale=1.5]  {$\boldsymbol{e}$};
% Text Node
\draw (605.34,102.04) node  [font=\large,xscale=1.2,yscale=1.2]  {$G( s)$};
% Text Node
\draw (131.4,98.37) node  [font=\large,xscale=1.2,yscale=1.2]  {$-$};
% Text Node
\draw (191.34,101.04) node  [font=\large,xscale=1.2,yscale=1.2]  {$\dfrac{1}{s}$};
% Text Node
\draw (202,163) node  [font=\large,xscale=1.2,yscale=1.2]  {$\omega _{r}$};
% Text Node
\draw (229,52) node  [font=\large,xscale=1.2,yscale=1.2]  {$\rho $};
% Text Node
\draw (338.75,103) node  [font=\large,xscale=1.2,yscale=1.2]  {$\dfrac{\left(\dfrac{s}{\omega _{d}} +1\right)\left(\dfrac{s}{\omega _{l}} +1\right)}{\left(\dfrac{s}{\omega _{t}} +1\right)\left(\dfrac{s}{\omega _{f}} +1\right)}$};
% Text Node
\draw (503,48) node  [font=\large,xscale=1.2,yscale=1.2] [align=left] {{\fontfamily{ptm}\selectfont {\large \textbf{PI}}}};
% Text Node
\draw (605,58) node  [font=\large,xscale=1.2,yscale=1.2] [align=left] {{\fontfamily{ptm}\selectfont \textbf{{\large Plant}}}};
% Text Node
\draw (341,21) node  [font=\large,xscale=1.2,yscale=1.2] [align=left] {{\fontfamily{ptm}\selectfont {\large \textbf{Lead}}}};
% Text Node
\draw (168,20) node  [font=\large,xscale=1.2,yscale=1.2] [align=left] {{\fontfamily{ptm}\selectfont {\large \textbf{FORE}}}};
\end{tikzpicture}}
\caption{The system with control $C_{g_1}$}
\label{F-N11}
\end{subfigure}
\begin{subfigure}{\columnwidth}
    \centering
    \resizebox{\textwidth}{!}{
  \tikzset{every picture/.style={line width=0.75pt}} %set default line width to 0.75pt        

\begin{tikzpicture}[x=0.75pt,y=0.75pt,yscale=-1,xscale=1]
%uncomment if require: \path (0,218); %set diagram left start at 0, and has height of 218

%Shape: Rectangle [id:dp3558696855554735] 
\draw  [line width=1.5]  (451.5,65) -- (552.5,65) -- (552.5,133) -- (451.5,133) -- cycle ;
%Shape: Ellipse [id:dp605355155007869] 
\draw  [line width=1.5]  (42.63,96.45) .. controls (42.63,88.84) and (49.3,82.68) .. (57.53,82.68) .. controls (65.76,82.68) and (72.43,88.84) .. (72.43,96.45) .. controls (72.43,104.06) and (65.76,110.22) .. (57.53,110.22) .. controls (49.3,110.22) and (42.63,104.06) .. (42.63,96.45) -- cycle ;
%Straight Lines [id:da8035539331450979] 
\draw [line width=1.5]    (642,102) -- (642,210) -- (57,210) -- (57.51,114.22) ;
\draw [shift={(57.53,110.22)}, rotate = 450.3] [fill={rgb, 255:red, 0; green, 0; blue, 0 }  ][line width=0.08]  [draw opacity=0] (11.61,-5.58) -- (0,0) -- (11.61,5.58) -- cycle    ;
%Straight Lines [id:da10064598351922416] 
\draw [line width=1.5]    (-0.5,98) -- (38.63,98.41) ;
\draw [shift={(42.63,98.45)}, rotate = 180.6] [fill={rgb, 255:red, 0; green, 0; blue, 0 }  ][line width=0.08]  [draw opacity=0] (11.61,-5.58) -- (0,0) -- (11.61,5.58) -- cycle    ;
%Straight Lines [id:da053830610313541416] 
\draw [line width=1.5]    (629.77,101.37) -- (667.5,101.94) ;
\draw [shift={(671.5,102)}, rotate = 180.87] [fill={rgb, 255:red, 0; green, 0; blue, 0 }  ][line width=0.08]  [draw opacity=0] (11.61,-5.58) -- (0,0) -- (11.61,5.58) -- cycle    ;
%Straight Lines [id:da3228273022349576] 
\draw [line width=1.5]    (553.49,101.5) -- (578.5,101.93) ;
\draw [shift={(582.5,102)}, rotate = 180.99] [fill={rgb, 255:red, 0; green, 0; blue, 0 }  ][line width=0.08]  [draw opacity=0] (11.61,-5.58) -- (0,0) -- (11.61,5.58) -- cycle    ;
%Shape: Rectangle [id:dp42206635347440513] 
\draw  [line width=1.5]  (581.17,73.09) -- (629.5,73.09) -- (629.5,131) -- (581.17,131) -- cycle ;
%Rounded Rect [id:dp5779728376664172] 
\draw  [fill={rgb, 255:red, 241; green, 241; blue, 241 }  ,fill opacity=1 ][dash pattern={on 1.69pt off 2.76pt}][line width=1.5]  (284,60.9) .. controls (284,44.39) and (297.39,31) .. (313.9,31) -- (403.6,31) .. controls (420.11,31) and (433.5,44.39) .. (433.5,60.9) -- (433.5,158.1) .. controls (433.5,174.61) and (420.11,188) .. (403.6,188) -- (313.9,188) .. controls (297.39,188) and (284,174.61) .. (284,158.1) -- cycle ;
%Straight Lines [id:da01758881003880053] 
\draw [line width=1.5]    (424.5,97) -- (424.5,160) -- (409.5,160) ;
%Straight Lines [id:da6308898192301104] 
\draw [line width=1.5]    (264.43,92.45) -- (303.5,92.04) ;
\draw [shift={(307.5,92)}, rotate = 539.4] [fill={rgb, 255:red, 0; green, 0; blue, 0 }  ][line width=0.08]  [draw opacity=0] (11.61,-5.58) -- (0,0) -- (11.61,5.58) -- cycle    ;
%Shape: Ellipse [id:dp9248377414166179] 
\draw  [line width=1.5]  (306.5,93.37) .. controls (306.5,85.76) and (313.17,79.6) .. (321.4,79.6) .. controls (329.63,79.6) and (336.3,85.76) .. (336.3,93.37) .. controls (336.3,100.97) and (329.63,107.14) .. (321.4,107.14) .. controls (313.17,107.14) and (306.5,100.97) .. (306.5,93.37) -- cycle ;
%Shape: Rectangle [id:dp4654420410714898] 
\draw  [line width=1.5]  (363.17,65.09) -- (401.5,65.09) -- (401.5,123) -- (363.17,123) -- cycle ;
%Straight Lines [id:da7545515630823963] 
\draw [line width=1.5]    (336.5,94) -- (359.5,94) ;
\draw [shift={(363.5,94)}, rotate = 180] [fill={rgb, 255:red, 0; green, 0; blue, 0 }  ][line width=0.08]  [draw opacity=0] (11.61,-5.58) -- (0,0) -- (11.61,5.58) -- cycle    ;
%Straight Lines [id:da5856664406079694] 
\draw  [dash pattern={on 4.5pt off 4.5pt}]  (289.97,92.23) -- (290,113) -- (362.5,114) ;
%Straight Lines [id:da08801959644931756] 
\draw [line width=1.5]    (403.5,96) -- (429.49,96.28) -- (448.5,96.05) ;
\draw [shift={(452.5,96)}, rotate = 539.31] [fill={rgb, 255:red, 0; green, 0; blue, 0 }  ][line width=0.08]  [draw opacity=0] (11.61,-5.58) -- (0,0) -- (11.61,5.58) -- cycle    ;
%Shape: Triangle [id:dp6059952685652308] 
\draw  [line width=1.5]  (359.56,161.69) -- (410.33,137.9) -- (409.5,184) -- cycle ;
%Straight Lines [id:da5380243100476392] 
\draw    (357.5,130) -- (409.78,55.46) ;
\draw [shift={(411.5,53)}, rotate = 485.04] [fill={rgb, 255:red, 0; green, 0; blue, 0 }  ][line width=0.08]  [draw opacity=0] (8.93,-4.29) -- (0,0) -- (8.93,4.29) -- cycle    ;
%Straight Lines [id:da9784260148823003] 
\draw [line width=1.5]    (359.56,161.69) -- (321,161) -- (321.37,111.14) ;
\draw [shift={(321.4,107.14)}, rotate = 450.42] [fill={rgb, 255:red, 0; green, 0; blue, 0 }  ][line width=0.08]  [draw opacity=0] (11.61,-5.58) -- (0,0) -- (11.61,5.58) -- cycle    ;
%Shape: Rectangle [id:dp7854456275343874] 
\draw  [line width=1.5]  (100.5,34) -- (265,34) -- (265,160) -- (100.5,160) -- cycle ;
%Straight Lines [id:da16932584244408477] 
\draw [line width=1.5]    (72.43,96.45) -- (97.44,96.88) ;
\draw [shift={(101.44,96.95)}, rotate = 180.99] [fill={rgb, 255:red, 0; green, 0; blue, 0 }  ][line width=0.08]  [draw opacity=0] (11.61,-5.58) -- (0,0) -- (11.61,5.58) -- cycle    ;

% Text Node
\draw (502,99) node  [font=\large,xscale=1.2,yscale=1.2]  {$k_{p}\left( 1+\dfrac{\omega _{i}}{s}\right)$};
% Text Node
\draw (57.53,94.61) node  [font=\large,xscale=1.2,yscale=1.2]  {$-$};
% Text Node
\draw (19.84,82.76) node  [font=\large,xscale=1.5,yscale=1.5]  {$\boldsymbol{r}$};
% Text Node
\draw (651.62,77.76) node  [font=\large,xscale=1.5,yscale=1.5]  {$\boldsymbol{y}$};
% Text Node
\draw (81.35,83.76) node  [font=\large,xscale=1.5,yscale=1.5]  {$\boldsymbol{e}$};
% Text Node
\draw (605.34,102.04) node  [font=\large,xscale=1.2,yscale=1.2]  {$G( s)$};
% Text Node
\draw (322.4,91.37) node  [font=\large,xscale=1.2,yscale=1.2]  {$-$};
% Text Node
\draw (382.34,94.04) node  [font=\large,xscale=1.2,yscale=1.2]  {$\dfrac{1}{s}$};
% Text Node
\draw (393,156) node  [font=\large,xscale=1.2,yscale=1.2]  {$\omega _{r}$};
% Text Node
\draw (420,45) node  [font=\large,xscale=1.2,yscale=1.2]  {$\rho $};
% Text Node
\draw (503,48) node  [font=\large,xscale=1.2,yscale=1.2] [align=left] {{\fontfamily{ptm}\selectfont {\large \textbf{PI}}}};
% Text Node
\draw (605,58) node  [font=\large,xscale=1.2,yscale=1.2] [align=left] {{\fontfamily{ptm}\selectfont \textbf{{\large Plant}}}};
% Text Node
\draw (359,13) node  [font=\large,xscale=1.2,yscale=1.2] [align=left] {{\fontfamily{ptm}\selectfont {\large \textbf{FORE}}}};
% Text Node
\draw (182.75,97) node  [font=\large,xscale=1.2,yscale=1.2]  {$\dfrac{\left(\dfrac{s}{\omega _{d}} +1\right)\left(\dfrac{s}{\omega _{l}} +1\right)}{\left(\dfrac{s}{\omega _{t}} +1\right)\left(\dfrac{s}{\omega _{f}} +1\right)}$};
% Text Node
\draw (185,19) node  [font=\large,xscale=1.2,yscale=1.2] [align=left] {{\fontfamily{ptm}\selectfont {\large \textbf{Lead}}}};
\end{tikzpicture}}
\caption{The system with control $C_{g_2}$}
\label{F-N12}
\end{subfigure}
\caption{Block diagrams of the Spyder plant with controllers $C_{g_1}$ (top), and $C_{g_2}$(bottom)}
\label{F-N1}
\end{figure}
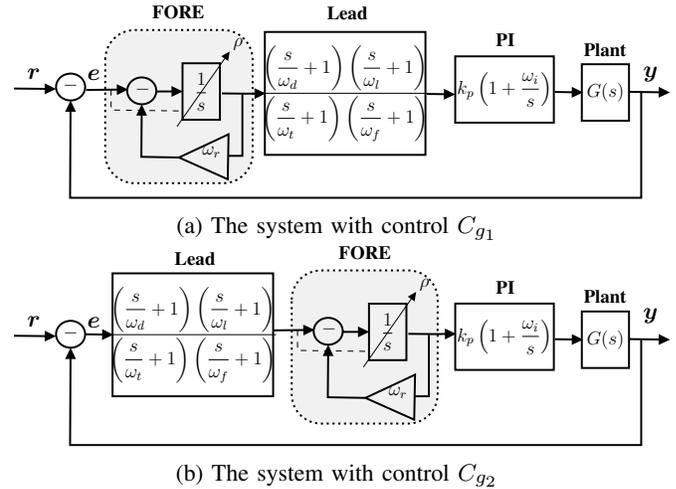
\begin{figure}[!t]
	\centering
	\includegraphics[width=0.52\hsize]{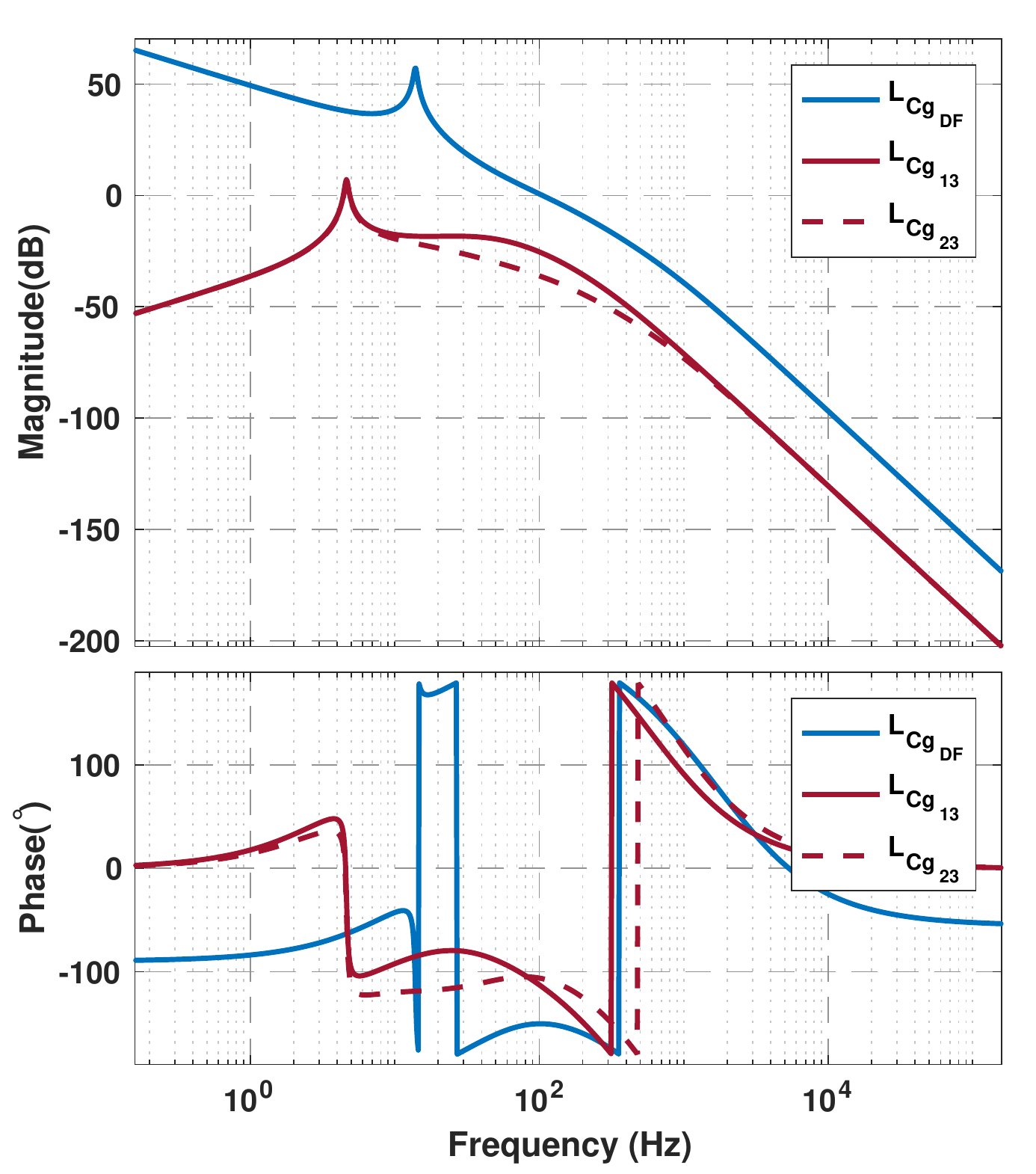}
	\caption{The DFs and the amplitudes of the third harmonics of the open-loop system with the controllers $C_{g_1}$ and $C_{g_2}$}
	\label{F-11}
\end{figure}
\begin{figure}[!t]
	\centering
	\begin{subfigure}[t]{0.48\columnwidth}
		\centering
		\includegraphics[width=\hsize]{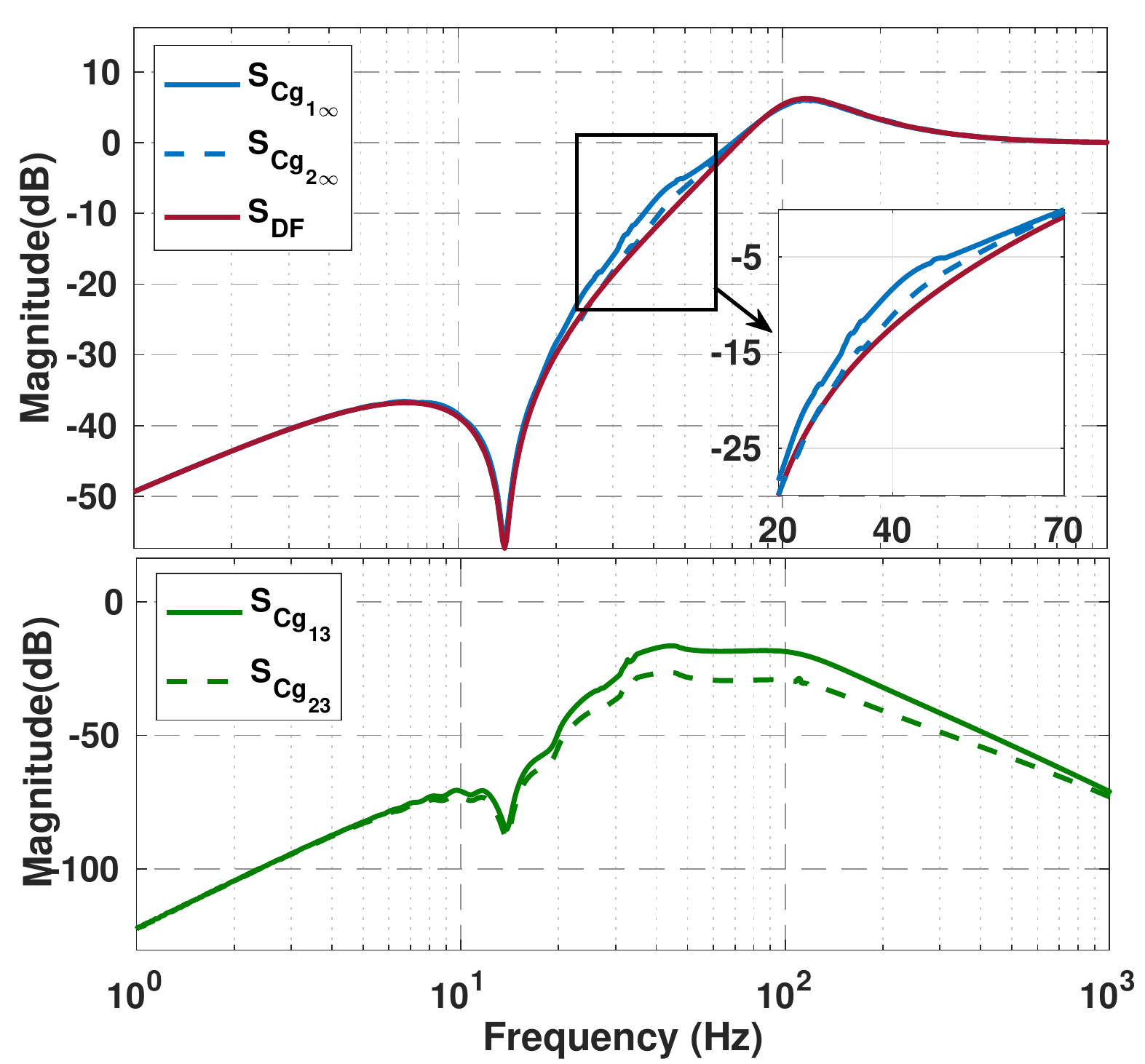}
		\caption{Sensitivity}
		\label{F-131}
	\end{subfigure}
	\hfil
	\begin{subfigure}[t]{0.48\columnwidth}
		\centering
		\includegraphics[width=\hsize]{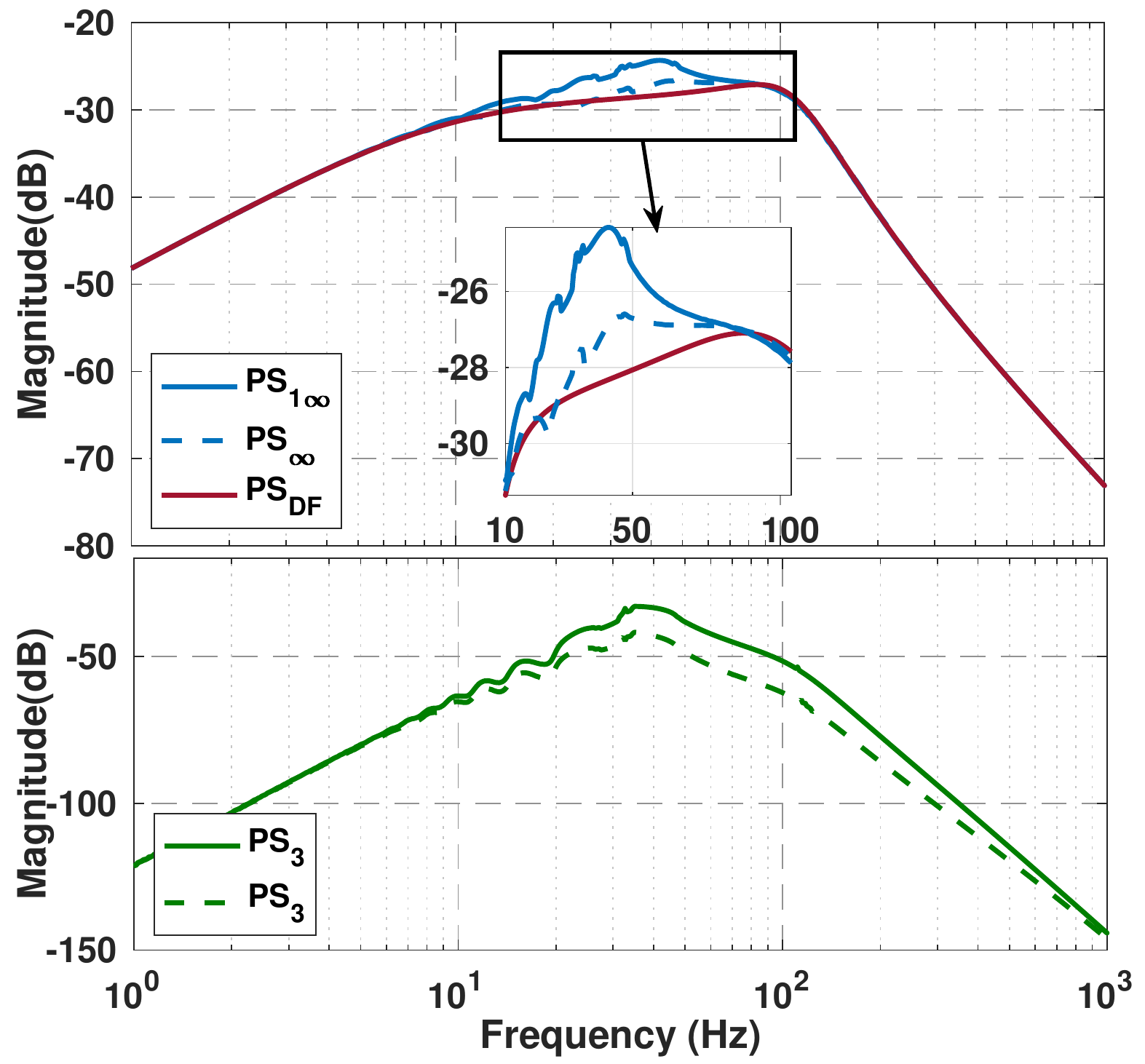}
		\caption{Process sensitivity}
		\label{F-132}
	\end{subfigure}
	\begin{subfigure}[t]{0.48\columnwidth}
		\centering
		\includegraphics[width=\hsize]{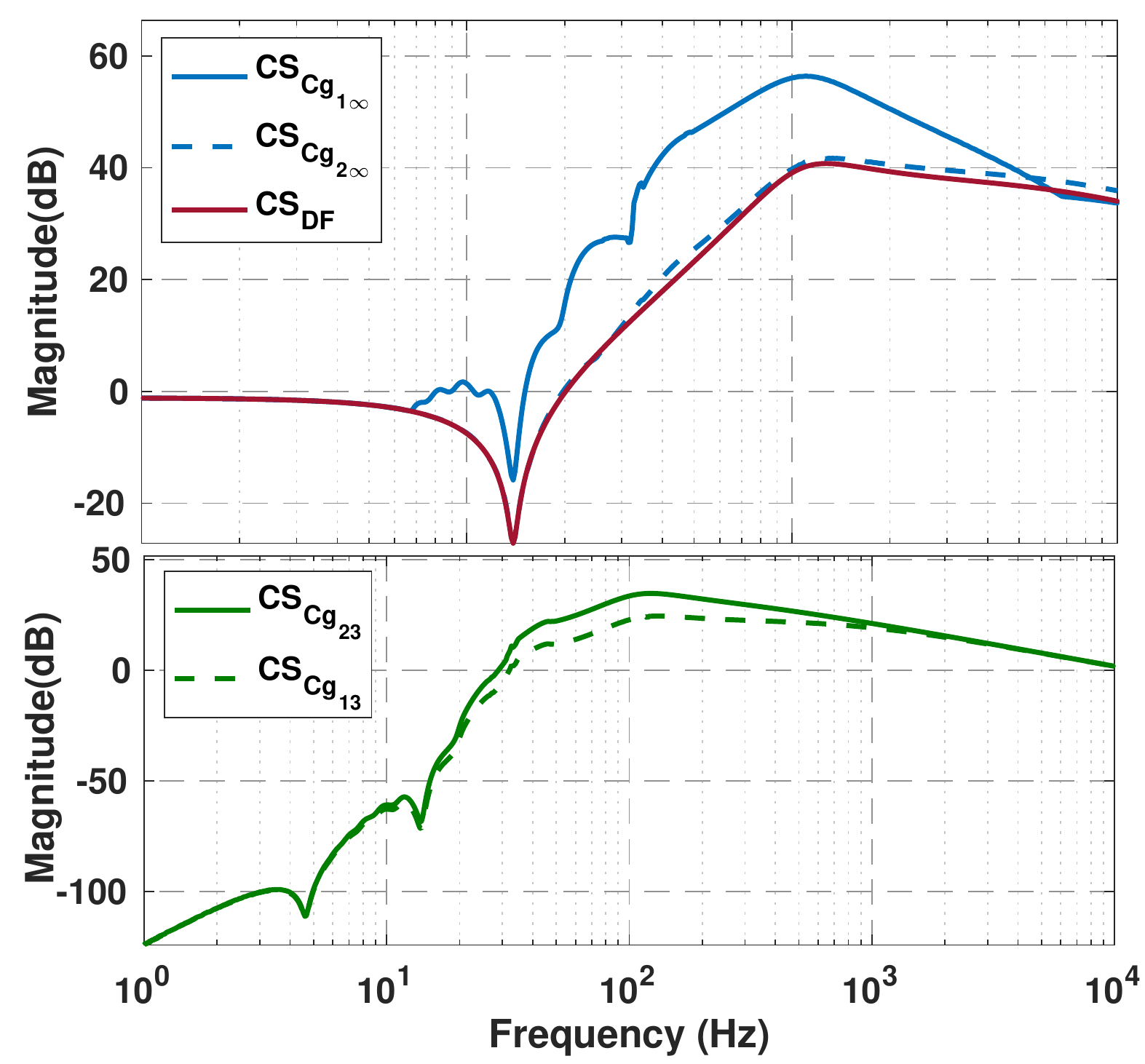}
		\caption{Control sensitivity}
		\label{F-133}
	\end{subfigure}
	\hfil
	\begin{subfigure}[t]{0.48\columnwidth}
		\centering
		\includegraphics[width=\hsize]{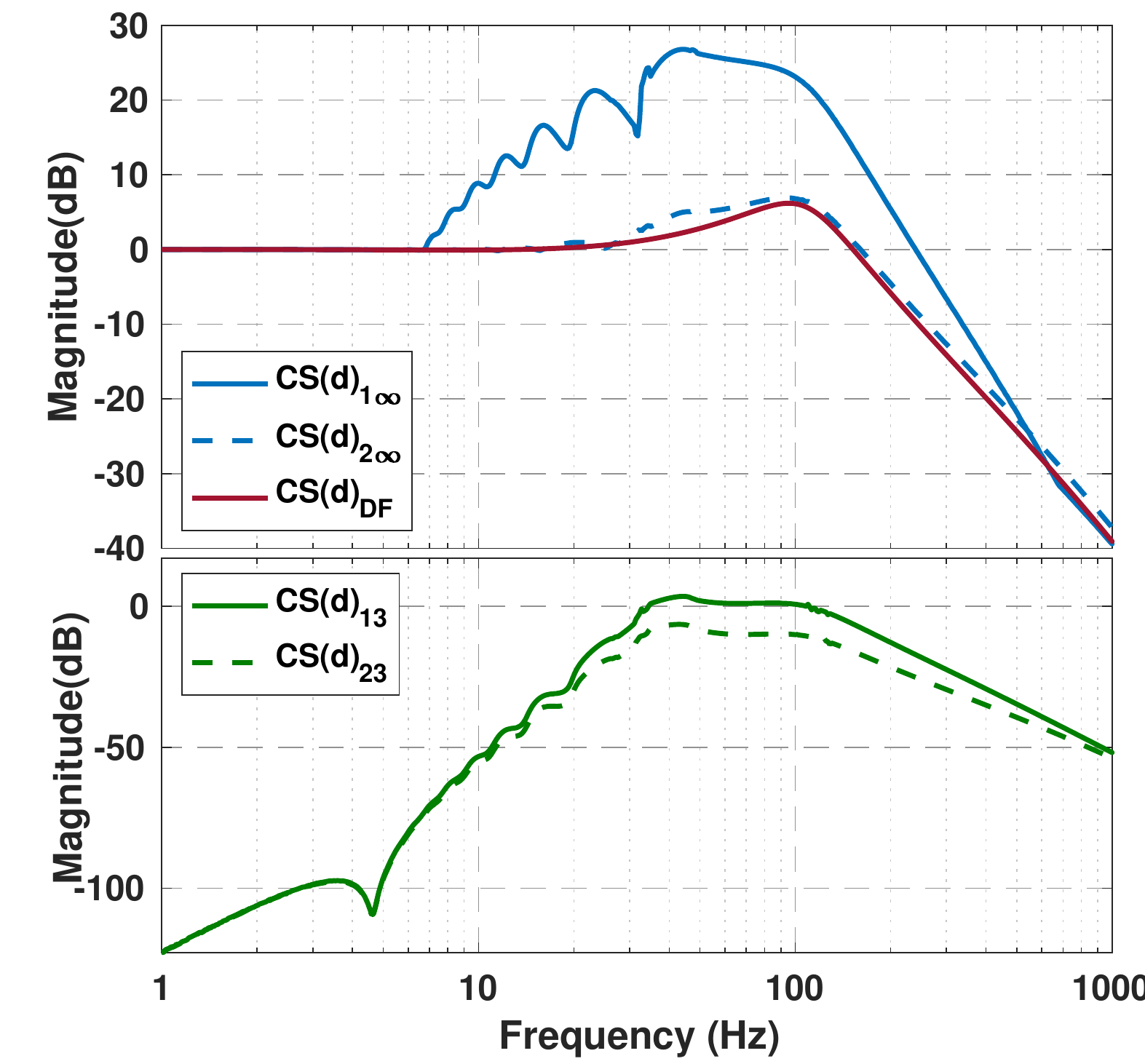}
		\caption{Control sensitivity to disturbance}
		\label{F-134}
	\end{subfigure}
	~ %add desired spacing between images, e. g. ~, \quad, \qquad, \hfill etc. 
	\caption{The DFs $(.\_$ DF), amplitudes of the third harmonics of the sensitivities (.$\_3$), and amplitudes of pseudo-sensitivities (.$\_\ \infty$) of the closed-loop system with the controllers $C_{g_1}$ and $C_{g_2}$}
	\label{F-13}
\end{figure}
In Fig.~\ref{F-13} the closed-loop frequency responses of the system with both controllers, including the amplitudes of the third harmonics, the DFs and amplitudes of pseudo-sensitivities, are presented. Note that there are significant differences between the results obtained using the DF method and the proposed tools. Unlike the DF method, the proposed tools reveal the effects of the ``position" of the control filters on the performance of the reset control systems. The differences in magnitude and phase of the high order harmonics of the open-loop system with these controllers (Fig.~\ref{F-11}) leads to discrepancies between the closed-loop frequency responses. As shown in Fig.~\ref{F-131}, the system with the controller $C_{g_2}$ has better tracking performance than that of the system with the controller $C_{g_1}$. In addition, the amplitude of the third harmonic of the sensitivity of the system with $C_{g_2}$ is smaller than that resulting from the use of the controller $C_{g_1}$ around the cross-over frequency. Moreover, as illustrated in Fig.~\ref{F-132}, the system with the controller $C_{g_2}$ has better disturbance rejection capability than the system with the controller $C_{g_1}$. As shown in Fig.~\ref{F-133} and Fig.~\ref{F-134}, the system with the controller $C_{g_1}$ has larger control input in comparison with the system with the controller $C_{g_2}$. As discussed, unlike the case of linear controller and the results obtained using the DF method, the control sensitivity due to the disturbance $CS_{d_{\infty}}$ is different from the complementary sensitivity, particularly at middle frequencies.
\newline In addition, the tracking error and the error due to the presence of disturbance of the system with the controller $C_{g_1}$ at 5 Hz are obtained experimentally (Table \ref{T-02}). As was shown, there are negligible differences between experimental and our proposed results. These small differences between the theoretical and the experimental results are due to quantization, digitalization of the controller, numerical approximations, and the presence of noise.
    \begin{table}
	\centering
	\caption{Comparison between the theoretical and experimental results in terms of tracking performance and disturbance rejection}
	\label{T-02}
	%\resizebox{0.5\textwidth}{!}{%
		\begin{tabular}{|c|c|c|}
			\hline
			\multirow{2}{*}{\textbf{Performance}}        & \multicolumn{2}{c|}{$C_{g_1}$}\\ \cline{2-3} 
			& Theory            & Experiment\\ \hline
			Tracking $\dfrac{e_r}{|r|}$(dB)              & -37.57            & -35.8  \\ \hline
			Disturbance rejection $\dfrac{e_w}{|w|}$(dB) & -33.1             & -34    \\ \hline
		\end{tabular}%
	%}
\end{table}

In summary the proposed methods allow predicting the closed-loop performance of reset control systems more accurately than the DF method. In addition, it reveals important features of reset controllers which are not exposed by the DF method.    
%%%%%%%%%conclusion
\section{Conclusion}\label{sec:6}
This paper has proposed an analytical approach to obtain closed-loop frequency responses for reset control systems, including high order harmonics. To this end, sufficient conditions for the existence of the steady-state solution of the closed-loop reset control systems driven by periodic inputs have been presented. Moreover, pseudo-sensitivities, which serve as a graphical tool for performance analysis of reset controllers, have been defined: these relate the error and control input of the system to the reference and the disturbance. All calculations can be performed in a user-friendly toolbox to make this approach easy of use. To show the effectiveness of the proposed method, the performances of a high-precision positioning stage with reset controllers have been assessed using the DF method and our proposed method. The results confirm that the proposed method predicts the closed-loop performance of reset control systems more accurately than the DF method.
%%%%%%%%%%%%%%%%%%%%%%%%%%%%%%%%%%%%%%%%%%%%Appendix
\appendices
\section{}
\begin{lemma}\label{AP1}
Consider the linear systems $\dot{x}_{p_1}(t)=A_px_{p_1}(t)+B_pu(t)$, $y_{p_1}(t)=C_px_{p_1}(t)$, with $x_{p_1}(0)=x_0$, $\dot{x}_{p_2}(t)=A_px_{p_2}(t)+B_pu(t)$, $y_{p_2}(t)=C_px_{p_2}(t)+W_I(t)$, with $x_{p_2}(0)=0$, and $\dot{z}(t)=A_pz(t)$, $W_I(t)=C_pz(t)$, with $z(0)=x_0$, in which $A_p$, $B_p$, and $C_p$ describe a realization of transfer function $P(s)$ (see Fig.~\ref{last}). Then $y_{p_1}(t)=y_{p_2}(t)$, for all $t\geq0$.   
\end{lemma}
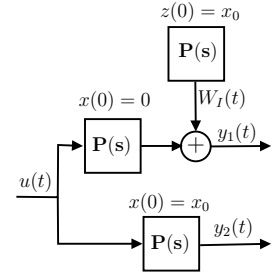
\begin{figure}
	\centering
\resizebox{0.4\hsize}{!}{	
\tikzset{every picture/.style={line width=0.75pt}} %set default line width to 0.75pt        
\begin{tikzpicture}[x=0.75pt,y=0.75pt,yscale=-1,xscale=1]
%uncomment if require: \path (0,299); %set diagram left start at 0, and has height of 299
%Straight Lines [id:da21224497609988013] 
\draw [line width=1.5]    (1.5,208) -- (44.63,208.45) -- (45,154) -- (67,154.85) ;
\draw [shift={(71,155)}, rotate = 182.2] [fill={rgb, 255:red, 0; green, 0; blue, 0 }  ][line width=0.08]  [draw opacity=0] (11.61,-5.58) -- (0,0) -- (11.61,5.58) -- cycle    ;
%Straight Lines [id:da1342683265261655] 
\draw [line width=1.5]    (44.63,208.45) -- (45,257) -- (129,257) ;
\draw [shift={(133,257)}, rotate = 180] [fill={rgb, 255:red, 0; green, 0; blue, 0 }  ][line width=0.08]  [draw opacity=0] (11.61,-5.58) -- (0,0) -- (11.61,5.58) -- cycle    ;
%Straight Lines [id:da7115202272243253] 
\draw [line width=1.5]    (204,155.5) -- (263,155.03) ;
\draw [shift={(267,155)}, rotate = 539.55] [fill={rgb, 255:red, 0; green, 0; blue, 0 }  ][line width=0.08]  [draw opacity=0] (11.61,-5.58) -- (0,0) -- (11.61,5.58) -- cycle    ;
%Straight Lines [id:da4039663983630011] 
\draw [line width=1.5]    (131.5,155) -- (171,155) ;
\draw [shift={(175,155)}, rotate = 180] [fill={rgb, 255:red, 0; green, 0; blue, 0 }  ][line width=0.08]  [draw opacity=0] (11.61,-5.58) -- (0,0) -- (11.61,5.58) -- cycle    ;
%Shape: Rectangle [id:dp39755663376477024] 
\draw  [line width=1.5]  (73,127) -- (131,127) -- (131,183) -- (73,183) -- cycle ;
%Shape: Rectangle [id:dp9163225278060194] 
\draw  [line width=1.5]  (160,31) -- (218,31) -- (218,87) -- (160,87) -- cycle ;
%Shape: Circle [id:dp017483216419601266] 
\draw  [line width=1.5]  (173,155.5) .. controls (173,146.94) and (179.94,140) .. (188.5,140) .. controls (197.06,140) and (204,146.94) .. (204,155.5) .. controls (204,164.06) and (197.06,171) .. (188.5,171) .. controls (179.94,171) and (173,164.06) .. (173,155.5) -- cycle ;
%Straight Lines [id:da2075124407193829] 
\draw [line width=1.5]    (188,88) -- (188.46,136) ;
\draw [shift={(188.5,140)}, rotate = 269.45] [fill={rgb, 255:red, 0; green, 0; blue, 0 }  ][line width=0.08]  [draw opacity=0] (11.61,-5.58) -- (0,0) -- (11.61,5.58) -- cycle    ;
%Shape: Rectangle [id:dp8826546796767498] 
\draw  [line width=1.5]  (134,229) -- (192,229) -- (192,285) -- (134,285) -- cycle ;
%Straight Lines [id:da3644719670354364] 
\draw [line width=1.5]    (193,258.5) -- (263,257.08) ;
\draw [shift={(267,257)}, rotate = 538.8399999999999] [fill={rgb, 255:red, 0; green, 0; blue, 0 }  ][line width=0.08]  [draw opacity=0] (11.61,-5.58) -- (0,0) -- (11.61,5.58) -- cycle    ;
% Text Node
\draw (193,15) node  [font=\large,xscale=1.3,yscale=1.3]  {$z( 0) =x_{0}$};
% Text Node
\draw (21.84,192.76) node  [font=\large,xscale=1.3,yscale=1.3]  {$u( t)$};
% Text Node
\draw (101,155) node  [font=\large,xscale=1.3,yscale=1.3]  {$\mathbf{P( s)}$};
% Text Node
\draw (231.84,141.76) node  [font=\large,xscale=1.3,yscale=1.3]  {$y_{1}( t)$};
% Text Node
\draw (227.84,240.76) node  [font=\large,xscale=1.3,yscale=1.3]  {$y_{2}( t)$};
% Text Node
\draw (191,57) node  [font=\large,xscale=1.3,yscale=1.3]  {$\mathbf{P( s)}$};
% Text Node
\draw (188.5,153.5) node  [font=\Large,xscale=1.3,yscale=1.3]  {$+$};
% Text Node
\draw (216,106) node  [font=\large,xscale=1.3,yscale=1.3]  {$W_{I}(t)$};
% Text Node
\draw (104,108) node  [font=\large,xscale=1.3,yscale=1.3]  {$x( 0) =0$};
% Text Node
\draw (162,257) node  [font=\large,xscale=1.3,yscale=1.3]  {$\mathbf{P( s)}$};
% Text Node
\draw (161,215) node  [font=\large,xscale=1.3,yscale=1.3]  {$x( 0) =x_{0}$};
\end{tikzpicture}}
	\caption{Diagram of the result in Lemma \ref{AP1}}
	\label{last}
\end{figure} 
\begin{proof}
Note that $W_I(t)=C_pe^{A_pt}x_0$. Thus, $$y_{p_2}(t)=y_{p_1}(t)=C_pe^{A_p(t)}x_0+\displaystyle\int_{0}^{t}e^{A_p(t-\tau)}B_pu(\tau)d\tau.$$
\end{proof}\section{}
\begin{lemma}\label{AP2}
Consider a positive and bounded function $V(t)$. Suppose that there exists a $\alpha>0$ such that
\begin{equation}\label{E-AP201} 
\begin{cases} 
\dot{V}\leq-\alpha V & t\in\mathcal{M},\\
V(\Delta x(t^+))=V(\Delta x(t))+\Xi(t,\delta), & t\notin\mathcal{M}.\\
\end{cases}
\end{equation}
If for $t$ sufficiently large
\begin{equation}\label{E-AP2011} 
\Xi(t,\delta)\leq0,
\end{equation}
then there exist $\alpha_m>0$ and $\mathcal{K}>0$ such that
\begin{equation}\label{E-AP202} 
V(t)\leq \mathcal{K}e^{-\alpha_m t},\text{ for all }t\geq0.
\end{equation}
\end{lemma}
\begin{proof}
Since $V$ is bounded, by (\ref{E-AP201}) and (\ref{E-AP2011}), $V$ achieves its maximum value at some time $t_{v_m}<\infty$. In other words, there exists a time $0\leq t_{v_m}<\infty$ such that
\begin{equation}\label{E-AP203} 
\begin{cases}
V(t_{v_m})\geq V(t), & t\leq t_{v_m},\\
V(t_{v_m})> V(t).      & t> t_{v_m},
\end{cases}
\end{equation}
Therefore, by (\ref{E-AP2011}) and well-posedness property, there exists a bounded set $\mathcal{T}=\{t_i>t_{v_m}|\ t_i\notin\mathcal{M}\land\Xi(t_i,\delta)>0,\ i\in\mathbb{N}\}$. Thus, using (\ref{E-AP203}) there exists a bounded set $\mathcal{A}=\{\alpha_i>0|\ V(t_{i})=e^{-\alpha_i(t_i-t_{v_m})}V(t_{v_m}),\ t_i\in\mathcal{T}\}$. Since the set $\mathcal{A}$ is bounded, there exists a $\alpha^\prime>0$ such that for all $\alpha_i\in\mathcal{A}$ one has that $\alpha^\prime\leq\alpha_i$. Now considering $\alpha_m=\min (\alpha,\alpha^\prime)$, based on (\ref{E-AP201}) and (\ref{E-AP2011}), yields
\begin{equation}\label{E-AP204} 
V(t)\leq e^{-\alpha_m (t-t_{v_m})}V(t_{v_m})=\mathcal{K}e^{-\alpha_m t},\text{ for all }t\geq0.
\end{equation}
Finally, if $\mathcal{T}$ and $\mathcal{A}$ are empty sets, then selecting $\alpha_m=\alpha$ the claim yields. 
\end{proof}
\section{}\label{APC}
Note that by~(\ref{E-314}) and (\ref{E-327}), $\displaystyle\lim_{\omega \to \infty} \bar{y}(t)=0$ and $\displaystyle\lim_{\omega \to \infty} T_n(j\omega)=0$. In addition, since $y(t)=r_0\displaystyle\sum_{n=1}^{\infty}T_n(j\omega)$, $\displaystyle\lim_{\omega \to \infty}\sum_{n=1}^{\infty} T_n(j\omega)=0$. On the other hand, if the transfer function of the controller $C_{\mathfrak{L}_1}$ is proper, then $\displaystyle\lim_{\omega \to \infty}C_{\mathfrak{L}_1}(nj\omega)=K_{c1}$; otherwise, $\displaystyle\lim_{\omega \to \infty}(nj\omega)^{n_c}C_{\mathfrak{L}_1}(nj\omega)=1$, with $n_c\geq1$. In the case in which $C_{\mathfrak{L}_1}$ is proper. 
\begin{equation}\label{E-APC01}
\begin{array}{*{35}{c}}
\displaystyle\lim_{\omega \to \infty}\dfrac{\underset{t_s\leq t\leq t_{s+q}}{\max}|e_R(t)-e_{R_1}(t)|}{\underset{t_s\leq t\leq t_{s+q}}{\max}|e_{R_1}(t)|}\leq\lim_{\omega \to \infty}\dfrac{K_{c1}r_0|\displaystyle\sum_{n=3}^{\infty}T_n(j\omega)|}{K_{c1}r_0}\\
\Rightarrow\displaystyle\lim_{\omega \to \infty}\dfrac{\underset{t_s\leq t\leq t_{s+q}}{\max}|e_R(t)-e_{R_1}(t)|}{\underset{t_s\leq t\leq t_{s+q}}{\max}|e_{R_1}(t)|}=0.
\end{array}
\end{equation} 
In the case in which $C_{\mathfrak{L}_1}$ is strictly proper. 
\begin{equation}\label{E-APC02}
\begin{array}{*{35}{c}}
\displaystyle\lim_{\omega \to \infty}\dfrac{\underset{t_s\leq t\leq t_{s+q}}{\max}|e_R(t)-e_{R_1}(t)|}{\underset{t_s\leq t\leq t_{s+q}}{\max}|e_{R_1}(t)|}<\lim_{\omega \to \infty}\dfrac{\dfrac{r_0}{\omega^{n_c}}|\displaystyle\sum_{n=3}^{\infty}\dfrac{T_n(j\omega)}{n^{n_c}}|}{\dfrac{r_0}{\omega^{n_c}}}\\
\Rightarrow\displaystyle\lim_{\omega \to \infty}\dfrac{\underset{t_s\leq t\leq t_{s+q}}{\max}|e_R(t)-e_{R_1}(t)|}{\underset{t_s\leq t\leq t_{s+q}}{\max}|e_{R_1}(t)|}=0.
\end{array}
\end{equation} 
%%%%%%%%%%%%%%%%%%%%%%%%%%%%%%%%%%%%%%%%%%%%%Acknowledgment
\section*{Acknowledgment}
The authors would like to thank Prof. Alfonso Ba$\tilde{\text{n}}$os  from University of Murcia, Spain, for his kind guidance and for the useful discussions.

\bibliographystyle{IEEEtran}
\bibliography{phd}

\begin{IEEEbiography}[{\includegraphics[width=1in,height=1.25in,clip,keepaspectratio]{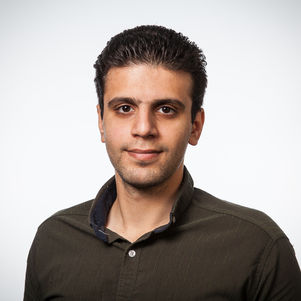}}]{Ali Ahmadi Dastjerdi}
received his master degree in mechanical engineering from Sharif University of Technology, Iran, in 2015. He is currently working as
a PhD candidate at the department of precision and microsystem engineering, TU Delft, The Netherlands. He has also collaborated with Prof. Alessandro Astolfi since 2019 as a sabbatical leave in Imperial College University. His primary research interests are on mechatronic systems design, precision engineering, precision motion control, and nonlinear control.
\end{IEEEbiography}
\begin{IEEEbiography}[{\includegraphics[width=1in,height=1.25in,clip,keepaspectratio]{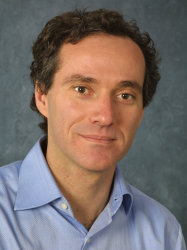}}]{Alessandro Astolfi}
was born in Rome, Italy, in 1967. He graduated in electrical engineering from the University of Rome in 1991. In 1992 he joined ETH-Zurich where he obtained a M.Sc. in Information Theory in 1995 and the Ph.D. degree with Medal of Honor in 1995 with a thesis on discontinuous stabilization of nonholonomic systems. In 1996 he was awarded a Ph.D. from the University of Rome La Sapienza for his work on nonlinear robust control. Since 1996 he has been with the Electrical and Electronic Engineering Department of Imperial College London, London (UK), where he is currently Professor of Nonlinear Control Theory and Head of the Control and Power Group. From 1998 to 2003 he was also an Associate Professor at the Dept. of Electronics and Information of the Politecnico of Milano. Since 2005 he has also been a Professor at Dipartimento di Ingegneria Civile e Ingegneria Informatica, University of Rome Tor Vergata. He has been a visiting lecturer in Nonlinear Control in several universities, including ETH-Zurich (1995-1996); Terza Univer- sity of Rome (1996); Rice University, Houston (1999); Kepler University, Linz (2000); SUPELEC, Paris (2001), Northeastern University (2013). His research interests are focused on mathematical control theory and control applications, with special emphasis for the problems of discontinuous stabilization, robust and adaptive control, observer design and model reduction. He is the author of more than 150 journal papers, of 30 book chapters and of over 240 papers in refereed conference proceedings. He is the author (with D. Karagiannis and R. Ortega) of the monograph Nonlinear and Adaptive Control with Applications (Springer-Verlag).
He is the recipient of the IEEE CSS A. Ruberti Young Researcher Prize (2007), the IEEE RAS Googol Best New Application Paper Award (2009), the IEEE CSS George S. Axelby Outstanding Paper Award (2012), the Automatica Best Paper Award (2017). He is a Distinguished Member of the IEEE CSS, IET Fellow, IEEE Fellow and IFAC Fellow. He served as Associate Editor for Automatica, Systems and Control Letters, the IEEE Trans. on Automatic Control, the International Journal of Control, the European Journal of Control and the Journal of the Franklin Institute; as Area Editor for the Int. J. of Adaptive Control and Signal Processing; as Senior Editor for the IEEE Trans. on Automatic Control; and as Editor-in-Chief for the European Journal of Control. He is currently Editor-in-Chief of the IEEE Trans. on Automatic Control. He served as Chair of the IEEE CSS Conference Editorial Board (2010- 2017) and in the IPC of several international conferences. He has been/is a Member of the IEEE Fellow Committee (2016, 2018/2019).
\end{IEEEbiography}
\begin{IEEEbiography}[{\includegraphics[width=1in,height=1.25in,clip,keepaspectratio]{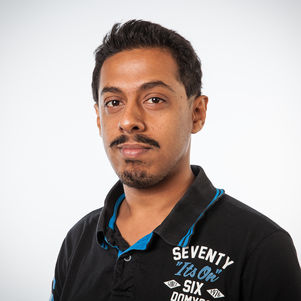}}]{Niranjan Saikumar}
received his PhD degree in electrical engineering from Indian Institute of Science, India in 2015. He is currently working as a postdoc at the department of precision and microsystem engineering, TU Delft, The Netherlands. His research interests are on precision motion control, and nonlinear precision control and mechatronic system with distributed actuation.
\end{IEEEbiography}
\begin{IEEEbiography}[{\includegraphics[width=1in,height=1.25in,clip,keepaspectratio]{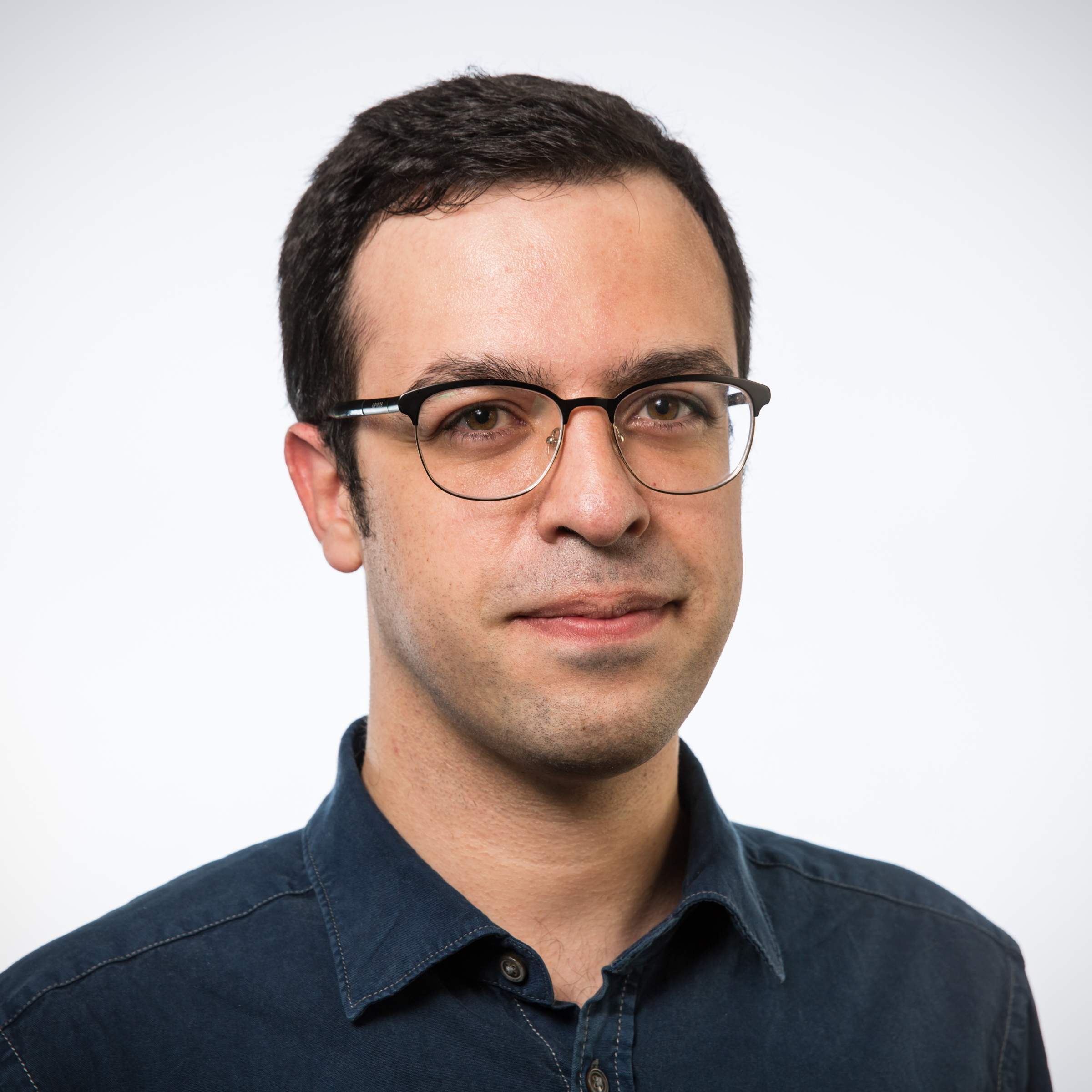}}]{Nima Karbasizadeh}
has received his M.Sc. degree in Mechatronics from University of Tehran, Iran in 2017. He is currently a PhD candidate at the department of precision and microsystem engineering, Delft University of Technology, the Netherlands. His research interests are precision motion control, nonlinear precision control, mechatronic system design and haptics.
\end{IEEEbiography}
\begin{IEEEbiography}[{\includegraphics[width=1in,height=1.25in,clip,keepaspectratio]{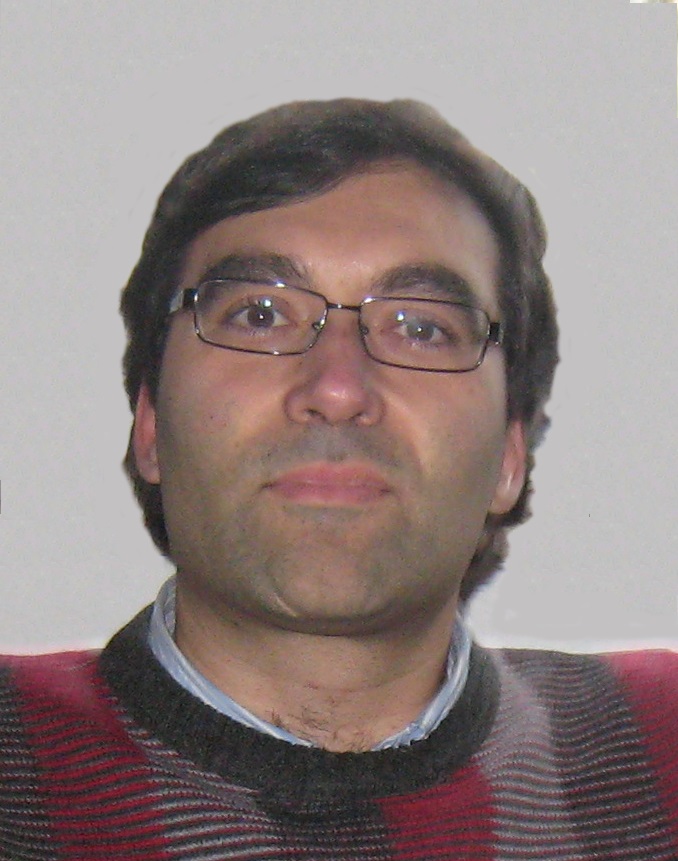}}]{Duarte Val\'erio}
is Associate Professor at Instituto Superior T\'ecnico University of Lisbon, where he got his MSc (2001) and PhD (2005) in Mechanical Engineering, with theses on fractional control, i.e.\ on the use of fractional (non-integer) order derivatives in control. He has worked with fractional control and fractional dynamic systems, and their applications in several areas, ever since. He also researches in the fields of energy conversion (in particular, the control of Wave Energy Converters, that produce electricity from the energy of sea waves) and bioengineering applications (modelling and control of dynamic systems such as biological processes).
He has co-authored over forty papers in journals with impact factors, three books, over sixty papers in conference proceedings, and nine book chapters.
\end{IEEEbiography}
\begin{IEEEbiography}[{\includegraphics[width=1in,height=1.25in,clip,keepaspectratio]{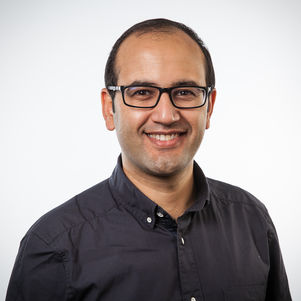}}]{S. Hassan HosseinNia}
received his PhD degree with honour "cum laude" in electrical engineering specializing in automatic control: application in mechatronics, form the University of Extremadura, Spain in 2013. His main research interests are in precision mechatronic system design, precision motion control and mechatronic system with distributed actuation and sensing. He has an industrial background working at ABB, Sweden. Since October 2014 he is appointed as an assistant professor at the department of precession and microsystem engineering at TU Delft, The Netherlands. He is an associate editor of the international journal of advanced robotic systems
and Journal of Mathematical Problems in Engineering.
\end{IEEEbiography}
\end{document}